\documentclass[final]{tlp}

\usepackage{aopmath}
\usepackage{multirow}
\usepackage{rotating}
\usepackage{fancyvrb}
\usepackage{url}
\usepackage{psfrag}
\usepackage{alltt}
\usepackage{amssymb}

\submitted{29 January 2008}
\revised{1 March 2009}
\accepted{2 April 2009}

\newcommand{\ogt}{\symbol{123}}
\newcommand{\cgt}{\symbol{125}}

\newtheorem{definition}{Definition} 

\newcommand{\hc}{\mbox{hasChild}}
\newcommand{\clever}{\mbox{Clever}}
\newcommand{\s}{\mbox{Successful}}
\newcommand{\h}{\mbox{Happy}}

\newdimen\boxfigwidth 

\def\magicbox{\begingroup
  \boxfigwidth=\hsize
  \advance\boxfigwidth by -2\fboxrule
  \advance\boxfigwidth by -2\fboxsep
  \setbox4=\vbox\bgroup\hsize\boxfigwidth
  \hrule height0pt width\boxfigwidth\smallskip
  \linewidth=\boxfigwidth
}
\def\endmagicbox{\smallskip\egroup\fbox{\box4}\endgroup}

\begin{document}
\bibliographystyle{acmtrans}
\title[Description Logic Reasoning in Prolog]{Efficient Description Logic Reasoning in Prolog:\\The DLog system}

\author[Gergely Luk\'{a}csy and P\'{e}ter Szeredi]
{GERGELY LUK\'{A}CSY and P\'{E}TER SZEREDI \\
Department of Computer Science and Information Theory\\
Budapest University of Technology and Economics\\
Budapest, Magyar tud\'{o}sok k\"{o}r\'{u}tja 2. H-1117, Hungary \\
E-mail: \{lukacsy,szeredi\}@cs.bme.hu
}

\pubyear{2009}
\maketitle

\begin{abstract}
Traditional algorithms for description logic (DL) instance retrieval
are inefficient for large amo\-unts of underlying data. As description
logic is becoming more and more popular in areas such as the Semantic Web and
information integration, it is very important to have systems which can
reason efficiently over large data sets.

In this paper we present an approach to transform description logic
axioms, formalised in the $\mathcal{SHIQ}$ DL language, into a Prolog
program under the Unique Name Assumption. This transformation is
performed with no knowledge about particular individuals: they are
accessed dynamically during the normal Prolog execution of the
generated program. This technique, together with the top-down Prolog
execution, implies that only those pieces of data are accessed which
are indeed important for answering the query. This makes it possible to store
the individuals in a database instead of memory, which results in
better scalability and helps using description logic ontologies
directly on top of existing information sources.

The transformation process consists of two steps: (1) the DL axioms are
converted to first-order clauses of a restricted form, (2) a Prolog program
is generated from these clauses. Step (2), which is the focus of the
present paper, actually works on more general clauses than those obtainable
by applying step (1) to a $\mathcal{SHIQ}$ knowledge base.

We first present a base transformation, the output of which can either be
executed using a simple interpreter, or further extended to executable
Prolog code. We then discuss several optimisation techniques, applicable to
the output of the base transformation. Some of these techniques are
specific to our approach, while others are general enough to be interesting
for description logic reasoner implementors not using Prolog.

We give an overview of \emph{DLog}, a DL reasoner in Prolog,
which is an implementation of the techniques outlined above.  We evaluate
the performance of DLog and compare it to some widely used description
logic reasoners, such as RacerPro, Pellet, and KAON2.
\end{abstract}

\section{Introduction}

Description Logics (DLs) are becoming widespread thanks to the recent
trend of using semantics in various systems and applications. As an
example, in the Semantic Web idea, semantics is captured in the form
of expressive ontologies, described in the OWL Web Ontology Language 
\cite{owlspec} which is intended to be the standard knowledge
representation format of the Web. The OWL DL fragment of this language
is mostly based on the $\mathcal{SHIQ}$ DL language.
Other application fields of description logics include natural
language processing \cite{885762}, medical systems \cite{SGHB02b},
information integration \cite{calvanese98description} and complex
engineering and computer technology systems \cite{Configuration}.

Similarly to \cite{motik06PhD}, the motivation for our work comes from
the realisation that description logics are, or soon will be used over
large amounts of data. In an information integration system, for
example, huge amounts of data are stored in external
databases. On the Web, as another example, we already have tremendous
amounts of meta-information which will significantly increase as the
Semantic Web vision becomes more and more tangible. Obviously, these
information sources cannot be stored directly in memory.

Thus, we are interested in querying description logic concepts where the
actual data set~-- the so called ABox~-- is bigger than the available
computer memory. We found that most existing description logic reasoners
are not suitable for this task, as these are not capable of handling ABoxes
stored externally, e.g.\ in databases. This is not a technical problem:
most existing algorithms for querying description logic concepts need to
examine the whole ABox to answer a query which results in scalability
problems and undermines the point of using databases. Because of this, we
started to investigate techniques which allow the separation of the
inference algorithm from the data storage.

We have developed a solution, where the inference algorithm is divided
into two phases. First we create a \emph{query-plan}, in the form of a
Prolog program, from the actual
DL knowledge base, without any knowledge of the content of the
underlying data set.  Subsequently, this query-plan can be run on real
data, to obtain the required results. 

Naturally, the quality of the query-plan greatly affects the
performance of the execution. We have applied several optimisations to
make the generated Prolog program more efficient. These ideas are
incorporated in the reference implementation system called DLog,
available at \texttt{http://dlog-reasoner.sourceforge.net}. 

From the Description Logic point of view, DLog is an ABox reasoning engine
for the full
$\mathcal{SHIQ}$ language. It deals with number restrictions as well
as with all other modelling constructs present in
$\mathcal{SHIQ}$. DLog maintains the Unique Name Assumption and
assumes that the ABox is consistent (see
Section~\ref{general} for more details).

The paper is structured as follows. Section~\ref{preliminaries} discusses
the background of the paper,
introducing the field of Description Logic and giving a summary of theorem proving approaches
for DLs. In Section~\ref{shiqreasoning} we start with two motivating examples
to demonstrate the non-trivial nature of the translation of
description logic axioms to Prolog. We then present a complete, but
inefficient solution for generating Prolog programs from
$\mathcal{SHIQ}$ knowledge bases. Section~\ref{generation} discusses
several optimisation schemes which significantly increase the
efficiency of execution. Section~\ref{dlog} presents the architecture
and the implementation details of the DLog system. In
Section~\ref{evaluation} we analyse the performance of DLog,
comparing it with other reasoning systems. Finally, in
Section~\ref{future} and \ref{conclusion}, we conclude with the discussion of future
work and the summary of our results.

\section{Background and related work}
\label{preliminaries}

In this section we first provide a brief introduction to Description
Logics, then we give an overview of traditional, tableau-based DL
reasoning approaches. Next, we discuss how resolution can be used for
DL inference, and summarise related work on using Logic Programming
for Description Logic reasoning, including our earlier
contributions. Finally, we present the Prolog Technology Theorem
Proving approach, the techniques of which are used extensively
throughout the paper.

\subsection{Description Logics}
\label{subsection:dl}

Description Logics (DLs) \cite{dlhb} is a family of simple logic
languages used for knowledge representation.  DLs are used for
describing various kinds of knowledge of a specific field as well as
of general nature. The description logic approach uses \emph{concepts}
to represent sets of \emph{objects}, and \emph{roles} to describe
binary relations between concepts.  Objects are the instances
occurring in the modelled application field, and thus are also called
\emph{instances} or \emph{individuals}.

A description logic knowledge base $\textit{KB}$ is a set of DL axioms
consisting of two disjoint parts: the \emph{TBox} and the
\emph{ABox}. These are sometimes referred to as
$\textit{KB}_\mathcal{T}$ and $\textit{KB}_\mathcal{A}$. The TBox
(terminology box), in its simplest form, contains terminology axioms
of the form $C \sqsubseteq D$ (concept $C$ is subsumed by concept $D$). The
ABox (assertion box) stores knowledge about the 
individuals in the world: a concept assertion of the form $C(i)$ denotes
that the \emph{individual name} $i$ is an instance of the concept $C$, while a
role assertion $R(i,j)$ means that individual names $i$ and $j$ are
related through role $R$. Usually one assumes that two different
individual names denote two different individuals (this is the so
called unique name assumption, or simply UNA).

Note the difference between ``individual names'' and
``individuals''. The former are syntactic elements of the DL language,
while the latter are the elements of the modelled domain. To make the
paper easier to read we will sometimes use the phrase ``individual''
instead of ``individual name'', assuming that the context makes it
clear that a syntactic element is being referred to.

Concepts and roles may either be \emph{atomic} (referred to by a
concept name or a role name) or \emph{composite}.  A composite concept
is built from atomic concepts using \emph{constructors}.  The
expressiveness of a DL language depends on the constructors allowed
for building composite concepts or roles. Obviously there is a
trade-off between expressiveness and the complexity of inference.

We use the DL language $\mathcal{SHIQ}$ in this paper. Here, concepts
(denoted by $C$ and $D$) are built from roles (denoted by $R$ and
$S$), atomic concepts, the top and bottom concepts ($\top$ and $\bot$)
using the following constructors: intersection ($C \sqcap D$), union
($C \sqcup D$), negation ($\neg{C}$), value restriction
($\forall{R\ldotp C}$), existential restriction ($\exists{R\ldotp C}$)
and qualified number restrictions ($\geqslant n\,R\ldotp C$ and
$\leqslant n\,R\ldotp C$). The only role constructor in
$\mathcal{SHIQ}$ is the inverse operator, thus roles can take the form
$R_A$ or $R_{A}^{-}$, where $R_A$ is an atomic role. 

The $\mathcal{SHIQ}$ language also allows the use of role subsumption
($R \sqsubseteq S$), role equivalence ($R \equiv S$), and transitivity
axioms ($\mathsf{Trans}(R)$). Note that a role equivalence $R \equiv
S$ can be eliminated by replacing it with the two axioms $R
\sqsubseteq S$ and $S \sqsubseteq R$.  The set of role subsumption
axioms is often called a \emph{role hierarchy}. Each $\mathcal{SHIQ}$
axiom has a straightforward translation in first-order logic (FOL).

An important sub-language of $\mathcal{SHIQ}$ is $\mathcal{ALC}$, where
number restrictions, role axioms and inverse roles are not allowed.

The basic inference tasks concerning the TBox can be reduced to
determining if a given concept $C$ is satisfiable with respect to a
given TBox. 

ABox inference tasks require both a TBox and an ABox. In this paper,
we will deal with two ABox reasoning problems: instance check and
instance retrieval.  In an \emph{instance check} problem, a
\emph{query-concept} $C$ and an individual $i$ is given.  The question
is whether $C(i)$ is entailed by the TBox and the ABox. In an
\emph{instance retrieval} problem the task is to retrieve all
individual names $i$, for which assertion $C(i)$ is entailed by the TBox
and an ABox, for a given query concept $C$.

For more details on Description Logics  we refer the reader to the first two
chapters of \cite{dlhb}. 

\subsection{Reasoning on DLs}

Several techniques have been developed for
ABox reasoning. Traditional ABox reasoning is based on the
\emph{tableau inference} algorithm, which tries to build a model
showing that a given concept assertion is satisfiable. To infer that an
individual $i$ is an instance of a concept $C$, an indirect assumption
$\neg{C}(i)$ is added to the ABox, and the tableau-algorithm is
applied. If this reports inconsistency, $i$ is proved to be an
instance of $C$. The main drawback of this approach is that it cannot
be directly used for high volume instance retrieval, because it would
require checking all instances in the ABox, one by one. 

To make tableau-based reasoning more efficient on large data sets,
several techniques have been developed in recent years, see e.g.\ 
\cite{aboxoptim}. These are used by the state-of-the-art description logic reasoners,
such as RacerPro \cite{HMSW04} or Pellet \cite{pellet}, the two
tableau reasoners used in our performance evaluation in
Section~\ref{evaluation}.

Some DL reasoners pose serious restrictions on the knowledge base
to ensure efficient execution with large amounts of instances.  For
example, \cite{instancestore} suggests a solution called the
\emph{instance store}, where the ABox is stored externally, and is
accessed in a very efficient way. The drawback is that the ABox may
contain only axioms of the form $C(a)$, i.e.\ we cannot make role
assertions.

\subsection{Resolution theorem proving for DLs}
\label{motik}

\cite{voronkov} discuss how a first-order theorem prover, such
as Vampire, can be modified and optimised for reasoning over
description logic knowledge bases. This work, however, mostly focuses
on TBox reasoning.

The paper \cite{shiqresolution} describes a resolution-based inference
algorithm which is not as sensitive to the increase of the ABox size 
as the tableau-based methods. The system KAON2 \cite{motik06PhD} is an
implementation of this approach, providing reasoning services over the
description logic language $\mathcal{SHIQ}$. In
Section~\ref{evaluation} we use KAON2 as one of the systems with which
we compare the performance of DLog.

The basic idea of KAON2 is to first transform a $\mathcal{SHIQ}$
knowledge base into a skolemized first-order clausal form. However,
instead of using direct clausification, first a structural
transformation \cite{plaisted} is applied on the
$\textit{KB}_\mathcal{T}$ axioms. This transformation 
eliminates the nested concept descriptions by introducing new
concepts; the resulting set of first-order clauses is denoted by
$\Xi(\textit{KB})$.  In the next step, basic superposition
\cite{nieuwenhuis95theorem}, a refinement of first-order resolution, is applied to saturate
$\Xi(\textit{KB}_\mathcal{T})$. The resulting set of clauses is
denoted by $\Gamma(\textit{KB}_\mathcal{T})$. Clauses
$\Gamma(\textit{KB}_\mathcal{T}) \union \Xi(\textit{KB}_\mathcal{A})$
are then transformed into a disjunctive datalog program
\cite{eiter97disjunctive} entailing the same set of ground facts as
the initial DL knowledge base. This program is executed using a
disjunctive datalog engine written specifically for KAON2. In this
approach, the saturated clauses may still contain (non-nested)
function symbols which are eliminated by introducing a new constant
$f_{i}$, standing for $f(i)$, for each individual $i$ in the ABox. This
effectively means that KAON2 has to read the \emph{whole content} of
the ABox before attempting to answer any queries.

Although the motivation and goals of KAON2 are similar to ours, unlike
KAON2 (1) we use a pure two-phase reasoning approach (i.e.\ the ABox is not
involved in the first phase) and (2) we translate into Prolog which has
well-established, efficient and robust implementations. More details are
provided in the upcoming sections.

\subsection{Description Logics and Logic Programming}
\label{dlp}

\cite{dlp} introduces the term Description Logic Programming
(DLP), advocating a direct transformation of $\mathcal{ALC}$
description logic concepts into Horn-clauses.  It poses some
restrictions on the form of the knowledge base, to disallow axioms
requiring disjunctive reasoning. As an extension, \cite{hornshiq}
introduces a fragment of the $\mathcal{SHIQ}$ language which can be
transformed into Horn-clauses. This work, however, still poses
restrictions on the use of disjunctions. In \cite{saor} and
\cite{sindice} authors present a semantic search engine which works
on web-scale and builds on the extension of the DLP idea. Further
important work on Description Logic Programming includes
\cite{samuel} and \cite{faithful}.

Another approach of utilising Logic Programming in DL reasoning was
proposed by the research group of the authors of the present
paper. Earlier results of this work have been published in several
conference papers. The first step of our research resulted in
a resolution-based transformation of ABox reasoning problems to Prolog
for the DL language $\mathcal{ALC}$ and an \emph{empty TBox}
\cite{padl06}. As the second step, we examined how ABox reasoning
services can be provided with respect to a \emph{non-empty} TBox: we
 extended our approach to allow ABox inference involving
$\mathcal{ALC}$ 
TBox axioms of a restricted form \cite{dl06}. In
\cite{semantics2006} we presented a system doing almost full
$\mathcal{ALC}$ reasoning, which uses an interpreter based on PTTP
techniques (see Section~\ref{pttp} below).

Zsolt Zombori has extended the saturation
technique of \cite{motik06PhD} so that there are no function symbols in
the resulting first-order clauses \cite{zombori}. The basic idea here
is to use a slightly modified version of the basic superposition, where the
order of certain resolution steps is changed. \cite{zombori} showed that
these modifications do not affect satisfiability and they require a finite
number of additional inference steps, compared to the ``standard'' basic
superposition.

\subsection{Prolog Technology Theorem Proving}
\label{pttp}

The Prolog Technology Theorem Prover approach (PTTP) was suggested by
Mark E.\ Stickel in the late 1980's \cite{stickel92prolog}. PTTP is a
sound and complete approach which builds a first-order theorem prover
on top of Prolog.  This means that an arbitrary set of general clauses
is transformed into a set of Horn-clauses and Prolog execution is used to
perform first-order logic reasoning. Note that PTTP does not support
first-order equality reasoning, but there are extensions of PTTP, such
as the PTTR system (Prolog Technology Term Rewriting), suitable for
this task  \cite{pttr}.

In PTTP, each first-order clause gives rise to a number of Horn-clauses, the
so-called \emph{contrapositives}. A FOL clause takes the form
$\mathop{\bigvee}_{1 \leq i \leq n}L_i$, where $L_i$ are literals (negated
or non-negated atomic predicates). This clause has $n$ contrapositives of
the form $L_k \leftarrow \neg L_1, \ldots, \neg L_{k-1}, \neg L_{k+1},
\ldots, \neg L_n$, for each $1 \leq k \leq n$. Having removed double
negations, the remaining negations are eliminated by introducing new
predicate names for negated literals. For each predicate name $P$ a new
predicate name $\textit{not}\_P$ is introduced, and all occurrences of $\neg{P}(X)$ are
replaced by $\textit{not}\_P(X)$, both in the head and in the body. The link between
the separate predicates $P$ and $\textit{not}\_P$ is created by \emph{ancestor
resolution}, see below.

Note that the use of contrapositives has the effect that each literal of a
FOL clause appears in the head of a Horn clause. This ensures that each
literal can participate in a resolution step, in spite of the restricted
selection rule of Prolog.

The PTTP approach uses ancestor resolution \cite{ancres} to support the
\emph{factoring} inference rule (the replacement of two unifiable literals
by a single most general unifier of the two literals). Ancestor resolution
is implemented in Prolog by building an \emph{ancestor list} which contains
\emph{open} predicate calls (i.e.\ calls which were entered or re-entered,
but have not been exited yet, according to the Procedure-Box model of
Prolog execution \cite{lp}).  Alternatively, an ``ancestor-of'' relation
between goals can be defined as the transitive closure of the ``parent-of''
relationship, where goal $PG$ is the \emph{parent} of the goal $G$, if $PG$
invokes a clause whose body contains $G$. The ancestor list contains all
ancestors of a given goal, usually in the newest-first order.

Ancestor resolution is an inference step checking if the ancestor list
contains a goal which can be unified with the negation of the current
goal. If this is the case, then the current goal succeeds and the
unification with the ancestor element is performed.  Note that in order to
retain completeness, as an alternative to ancestor resolution, one has to
try to prove the current goal using normal resolution, too.

There are two further features in the PTTP approach. First, to avoid
infinite loops, iterative deepening is used instead of the
standard depth-first Prolog search strategy. Second, in contrast with
most Prolog systems, PTTP uses occurs check during unification.

To sum up, PTTP uses five techniques to build a first-order theorem
prover on the top of Prolog: contrapositives, renaming of negated literals,
ancestor resolution, iterative deepening, and occurs check.

\section{DL reasoning in Prolog}
\label{shiqreasoning}

We present a pure two-phase approach to $\mathcal{SHIQ}$ ABox inference. In
the first phase, the $\mathcal{SHIQ}$ axioms are transformed to a Prolog
program. The second phase is the execution of this program. Importantly,
the ABox axioms are not modified by this transformation, and so the ABox
can be stored externally, e.g.\ in a database.

The first phase, the transformation, is itself divided into two stages.
First, the $\mathcal{SHIQ}$ axioms are converted into a set of first-order
clauses of a specific form. The second stage deals with the transformation
of FOL clauses to a Prolog program.

We first summarise some general assumptions and present two motivating
examples. Next, we give an outline of the first stage of the transformation.
Before proceeding to the second stage, we introduce the notion of \emph{DL
clause}, which is a first-order clause satisfying certain requirements.
Each clause produced by the first stage of the transformation
satisfies these requirements, but there are interesting DL clauses
which cannot be derived from a $\mathcal{SHIQ}$ KB.

The second stage takes an arbitrary set of DL clauses and transforms these
to a Prolog program. We first show how the PTTP approach can be specialised
for DL clauses, resulting in a so-called DL program. We then present a
simple interpreter for DL programs. Next, we describe how to extend DL
programs so that they can be directly executed by Prolog, thus making it
possible to compile a $\mathcal{SHIQ}$ KB to an executable Prolog
program. Finally, we show some examples of this complete transformation process.

\subsection{General considerations}
\label{general}

Throughout this paper we assume that (1) different individual names denote
different individuals (Unique Name Assumption) and (2) the ABox is consistent.

Note that in the absence of UNA one may have to perform complex
deductions to determine whether two individuals are distinct. Namely,
the individuals $i_1$ and $i_2$ can be inferred to be different if one
can find an \emph{arbitrary} concept $C$, such that both $C(i_1)$ and
$\neg C(i_2)$ hold. Thus deciding a simple inequality question
potentially requires reading the whole ABox, which makes it impossible to
perform ABox reasoning in a focused way.

Similarly, detecting the inconsistency of an ABox requires checking the
whole content of the ABox. 

As the main advantage of our approach -- the focused nature of
reasoning -- is lost in both cases, we advocate using other approaches
(e.g. tableau algorithms) for checking ABox consistency and answering
ABox queries in the absence of UNA.

We also assume that the ABox is extensionally reduced, i.e.\ beside
roles, it contains only atomic concepts and their negations. An
arbitrary knowledge base can be easily transformed to satisfy this
constraint. First, one has to replace all composite concepts in the
ABox (except for the negated atomic concepts) by new atomic
concepts. Next, one has to extend the TBox with appropriate concept
axioms, which define the newly introduced concept names to be
equivalent to the composite concept they stand for.

In Sections~\ref{shiqreasoning} and \ref{generation} we assume that no
predicate name contains the character \texttt{\_} (underline). This makes
it possible to use prefixes containing an underline (such as
\texttt{not\_}) as names of various auxiliary predicates. This 
restriction does not apply in the DLog system, discussed in Section~\ref{dlog}.

\subsection{Translating by hand: two motivating examples}\label{motivex}

Databases and the \emph{negation as failure} feature of Prolog use the 
\emph{closed world assumption} where
any object which is not known to be an instance of concept $C$ is
treated as an instance of $\neg{}C$. In contrast with this, the
\emph{open world assumption} (OWA) is used in classical logic
reasoning, and thus in DL reasoning as well.  When reasoning under
OWA, one is interested in obtaining statements which hold in all
models of the knowledge base, i.e. those entailed by the knowledge
base.

Figure~\ref{figure:iocaste_kb} shows a famous DL example about the family
of Oedipus and Iocaste, which is often used to demonstrate the difference
between open and closed world reasoning, see e.g.\ \cite{dlhb}.

\begin{figure}[htbp]
\begin{Verbatim}[numbers=left,numbersep=2pt,frame=single,commandchars=\\\{\}]
\(\exists\texttt{hasChild}\ldotp(\texttt{Patricide} \sqcap \exists\texttt{hasChild}\ldotp\neg\texttt{Patricide}) \sqsubseteq \texttt{Ans}\)

hasChild(Iocaste,Oedipus).     hasChild(Iocaste,Polyneikes).
hasChild(Oedipus,Polyneikes).  hasChild(Polyneikes,Thersandros).
Patricide(Oedipus).            \(\neg\)Patricide(Thersandros).
\end{Verbatim}
\caption{The Iocaste knowledge base.}
\label{figure:iocaste_kb}
\end{figure}

The only TBox axiom is shown in line 1, while the content of
the ABox is given in lines 3--5. The TBox axiom expresses that
somebody is considered to be an answer if she has a patricide child,
who, in turn, has a non-patricide child. The ABox axioms describe the
\texttt{hasChild} binary relation between certain individuals and also
express the facts that \texttt{Oedipus} is known to be patricide,
while \texttt{Thersandros} is known to be non-patricide (note that
both patricide and non-patricide are unary relations). Our task is to
solve the instance-check problem \texttt{Ans(Iocaste)}, i.e.\ to
decide if the given knowledge base entails the fact that \texttt{Iocaste}
belongs to the answer concept  \texttt{Ans}. 

Note that \texttt{Iocaste} can be shown to be an answer, in spite of
the fact that one cannot \emph{name} the child of \texttt{Iocaste} who
has the desired property. That is, solving this specific instance
check problem requires \emph{case analysis}: the child in question is
either \texttt{Polyneikes} or \texttt{Oedipus}, depending on
\texttt{Polyneikes} being a patricide or not.

Also note that the trivial Prolog translation of the DL knowledge base
in Figure~\ref{figure:iocaste_kb}, shown below, is not appropriate, as
the goal \texttt{:- Ans(i)} fails.

{\small
\begin{Verbatim}[numbers=left,numbersep=2pt,frame=single]
Ans(A) :-  hasChild(A, B), Patricide(B), hasChild(B, C), not_Patricide(C).

Patricide(o).   not_Patricide(t). 
hasChild(i, o). hasChild(i, p). hasChild(o, p). hasChild(p, t).
\end{Verbatim}
}

Here, to follow the standard DL notation, predicate names
corresponding to concepts start with capitals, while role names are
written in lower case. For the sake of conciseness we omit the
apostrophes around Prolog predicate names starting with capitals and
we also use the abbreviations \texttt{i}, \texttt{o}, \texttt{p}, and \texttt{t} for instance
names.

Note that using negation as failure (the \verb"\+" operator) would not
solve the problem: when the goal \texttt{not\_Patricide(C)} in line 1 is
replaced by \verb"\+ Patricide(C)", every instance not known to be
patricide is viewed as non-patricide, which is not correct. For example,
consider the ABox containing the axioms \verb+hasChild(i1, i2)+, \verb+hasChild(i2, i3)+, and
\verb+Patricide(i2)+. This ABox does not entail \texttt{Ans(i1)},
but the Prolog program using negation as failure does return success for
this query.

There is an infinite number of ABox \emph{patterns}
which allow an individual to be proven to belong to concept
\texttt{Ans} \cite{padl06}.  These patterns are shown in
Figure~\ref{figure:patterns}. Here the nodes of the pattern graph
stand for individuals, while the edges represent the \texttt{hasChild}
role instances. Furthermore, \texttt{P} and $\neg$\texttt{P} stand for
\texttt{Patricide} and \texttt{not\_Patricide}, respectively. Note
that case $n = 2$ corresponds to the ABox given in
Figure~\ref{figure:iocaste_kb}.

\begin{figure}[htbp]
  \centering
  \psfrag{n0}{$n = 1$}
  \psfrag{n1}{$n = 2$}
  \psfrag{n2}{$n = k$}
  \psfrag{a}{\texttt{i}}
  \psfrag{b}{\texttt{o}}
  \psfrag{c}{\texttt{p}}
  \psfrag{c0}{e$_1$}
  \psfrag{c1}{e$_2$}
  \psfrag{c2}{e$_k$}
  \psfrag{d}{\texttt{t}}
  \psfrag{P}{\texttt{P}}
  \psfrag{NP}{$\neg$\texttt{P}}

  \psfrag{dots}{\dots}
  \includegraphics[scale=0.57]{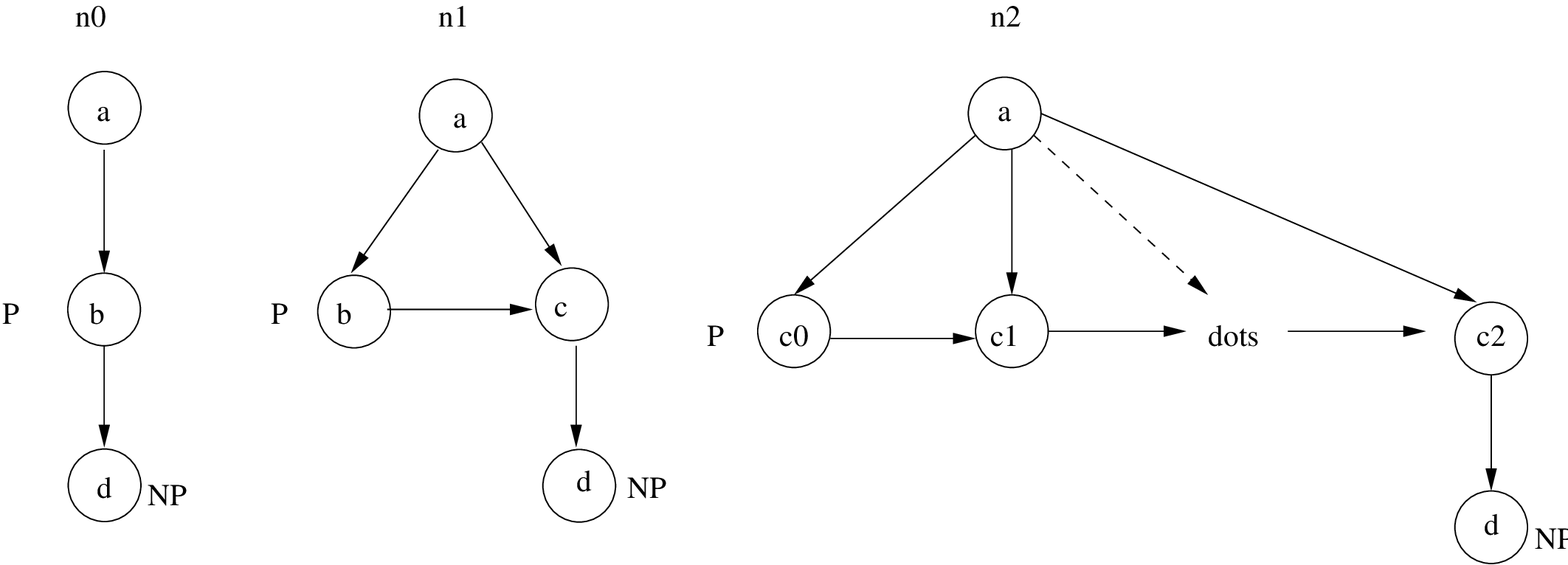}
\caption{Iocaste ABox patterns.}
\label{figure:patterns}
\end{figure}

Consider the ABox corresponding to the general case (the
rightmost pattern). We show that the individual \texttt{i} does belong to the concept
\texttt{Ans}. Assume that there is a model of this ABox in which $\neg
\texttt{Ans}(\texttt{i})$ holds. We show by induction that, for each $j = 1,
\ldots, k$, $\texttt{Patricide}(\texttt{e}_j)$ holds in this model. This is true for $j =
1$. Assume that this is true for $j = m$. Because $\texttt{e}_m$ is a patricide
child of \texttt{i}, where the latter does not belong to \texttt{Ans},
all children of $\texttt{e}_m$ have to be patricide. Thus
$\texttt{e}_{m+1}$ is a patricide, which completes the inductive
proof. Hence $\texttt{e}_k$ is a patricide child of \texttt{i}, who has a
non-patricide child \texttt{t}, and thus \texttt{i} belongs to
\texttt{Ans}. This contradicts our initial, indirect assumption, thus
proving that \texttt{i} belongs to the concept \texttt{Ans}. See paper
\cite{padl06} for the proof that the patterns of
Figure~\ref{figure:patterns} give an exact characterisation of ABoxes
entailing $\texttt{Ans}(\texttt{i})$, w.r.t.\ the TBox shown in line 1 of
Figure~\ref{figure:iocaste_kb}.

\begin{figure}[htbp]
\begin{Verbatim}[numbers=left,numbersep=2pt,frame=single,commandchars=\\\{\}]
Ans(X) :- hasChild(X,Y), hasChild(Y,Z), not_Patricide(Z), dPatricide(Y,X).

dPatricide(Z, _)    :- Patricide(Z).
dPatricide(Z, X)    :- hasChild(X, Y), hasChild(Y, Z), dPatricide(Y, X).

Patricide(o).   not_Patricide(t). 
hasChild(i, o). hasChild(i, p). hasChild(o, p). hasChild(p, t).
\end{Verbatim}
\caption{A Prolog translation of the Iocaste knowledge base.}
\label{figure:iocaste_handmade}
\end{figure}

A Prolog program, written by hand, solving the Iocaste problem is
presented in Figure~\ref{figure:iocaste_handmade}. We have shown in
\cite{padl06} that this program is a sound and complete translation of
the Iocaste problem and it captures exactly the patterns shown in
Figure~\ref{figure:patterns}. To see this, notice that
\texttt{dPatricide(Z,X)} describes patterns of the form shown in
Figure~\ref{figure:dpatricide}. The first clause of
\texttt{dPatricide(Z,X)} (line 3) corresponds to the degenerate
pattern for the case $n = 1$. The second clause (line 4) states that
a new pattern corresponding to \texttt{dPatricide(Z,X)} can be obtained by
extending a pattern corresponding to \texttt{dPatricide(Y,X)} by two new
\texttt{hasChild} edges between \texttt{(X, Y)} and \texttt{(Y, Z)}.

\begin{figure}[htbp]
  \centering
  \psfrag{n0}{$n = 1$}
  \psfrag{n1}{$n = 2$}
  \psfrag{n2}{$n = k$}
  \psfrag{a}{\texttt{X}}
  \psfrag{b}{\texttt{Z}}
  \psfrag{b1}{\texttt{Y}}
  \psfrag{c}{\texttt{Z}}
  \psfrag{c0}{\texttt{Y}$_1$}
  \psfrag{c1}{\texttt{Y}$_2$}
  \psfrag{c2}{\texttt{Z}}
  \psfrag{d}{\texttt{t}}
  \psfrag{P}{\texttt{P}}
  \psfrag{NP}{$\neg$\texttt{P}}

  \psfrag{dots}{\dots}
  \includegraphics[scale=0.57]{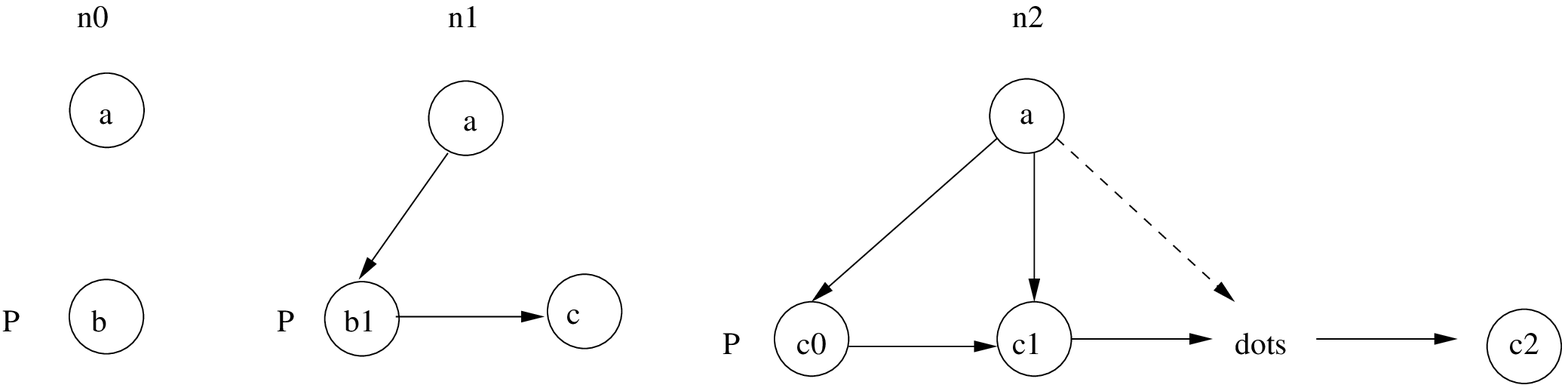}
\caption{The pattern captured by \texttt{dPatricide/2}.}
\label{figure:dpatricide}
\end{figure}

Note that the program in Figure~\ref{figure:iocaste_handmade} may not
terminate if the \texttt{hasChild} relations form a directed cycle in
the ABox. If this cannot be excluded, then termination can be ensured, for
example,  by  tabling \cite{xsb}, or \emph{loop elimination} (see
Section~\ref{specpttp}).

Unlike in the Iocaste problem, we do not always need to use case analysis
and therefore we can generate simpler programs. For example, let us
consider the DL knowledge base presented in
Figure~\ref{figure:happy_kb}. Here we consider someone happy if she
has a child who, in turn, has both a clever child and a pretty child
(line 1).

\begin{figure}[htbp]
\begin{Verbatim}[numbers=left,numbersep=2pt,frame=single,commandchars=\\\{\}]
\(\exists\texttt{hasChild}\ldotp(\exists\texttt{hasChild}\ldotp\texttt{Clever} \sqcap \exists\texttt{hasChild}\ldotp\texttt{Pretty}) \sqsubseteq \texttt{Happy}\)

Clever(lisa). Pretty(lisa). hasChild(kate,bob). hasChild(bob,lisa).
\end{Verbatim}
\caption{The Happy knowledge base.}
\label{figure:happy_kb}
\end{figure}

The ABox given in line 3, together with the TBox axiom in line 1,  implies that \texttt{kate}
is happy. In this case, there is a straightforward Prolog translation for
the TBox,
as shown in Figure~\ref{figure:happystraightforwardtranslation}.

\begin{figure}[htbp]
\begin{Verbatim}[numbers=left,numbersep=2pt,frame=single,commandchars=\\\{\}]
Happy(A) :- hasChild(A, B), hasChild(B, C), hasChild(B, D),
            Clever(C), Pretty(D).

Clever(lisa). Pretty(lisa). hasChild(kate, bob). hasChild(bob, lisa).
\end{Verbatim}
\caption{The straightforward Prolog translation of the Happy knowledge base.}
\label{figure:happystraightforwardtranslation}
\end{figure}

One of the aims in the DLog project is to create a framework where
problems not requiring case analysis result in straightforward Prolog
programs. As we show later in Section~\ref{generation}, we can actually
generate programs for the Iocaste and Happy problems which are the same as,
or very close to, the handmade programs presented here.

\subsection{Building first-order clauses from a $\mathcal{SHIQ}$ knowledge base}
\label{shiqprinciples}
\setcounter{equation}{0}

In this section we deal with the first stage of the $\mathcal{SHIQ}$ to
Prolog transformation: converting a $\mathcal{SHIQ}$ $\textit{KB}$ to a set
of first-order clauses of a specific form. The details of this
transformation are presented in \cite{zombori}, here we only give an
outline and an illustrative example.

The basic idea of this conversion is to bring forward the
inference steps that are independent of the ABox. In doing so, our aim
is not to compute all consequences of the TBox -- that would require
too much time and is not needed anyway --, but to perform those steps
that complicate the ABox reasoning. Most notably, the translation of a
DL TBox to first-order clauses involves introducing skolem functions
which require special treatment. However, the fact that the ABox
is function-free suggests that all inference steps involving
function symbols can be performed before accessing the ABox. Hence,
instead of complicating the ABox reasoning, we break the reasoning
into two parts: an ABox independent TBox transformation is performed
as the first phase, and this is followed by the actual data reasoning as
the second phase. 

In \cite{zombori} a new calculus is introduced, which extends the  work
described in \cite{motik06PhD}. This calculus, similar to basic
superposition, is shown to
be sound, complete, and terminating for any input derived from a $\mathcal{SHIQ}$ knowledge
base. For any proof within the calculus, we can order the inference steps
in such a way that all steps involving function symbols precede all steps
involving clauses derived from the ABox. In the first stage of the
reasoning we perform the steps that do not require the ABox. The clauses
containing function symbols cannot play any role afterwards, thus we can
simply remove them. The second stage -- which is the focus
of the present paper -- makes use of the function-free nature of the
clauses, when transforming these to a Prolog program.

Note that, as opposed to \cite{motik06PhD}, all clauses containing function
symbols are eliminated in the DL to Prolog transformation. This forms the
basis of a pure two-phase reasoning framework, which allows us to store the
content of the ABox in an external database.

For an arbitrary $\mathcal{SHIQ}$ knowledge base $\textit{KB}$, let us
denote by $\textit{DL}(\textit{KB})$ the set of first-order clauses
resulting from the first stage of the transformation.  In the rest of the
paper we only make use of the fact that $\textit{DL}(\textit{KB})$ contains
clauses of a specific form, as listed in Figure~\ref{figure:dlclauses}.

\begin{figure}[htbp]
  \begin{magicbox}
\begin{enumerate}
\item[(1)]
$\neg R(x,y) \vee S(y,x)$\vspace*{0.1cm}
\item[(2)]
$\neg R(x,y) \vee S(x,y)$\vspace*{0.1cm}
\item[(3)]
$\mathbf{C}(x)$\vspace*{0.1cm}
\item[(4)]
$\mathop{\bigvee}_{i,j,k} \neg R\sb{k}(x\sb{i},x\sb{j}) \vee \mathop{\bigvee}_{i} \textbf{C}(x\sb{i}) \vee \mathop{\bigvee}_{i,j} (x\sb{i}=x\sb{j})$\vspace*{0.1cm}
\item[(5)]
$R(a,b)$\vspace*{0.1cm}
\item[(6)]
$C(a)$
\end{enumerate}
\end{magicbox}
\caption{The structure of $\textit{DL}(\textit{KB})$.}
\label{figure:dlclauses}
\end{figure}

Here clauses (1)--(4) are derived from $\textit{KB}_\mathcal{T}$, the
TBox part of the knowledge base, while clauses (5)--(6) are derived from
the ABox. As for first-order clauses, all variable symbols
appearing in (1)--(4) are universally quantified.  $R$ and $S$ denote
binary predicate names, which correspond to roles.  $C$ is a possibly
negated unary predicate name, corresponding to a concept.  Symbols $a$
and $b$ are constants. $\mathbf{C}(x)$ denotes a nonempty disjunction
of (positive or negative) unary literals, all having the variable $x$ as
their argument: $\mathbf{C}(x) = (\neg)C_1(x) \vee \ldots
\vee (\neg)C_n(x)$, $n \geq 1$. 

Clause (4) requires further
explanation, as it is known to satisfy certain constraints. First, it
contains at least one binary literal, at least one unary literal, and
a possibly empty set of variable equalities. Second, its binary
literals contain all the variables of the clause. Third, if we build a
graph from the binary literals by converting $\neg R(x,y)$ to an edge
$x \rightarrow y$, then the graph obtained this way will always be a tree.

We illustrate the transformation of the TBox with a small
example. Although the axioms are first translated to first-order
clauses and the reasoning is performed on this form, we will give the
DL equivalents of the transformed clauses, to make the example more
compact. Let us consider the following TBox: 
\begin{eqnarray}
 \label{1} \top & \sqsubseteq & (\leqslant 1\, \hc\ldotp\s) \\
 \label{2} \top & \sqsubseteq & (\geqslant 1\, \hc\ldotp\clever) \\
 \label{3} \clever & \sqsubseteq & \s \\
 \label{4} (\geqslant 2\, \hc\ldotp\top) & \sqsubseteq & \h
\end{eqnarray}

The transformation of \cite{zombori} will effectuate the following
three changes in the TBox:
\begin{itemize}
\item We know that everybody has a clever child (\ref{2}), who is also
  successful (\ref{3}). But, since there can only be at most one successful
  child (\ref{1}), it is
  impossible for a child to be successful and not clever. Accordingly,
  we will deduce the following axiom (more precisely, we deduce the
  first-order clause corresponding to this axiom):
\begin{equation}
\label{5} \top \sqsubseteq (\forall \hc\ldotp(\clever
\sqcup \lnot \s))
\end{equation}
\item How can a person turn out to be happy? If she has two
  children. But we already know that everyone has at least one clever
  child. So if she happens to have a non clever child, then this child
  cannot be identical to the clever one, so they are really two
  distinct children, hence the person is happy. Thus the following axiom is
  deduced:
\begin{equation}
\label{6} (\exists \hc\ldotp \lnot \clever) \sqsubseteq
\h
\end{equation}
\item If we translate these 6 axioms to first-order clauses, (\ref{2})
  is the only one that will give rise to skolem functions (skolem
  functions are derived from $\geq$-concepts on the right side of
  $\sqsubseteq$ and from $\leq$-concepts on the left side of
  $\sqsubseteq$). But we only need (\ref{2}) to deduce (\ref{5}) and
  (\ref{6}). Once this is done, we can dispose of (\ref{2}).
\end{itemize}
The following 5 axioms are thus produced as the output of the first stage:
\begin{eqnarray*}
\top &\sqsubseteq& (\leqslant 1\, \hc\ldotp\s) \\
\clever &\sqsubseteq& \s \\
(\geqslant 2\, \hc\ldotp\top) &\sqsubseteq& \h \\
\top &\sqsubseteq& (\forall \hc\ldotp(\clever
\sqcup \lnot \s)) \\
(\exists \hc\ldotp \lnot \clever) &\sqsubseteq&
\h
\end{eqnarray*}
The corresponding first-order clauses (where the variables are all
universally quantified)  are  the following:

\begin{displaymath}
 \hspace*{-0.77cm}
 \lnot \hc(x,y_1) \lor \lnot \hc(x,y_2) \lor \lnot \s(y_1) \lor \lnot \s(y_2) \lor y_1=y_2 
\end{displaymath}
\begin{displaymath}
 \hspace*{-0.77cm}
 \lnot \clever(x) \lor \s(x)
\end{displaymath}
\begin{displaymath}
 \hspace*{-0.77cm}
 \lnot \hc(x,y_1) \lor \lnot \hc(x,y_2) \lor \h(x) \lor y_1=y_2
\end{displaymath}
\begin{displaymath}
 \hspace*{-0.77cm}
 \lnot \hc(x,y) \lor \clever(y) \lor \lnot \s(y)
\end{displaymath}
\begin{displaymath}
 \hspace*{-0.77cm}
 \lnot \hc(x,y) \lor \clever(y) \lor \h(x)
\end{displaymath}

\noindent Note that these clauses are indeed of the form listed in
Figure~\ref{figure:dlclauses}. As the calculus used here is
shown to be complete and sound in  \cite{zombori}, we know that no further TBox clauses
need to be inferred and that the omission of clause (\ref{2}) does not
invalidate any ABox inferences.

An important feature of the first stage is that it eliminates
transitivity axioms by introducing auxiliary unary predicates,
following the technique described in \cite{motik06PhD}. 

Finally, a minor technical remark: the clauses produced from a
$\mathcal{SHIQ}$ knowledge base may contain binary literals
corresponding to inverse roles. We avoid the need for constructing
inverse role names by the following transformation: the predicate
${R_A}^{-}(X, Y)$ is replaced by $R_A(Y, X)$, where $R_A$ is an atomic
role.

\subsection{DL clauses}
\label{dlclauses}

In the remaining part of this paper we focus on how to transform
clauses of the form shown in Figure~\ref{figure:dlclauses} into efficient
Prolog code. However, we note that for the general transformation,
discussed in the present section, we
use only certain properties of the clauses. These properties are
satisfied by a subset of first-order clauses which is, in fact, larger
than the set of clauses that can be generated from a $\mathcal{SHIQ}$
KB. These properties are summarised in the following definition.

\begin{definition}[DL-clauses]
\label{def:dlike}
A first-order clause $C$  is
said to be a \emph{DL clause} if it satisfies the following
properties.

\begin{enumerate}
\item[(p1)] $C$ consists of unary, binary and equality
  literals only. Moreover, $C$ is function-free, i.e.\ there is no
  literal in $C$ which contains function symbols.
\item[(p2)] $C$ either contains a binary literal, or it is ground, or it
contains no constants, no (in)equalities, and exactly one variable.
\item[(p3)] If there is a binary literal in $C$ then each variable in $C$ occurs in at least one binary literal.
\item[(p4)] If $C$ contains a positive binary literal $B$, then all
  the remaining literals, i.e.\ those in $C' = C \setminus \{B\}$, are negative binary literals, and the
  set of variables of $C'$ and $B$ is the same.
\end{enumerate}
\end{definition}

\noindent Note that the subcondition of (p2) ``contains \ldots no
(in)equalities'' is practically unnecessary. More precisely, if we
remove this subcondition, we can show that any (in)equalities
occurring in $C$ can be trivially deleted. Assume that there is a DL clause
$C$ which contains no binary literal and is not ground. Because of the
weaker form of (p2), we still know that $C$ contains no constants and
exactly one variable. Thus, any equality literals contained in $C$
have to be of the form $x = x$ or $x \neq x$. In the first case the
literal is always true, making $C$ useless, while the $x \neq x$
literal is always false and so it can be removed from $C$.

Now we formulate the following proposition (proved by simply checking
each of the clauses in Figure~\ref{figure:dlclauses}).

\begin{proposition}
For a given $\mathcal{SHIQ}$ knowledge base $\textit{KB}$, every
clause $C \in \textit{DL}(\textit{KB})$ is a DL clause. 
\end{proposition}

\noindent Note that the properties in Definition~\ref{def:dlike} are
necessary but not sufficient conditions for being a clause of the
form shown in Figure~\ref{figure:dlclauses}, i.e.\ these properties
may also hold for a clause which cannot be derived from a
$\mathcal{SHIQ}$ knowledge base. An example for such a clause is the
following:

\begin{equation}
\label{eq:nonshiq}
P(x) \vee \neg R(x, x)\textrm{.}
\end{equation}

\noindent In the rest of this section we discuss how to transform an
\emph{arbitrary} set of DL clauses, i.e.\ clauses satisfying Definition~\ref{def:dlike}, to a Prolog
program. However, in Section~\ref{generation}, which presents several 
optimisations of this transformation process, we will restrict the discussion
to inputs produced from $\mathcal{SHIQ}$ KBs, i.e.\ sets of clauses of the form shown in
Figure~\ref{figure:dlclauses}.

Let us now consider a certain type of unary predicates,
namely those corresponding to the $\top$ (top) concept. 
\begin{definition}[Top predicate]
Let $S$ be a set of DL clauses and let $p$ be a unary predicate
name which appears somewhere in $S$. Predicate name $p$ is said to
be a \emph{top predicate} if $S$ entails $\forall x\ldotp p(x)$.
\end{definition}

\noindent One can view top predicates as degenerate, as their negations
correspond to unsatisfiable concepts. Recall that it is
normally considered a modelling error if a DL knowledge base
contains an unsatisfiable concept, i.e.\ a concept equivalent to
$\bot$. 

Technically, we have to deal with top predicates because of a subtle
difference between the requirements of FOL and DL reasoning. When PTTP
is asked to list all $x$ instances satisfying $p(x)$, where $p$ is
a top predicate, it will normally return without instantiating $x$, which
indicates that all domain elements satisfy $p$. In contrast with this, a DL
reasoner is expected to enumerate all named individuals in the ABox, as
the answers to an instance retrieval query concerning a concept
corresponding to a top predicate. 

It can be shown that for DL clauses the top-predicate property does not depend on the
ground clauses, i.e.\ on the ABox part of the knowledge
base. Specifically, for $\mathcal{SHIQ}$ knowledge bases, one can
determine if $p$ is a top predicate by checking the satisfiability of
the concept $\neg p$, using a suitable TBox reasoning engine.

In order to be able to formulate our results in a simpler form, we 
define a transformation  removing all top predicates from a knowledge base.

\begin{definition}[Reduced form of a set of DL clauses]
\label{def:reduced}
Let $S$ be a set of DL clauses. We modify $S$ in the following
way. (1) We remove all literals in $S$ which refer to a negated
top predicate.  (2) We remove every clause $C$ from $S$ where $C$ contains
a positive literal with a top predicate. The remaining set of clauses
is called the \emph{reduced form} of $S$.
\end{definition}

\noindent The following proposition shows that this reduction step preserves all
information except for the top predicates.

\begin{proposition}
Let $S$ be a set of DL clauses. Let us extend the \emph{reduced form}
of $S$ with clauses of the form $p(x)$, for each top predicate $p$ in
$S$. This extended set of clauses is equivalent to $S$.
\end{proposition}

\begin{proof}\small
Easily follows form the fact that both transformation steps (1) and (2) in
Definition~\ref{def:reduced} are sound.
\end{proof}
\noindent In the following, we will restrict our attention to sets of DL
clauses which are in reduced form.

\subsection{Specialising PTTP for DL clauses}
\label{specpttp}

In this section we discuss how to specialise various features of PTTP for
the case of DL clauses.

\paragraph{Contrapositives}
The first step in applying the PTTP approach to a set of DL clauses $S$ is
to generate the contrapositives of each clause in $S$. This, in turn,
requires the introduction of new predicate names for negated literals.

We now reiterate the corresponding definitions from Section~\ref{pttp}. On
one hand, we extend these to handle equalities. On the other hand, we
specialise these definitions for DL clauses.  Recall that a DL clause is a
nonempty disjunction of literals, each literal is a possibly negated atomic
predicate, and an atomic predicate can take one of the following three
forms:
\begin{itemize}
\item a unary predicate $p(x)$
\item a  binary predicate $p(x,y)$
\item an equality $x=y$ 
\end{itemize}

\begin{definition}[The canonical form of literals]
Let $L$ be a literal, or a literal preceded by a negation symbol. The
canonical form of $L$, denoted by $\textit{can}(L)$, is defined as
follows.

\[ \textit{can}(L) = \left\{ \begin{array}{ll}
  \texttt{not\_}p(x)     & \mbox{if $L = \neg p(x)$}\\
  \texttt{not\_}p(x,y)   & \mbox{if $L = \neg p(x,y)$}\\
  \texttt{dif}(x,y)      & \mbox{if $L = \neg (x = y)$}\\
  \textit{can}(L')       & \mbox{if $L = \neg \neg L'$}\\
  L                      & \mbox{otherwise}
\end{array}
\right.
\]
\end{definition}

In the above definition, the first two cases remove a negation from
before atomic predicates and prefix the predicate name with
`\texttt{not\_}'. The third case transforms an inequality to a call of the
predicate \texttt{dif}. This is a predicate available in most Prolog
systems which ensures that its two arguments are not unifiable. For Prolog
programs generated from DL clauses, where the variables are instantiated to
constants sooner or later, this ensures that the arguments of \texttt{dif}
are indeed different. The fourth case removes double negation, while the
last one states that non-negated literals are left unchanged. This implies
that an equality is handled by the Prolog built-in predicate `\texttt{=}'.
We can implement inequality and equality using \texttt{dif} and
`\texttt{=}' because of the Unique Name Assumption, which states that an
inequality holds for any two distinct individual names, and thus an
equality can hold only when its two sides are identical individual names.

\begin{definition}[DL-contrapositive of a DL clause]
Let $\textit{DLC} = \mathop{\bigvee}_{1 \leq i \leq n}L_i$  be an arbitrary DL clause.
The Horn clause 
$$\textit{can}(L_k) \texttt{ :- }\textit{can}(\neg L_1), \ldots, \textit{can}(\neg L_{k-1}), \textit{can}(\neg L_{k+1}),
\ldots, \textit{can}(\neg L_n)$$
is called a DL-contrapositive of $\textit{DLC}$, provided $L_k$ is a (possibly
negated) unary or binary predicate.
\end{definition}

Note that we do not consider Horn clauses with an equality or inequality in
the head. Such clauses could be used to infer that two individuals are
equal or distinct. However, we work with the Unique Name Assumption, which
decides the issue of equality, and so such deductions are unnecessary.

\begin{definition}[DL program]
Let $S$ be a set of DL clauses. The  \emph{DL program} corresponding to $S$ is
denoted by $\textit{PDL}(S)$, and contains all DL-contrapositives of the
clauses in $S$,
i.e.\ $\textit{PDL}(\textit{S}) = \{ C | C$ is a DL-contrapositive of
$C_0$, and $C_0 \in S\}$.
\end{definition}

\noindent Horn clauses are usually grouped into predicates, according to
the functor of the clause head. The functor of a term is a pair consisting
of a name and arity (number of arguments). In Prolog, functors are normally
denoted by the expression \emph{Name/Arity}, for example
\texttt{foo/2}. Thus a DL program can be also viewed as a set of \emph{DL
  predicates}, each of which consists of all clauses of the DL program
which have a given head functor.  In the rest of the paper it is the
context which determines whether we view the DL program as a set of Horn
clauses, or as a set of DL predicates. Accordingly, we use the term ``DL
predicates'' as a synonym of ``DL program''.

As an example, in Figure~\ref{figure:dl_predicates_iocaste} we show
the four DL predicates produced from the Iocaste knowledge base of
Figure~\ref{figure:iocaste_kb}. Notice, for example, that the first
clause of the \texttt{Patricide/1} predicate comes from the ABox,
while the second comes from the TBox. We also show the six DL
predicates of the Happy KB in Figure~\ref{figure:dl_predicates_happy}.
We have not included the DL predicate \texttt{not\_hasChild/2} in
these examples, because we will soon prove that clauses with a negated
binary literal in the head are not needed (cf.\
Proposition~\ref{prop:remneg}).

\begin{figure}[htbp]
\begin{Verbatim}[numbers=left,numbersep=2pt,frame=single]
Ans(A)           :- hasChild(A, B), hasChild(B, C), 
                    Patricide(B), not_Patricide(C).
Patricide(o).
Patricide(A)     :- hasChild(B, A),  hasChild(C, B), 
                    Patricide(B), not_Ans(C).
not_Patricide(t). 
not_Patricide(A) :- hasChild(A, B), hasChild(C, A), 
                    not_Ans(C), not_Patricide(B).

hasChild(i, o). hasChild(i, p). hasChild(o, p). hasChild(p, t).
\end{Verbatim}
\caption{DL predicates of the Iocaste problem.}
\label{figure:dl_predicates_iocaste}
\end{figure}

\begin{figure}[htbp]
\begin{Verbatim}[numbers=left,numbersep=2pt,frame=single]
Happy(A)      :- hasChild(A, B), hasChild(B, C), hasChild(B, D),
                 Clever(C), Pretty(D).

not_Clever(A) :- hasChild(B, C), hasChild(C, A), hasChild(C, D),
                 Pretty(D), not_Happy(B).

not_Pretty(A) :- hasChild(B, C), hasChild(C, D), hasChild(C, A),
                 Clever(D), not_Happy(B)

Clever(lisa). Pretty(lisa). hasChild(kate, bob). hasChild(bob, lisa).
\end{Verbatim}
\caption{DL predicates of the Happy problem.}
\label{figure:dl_predicates_happy}
\end{figure}

\paragraph{Conjunctive queries}
Given the notion of DL programs, we now discuss how such a program can
be queried.  Instance retrieval queries include possibly negated
atomic concepts and unnegated (positive) binary roles, for example \texttt{not\_Patricide(A)} and \texttt{hasChild(A,
  B)}. The former is supposed to enumerate all possible individuals
known to be non-patricide. The latter is expected to enumerate all
pairs of individuals between whom the \texttt{hasChild} relation
holds. In this paper we support \emph{conjunctive queries}
\cite{1210833}, which are conjunctions of the above instance retrieval
constructs. The execution of a conjunctive query with $n$ distinct
variables is expected to return a set of $n$-tuples, each being a
variable assignment satisfying all the conjuncts.  An example of a
conjunctive query with three variables is \texttt{(Patricide(X),
  hasChild(X, Y), not\_Patricide(Y), hasChild(Y, Z))}.

\paragraph{Basic simplifications}
We now discuss three basic simplifications of PTTP for the special case of
DL clauses. First, let us notice that the occurs check is not necessary, as
DL clauses are function-free. Next, for the case of conjunctive queries, we
claim that (a) contrapositives with negated binary literals in the head can
be removed from the DL program, and (b) ancestor resolution is not needed
for roles. Before proving these claims let us observe the following
proposition.

\begin{proposition}
\label{prop:negbin}
In a DL program, a negated binary predicate can only be invoked within a
negated binary predicate.
\end{proposition}

\begin{proof}\small
Let $C$ be a clause in the DL program, such that the body of $C$
contains a negated binary goal $G$. Accordingly, $C$ is the
contrapositive of a DL clause where the binary literal corresponding
to $G$ is a positive literal. However, because of property (p4) in
Definition~\ref{def:dlike}, we know that this DL clause cannot contain
any more positive binary literals and, moreover, it can only contain
negative binary literals. Thus, the head of $C$ must correspond to a
negative binary literal.
\end{proof}

\begin{proposition}
\label{prop:remneg}
Removing contrapositives with negated binary literals in the head
from a DL program does not affect the execution of conjunctive queries.
\end{proposition}

\begin{proof}\small
This is a direct conclusion of Proposition~\ref{prop:negbin} and the
fact that a conjunctive query cannot contain negated binary
goals.
\end{proof}

\noindent Note that when the clauses with a negated binary literal in
the head are removed, no negative binary literals will remain in the
bodies (as the latter only appear in clauses with negative binary
heads, cf.\  Proposition~\ref{prop:negbin}). Thus, unless stated otherwise,
the term \emph{binary predicate} will refer to unnegated binary
predicates, from now on.

\begin{proposition}
Ancestor resolution is not required for binary predicates to answer
conjunctive queries w.r.t.\ a DL program $S$.
\end{proposition}

\begin{proof}\small
This trivially follows from the fact that negative binary predicates are
never called and can never occur in the ancestor list.
\end{proof}

\paragraph{The binary-first rule}
 The next simplification of PTTP, the replacement
of iterative deepening by \emph{loop elimination}, requires
that a specific restriction is imposed on the placement of the binary goals
in clause bodies. We now present an important property of binary goals,
introduce the \emph{binary-first} body ordering rule, and discuss its
implications.

\begin{proposition}[Binary instantiation]
\label{prop:binterm}
Let $S$ be a set of DL predicates and $B$ be a binary goal. If
the Prolog execution of $B$ w.r.t.\ $S$ terminates with success, it
instantiates both its arguments.
\end{proposition}

\begin{proof}\small
Let us indirectly assume that there is a binary goal $B(X,Y)$ which
terminates, but one of its arguments, let us say $X$, remains
uninstantiated. Since $B(X,Y)$ terminates, there is a finite Prolog
proof tree $T$ for it. Let us consider the nodes in $T$ containing a
binary goal with $X$ as one of its arguments. As $T$ is finite, there
exists a ``lowest'' of these, i.e.\ a node with no occurrences of $X$
in binary goals below it. However, this contradicts property (p4) of
DL clauses in Definition~\ref{def:dlike}.
\end{proof}

\begin{definition}[The binary-first rule]
\label{def:binfirst}
The body of a Prolog clause $C$ is ordered according to the
\emph{binary-first rule} if (1) each binary goal $B$ in the body of $C$ precedes all unary and
(in)equality goals containing any of the variables occurring in $B$, and
(2) if the body of $C$ contains a binary goal with the head variable as an argument, then
at least one such goal precedes all unary goals.
\end{definition}

\noindent For an arbitrary clause $C$
containing a binary goal condition (1) ensures that all unary and equality goals within
$C$ are called with a ground argument, while condition (2) guarantees that
by the time the first unary goal in $C$ is called, the head variable is ground.

\begin{proposition}[Groundness of unary and equality predicates]
\label{prop:ground}
Let $S$ be a set of DL predicates and let us use the binary-first rule
during the Prolog execution. (a) If a unary predicate is invoked with
a variable argument then its parent goal (the goal which is calling the
current goal) is a unary predicate with the same variable. (b)
equality predicates (i.e.\ \texttt{=\,/2} and \verb+dif/2+) are always
invoked with ground arguments.
\end{proposition}

\begin{proof}\small
(a) Let $G$ be the unary goal which is invoked with the variable
  argument $V$ in clause $C$\@. Because of condition (p4) in
  Definition~\ref{def:dlike} and Proposition~\ref{prop:remneg}, a
  unary goal can only be invoked from within a clause of a unary
  predicate. If there are binary goals within the body of $C$, then
  $G$ is always preceded by a binary goal containing $V$ according to
  (p3) in Definition~\ref{def:dlike} and the binary-first
  rule. Because of Proposition~\ref{prop:binterm}, however, we know
  that variable $V$  is already instantiated by the time $G$ is
  invoked. This means there can be no binary goals in the body of $C$, and
  so, according to (p2) in Definition~\ref{def:dlike}, $C$ contains a
  single variable. Consequently, $V$ is the variable appearing in the head of $C$.

\noindent (b) Assume that an equality goal is invoked with an
uninstantiated variable $V$ within a Horn clause obtained from the DL
clause $C$\@. Because of (p2) and (p3) there has to be a binary literal
in clause $C$ containing $V$\@. Because of (p4) this binary literal is
negative. The binary-first rule means that the given binary literal is
executed before the equality predicate, and
Proposition~\ref{prop:binterm} ensures that $V$ is instantiated, which
contradicts our initial, indirect assumption.
\end{proof}

\begin{proposition}
If the binary-first rule is applied, the \texttt{=\,/2} and
\verb+dif/2+ predicate invocations can be replaced by \texttt{==\,/2}
and \verb+\+\texttt{==\,/2} (the standard Prolog term comparison predicates
checking if their arguments are identical, and non-identical respectively).
\end{proposition}
\begin{proof}\small
When invoked with ground arguments, the built-in predicates \texttt{==\,/2} and
\verb+\+\texttt{==\,/2} have the exact same semantics as the predicates they
replace. 
\end{proof}

Let us
now examine what ancestor--descendant pairs are possible for unary
predicates. In general, we have the following five cases, where
variables \texttt{X} and \texttt{Y} are distinct, but predicate names
\texttt{q} and \texttt{p}, as well as constants \texttt{i} and
\texttt{j} can be the same.

\begin{enumerate}
\item[(c1)] within executing \texttt{p(i)} we encounter a goal \texttt{q(j)}
\item[(c2)] within executing \texttt{p(i)} we encounter a goal \texttt{q(X)}
\item[(c3)] within executing \texttt{p(X)} we encounter a goal \texttt{q(i)}
\item[(c4)] within executing \texttt{p(X)} we encounter a goal \texttt{q(X)}
\item[(c5)] within executing \texttt{p(X)} we encounter a goal \texttt{q(Y)}
\end{enumerate}

\noindent The following proposition states that, in the case of DL
predicates, some of these cases cannot occur.

\begin{proposition}
\label{prop:anc}
Let $S$ be a set of DL predicates. When using the binary-first rule,
cases (c2), (c3), and (c5) cannot occur during Prolog
execution. Furthermore, if a unary goal is invoked  with a variable
argument \texttt{X}, then all its ancestors (including the outermost one,
the concept query goal) are unary goals having variable \texttt{X} as their argument.
\end{proposition}

\begin{proof}\small
When the binary-first rule is used, the cases (c2) and (c5) cannot
occur, as a direct consequence of
Proposition~\ref{prop:ground}. Furthermore Proposition~\ref{prop:binterm}
and part (2) of the definition of the binary-first rule ensure that the
parent of a ground unary goal is ground, too. This implies that all
ancestors of a ground unary goal are ground, hence case (c3) cannot occur.

Now assume that the condition of the second claim holds,
i.e.\ there is a unary goal with a variable argument on the ancestor list.
This means that case (c4) has to apply to the given goal. Consequently, the argument
of its parent goal is the same variable. By repeatedly applying this
argumentation we can conclude that all ancestors of the given goal have the
same variable as their argument.
\end{proof}

\paragraph{Loop elimination}
As the next simplification of the PTTP approach for DL programs,
we replace iterative deepening by normal Prolog
depth-first search, extended with a straightforward \emph{loop elimination}
technique. This feature, which involves  pruning certain branches of the
Prolog search tree,  appeared already in PTTP, as an optimisation
\cite{stickel92prolog}. However, in the context of DL programs, as opposed to
arbitrary first-order logic clauses, loop elimination can itself ensure
termination, as discussed below. 

In the next two definitions we refer to an extension of Prolog
execution where the list of ancestor goals is maintained.

\begin{definition}[Goals subject to loop elimination]
A Prolog goal $G$ encountered in the context of an ancestor list $L$ is
\emph{subject to loop elimination} if $G$ occurs in $L$: more
precisely, if $L$ contains an element $G'$ for which $G$ \texttt{==}
$G'$ holds. Recall that \texttt{==} denotes the standard Prolog predicate
which succeeds if its operands are identical.
\end{definition}

\begin{definition}[Loop elimination]
Let $P$ be a Prolog program and $G$ a Prolog goal. Executing $G$
w.r.t.\ $P$ using \emph{loop elimination}  means the
Prolog execution of $P$  extended in the following way: we
stop the given execution branch with a failure whenever we encounter a
goal $G$ which is subject to loop elimination.
\end{definition}

\noindent Using the notion of loop elimination we can formulate some
termination results.

\begin{proposition}[Termination of DL execution]
\label{prop:terminate}
Let $S$ be a set of DL predicates. Assuming loop elimination and that the
binary-first rule is used, the execution of an arbitrary goal w.r.t.\ $S$
always terminates.
\end{proposition}

\begin{proof}\small
Let us indirectly assume that there exists a goal $G$ the execution of
which does not terminate. Because of loop elimination this can only
happen if we can build an ancestor list with infinitely many distinct
goals. Since the number of predicate and constant names is finite this
means that the ancestor list contains an infinite number of distinct
variables.

However, according to Proposition~\ref{prop:anc}, unary goals on the
ancestor list contain at most one variable. Property (p4) in
Definition~\ref{def:dlike} implies that any variable appearing in a clause
body within a binary DL predicate appears in the corresponding clause head,
too. Thus a new variable can only be introduced when a binary goal is
invoked in a unary predicate. However, property (p4) also implies that a
binary predicate invokes binary goals only, with no new variables. Furthermore,
Proposition~\ref{prop:binterm} states that by the time a binary goal exits,
both its arguments are instantiated. This means that the ancestor list can
contain at most two uninstantiated variables at any time, contradicting our
indirect assumption.
\end{proof}

\noindent Having proved that loop elimination and the
binary-first rule guarantee termination, let us consider the issue
whether loop elimination is complete, i.e.\ any solution that can be
obtained by PTTP can also be obtained in the presence of loop
elimination. 

Note that for normal Prolog execution, loop elimination is obviously
complete. That is, given an arbitrary proof tree of a goal $P$ where
goal $G_1$ appears in the subtree of an identical goal $G_2$ we can
always create a new proof tree of $P$ where we replace the proof of
$G_1$ by the proof of $G_2$. Continuing this process we can obtain a
proof tree of $P$ which does not contain any goals subject to loop
elimination.

However, PTTP extends the normal Prolog execution by applying ancestor
resolution for goals. This means that successful execution of a goal
$G$ may depend on the location of $G$ within a proof tree (as this
determines the ancestors of $G$). The
completeness of loop elimination in the presence of ancestor
resolution was first stated in \cite{stickel92prolog}. We now give a
reformulation of this statement. 

\begin{proposition}[Completeness of loop elimination]
\label{prop:completeness}
Let $T$ be a proof tree of a goal $G$ corresponding to a PTTP
execution, which contains a goal subject to loop elimination. It is
possible to create another proof tree of goal $G$ which contains no
goals subject to loop elimination.
\end{proposition}

\paragraph{Deterministic ancestor resolution}
We now present an important property of ancestor resolution for DL
programs, which is the basis of our last simplification of the PTTP approach.

\begin{proposition}[Deterministic ancestor resolution]
\label{prop:nounification}
If loop elimination and the binary-first rule is applied for DL
predicates, exactly one ancestor can be applicable in a successful
ancestor resolution step, i.e.\ ancestor resolution is deterministic.
\end{proposition}

\begin{proof}\small
Let us examine cases (c1) and (c4), allowed by
Proposition~\ref{prop:anc}, for \texttt{q = not\_p}, i.e.\ the case
relevant for ancestor resolution. In the case of (c1), ancestor resolution
succeeds if \texttt{i} and \texttt{j} are the same, and fails
otherwise. Note that this ancestor resolution step can succeed only
once. This is because loop elimination ensures that the ancestor list
cannot contain \texttt{p(i)} more than once.   Case (c4) succeeds with no
substitution and, similarly to (c1), it can succeed only once. This is because
\texttt{p(X)} cannot occur in the ancestor list more than once and if
\texttt{p(X)} is there, then no goal of the form \texttt{p(Y)} can occur on
the ancestor list, where \texttt{Y} is a variable different from
\texttt{X} (cf.\ Proposition~\ref{prop:anc}).
\end{proof}

\paragraph{Principles of DLog execution}
 To conclude this section, Figure~\ref{figure:pttpdif} gives a
summary of the principles we use in the execution of DL predicates and
compare these to their counterparts in PTTP. 

\begin{figure}[htbp]
\begin{magicbox}
\vspace*{1.3ex}
\begin{enumerate}
\item[(a)] DLog uses normal Prolog unification rather than unification with occurs check
\item[(b)] DLog uses loop elimination instead of iterative deepening
\item[(c)] DLog eliminates contrapositives with negated binary literals in the head
\item[(d)] DLog does not apply ancestor resolution for roles
\item[(e)] DLog uses deterministic ancestor resolution
\end{enumerate}
\end{magicbox}
\caption{A comparison of DLog with generic PTTP.}
\label{figure:pttpdif}
\end{figure}

\noindent We also formulate the main result of this section as the
following theorem.

\begin{theorem}[soundness and completeness of the DLog execution]
\label{theorem:dlprograms}
Let $S$ be a set of DL clauses in reduced form and $Q$ a conjunctive
query. Let $P$ be a set of Prolog clauses obtained from $\textit{PDL}(S)$ by
removing clauses with negated binaries in the head,  ordering clause
bodies according to the binary-first rule, and replacing \texttt{=\,/2} and \verb+dif/2+
 by \texttt{==\,/2} and \verb+\+\texttt{==\,/2}, respectively.  Let us extend a
standard Prolog engine with (1) loop elimination and (2) deterministic
ancestor resolution for unary predicates only. If the extended Prolog
engine is invoked with the program $P$ and goal $Q$, it will terminate
and enumerate those and only those ground instantiations of the
variables of $Q$ for which $Q$ is entailed by $S$.
\end{theorem}

\begin{proof}\small
This is a direct consequence of the fact that PTTP is a sound and
complete FOL theorem proving technique, and of Propositions
\ref{prop:remneg}, \ref{prop:terminate}, \ref{prop:completeness} and
\ref{prop:nounification}.
\end{proof}

\begin{corollary}
Let $\textit{KB}$ be a $\mathcal{SHIQ}$ knowledge base and $Q$ a
conjunctive query in which no concept equivalent to $\top$ occurs. In
this case the technique of Theorem~\ref{theorem:dlprograms}, applied
to $\textit{DL}(\textit{KB})$ and $Q$ provides finite, sound and
complete execution for conjunctive queries.
\end{corollary}

\subsection{Interpreting DL predicates}

In Figure~\ref{figure:interpreter} we show a complete interpreter,
which is able to execute DL predicates stored as normal dynamic
predicates in Prolog.

\begin{figure}[htbp]
\begin{Verbatim}[numbers=left,numbersep=2pt,frame=single]
interp(true, _) :- !.
interp((Goal1, Goal2), AncList) :- !,                   
    interp(Goal1, AncList),                             
    interp(Goal2, AncList).
interp(Goal, AncList) :-                                
    (   equality(Goal) -> call(Goal)                    % (in)equalities
    ;   member(Goal0, AncList), Goal0 == Goal -> fail   % loop elimination                              
    ;   neg(Goal, NegGoal), memberchk(NegGoal, AncList) % ancestor resolut.
    ;   NewAncList = [Goal|AncList],                            
        clause(Goal, Body),                     
        interp(Body, NewAncList)
    ).

equality(_ == _).
equality(_ \== _).
\end{Verbatim}
\caption{A full interpreter for DL clauses.}
\label{figure:interpreter}
\end{figure}

The interpreter is invoked through the predicate \texttt{interp/2} with a
conjunctive query in the first, and an empty ancestor list in the second
argument. The interpreter handles  (in)equalities (line 6), ensures loop
elimination (line 7) and provides deterministic ancestor
resolution (cf.\ the use of \texttt{memberchk/2} in line 8). The new
ancestor list is built in line 9. 
The term 
\texttt{NegGoal} in line 8 is the negated version of \texttt{Goal}, as
defined below.

\begin{definition}[Negated version of a goal or a predicate]
\label{def:negatedterm}
The negated version of a Prolog goal $G$, denoted by $\textit{not\_}G$, is
constructed by removing the \texttt{not\_} prefix from the predicate name
of $G$, if it has such a prefix; or otherwise adding this prefix to the
predicate name. 

We overload this notation, and use it for predicate names and functors as well.
\end{definition}

For example, if $G_1=$\,\texttt{p(X)} and $G_2=\texttt{not\_p(X)}$, then
their negated versions are $\textit{not\_}G_1=\texttt{not\_p(X)}$ and
$\textit{not\_}G_2=\texttt{p(i)}$. Also, if $P$ is the predicate
\texttt{not\_p/1}, then $\textit{not\_}P$ denotes the predicate
\texttt{p/1}. 

We now show an example of invoking the interpreter. Assume that the DL
predicates of the Iocaste problem, as shown in
Figure~\ref{figure:dl_predicates_iocaste}, are loaded as dynamic Prolog
predicates.  One can then run the Iocaste query in the following way:

\begin{Verbatim}[numbers=none,numbersep=2pt,frame=single]
| ?- setof(X, interp('Ans'(X), []), Sols).
Sols = [i] ;
no
\end{Verbatim}

Note that the interpreter may return a solution several times, but the
standard Prolog predicate \texttt{setof/3} forms a \emph{set} of the
solutions, i.e.\ an ordered list containing each solution only once. In
Section~\ref{projection} we discuss an optimisation which ensures that
each solution of a unary predicate is returned exactly once.

According to
Theorem~\ref{theorem:dlprograms}, the interpreter is a sound and
complete theorem prover for DL programs and composite queries.

\subsection{Compiling DL predicates}
\label{compiling}

The interpreted solution is pretty straightforward. However, for
performance reasons, we also consider generating Prolog code which does
not require a special interpreter.  The idea is to include loop
elimination and ancestor resolution in the DL predicates themselves, and to
extend the predicates with an additional argument for storing the ancestor
list. 

In contrast with the interpreter, the compiler treats TBox and ABox clauses
separately. This is crucial to allow efficient execution of ABox queries,
e.g.\ by using databases. Therefore, we now distinguish between the TBox
and ABox part of a DL program:

\begin{definition}
Let $P$ be a  DL program. The ABox part of $P$, denoted by
$P_{\mathcal{A}}$, is the set of all ground facts in $P$. The TBox part of $P$, denoted
by $P_{\mathcal{T}}$, contains all remaining clauses,
i.e.\ $P_{\mathcal{T}} = P \setminus P_{\mathcal{A}}$.
\end{definition}

For
example, in Figure~\ref{figure:dl_predicates_iocaste}, clauses in
lines 3, 6 and 10 form the ABox DL predicates, while the remaining lines contain
the TBox DL predicates.

We need the following notion for describing the transformation process.

\begin{definition}[Signature]
Let $P$ be a  DL program. The \emph{signature} of  $P$ is the set of functors of the form $C/1$
and $R/2$ where $C$ is a unary predicate name and $R$ is a binary
predicate name which appears anywhere in $P$.
\end{definition}

\noindent We will apply the notion of signature to the ABox and TBox part of a
DL program (as these parts can be viewed as DL programs themselves). 
For example, if $P$ is the Iocaste DL program  shown in
Figure~\ref{figure:dl_predicates_iocaste}, then the signature of $P$ is $\{$\verb+Ans/1+,
\verb+not_Ans/1+, \verb+Patricide/1+, \verb+not_Patricide/1+,
\verb+hasChild/2+$\}$.  Note that predicate \verb+not_Ans/1+ has no
clauses, but it still belongs to the signature. The signature of
$P_{\mathcal{T}}$ is the same as that of $P$, while the signature of
$P_{\mathcal{A}}$ excludes 
\verb+Ans/1+ and \verb+not_Ans/1+.

We now define two auxiliary transformations which are used in the compilation of a DL
predicate into Prolog code.

\begin{definition}[The expanded version of a term]
Let $T$ be an arbitrary Prolog term with name $N$ and arguments $A_1,
\ldots, A_k$. Let $Z$ be another Prolog term. The \emph{expanded version} of
$T$ w.r.t.\ $Z$, denoted by $\textit{Expd}(T, Z)$ is defined as the
term $N(A_1, \ldots, A_k, Z)$.
\end{definition}

\begin{definition}[The ancestorised form of a clause]
For an arbitrary Prolog clause $C$, whose head is $H$ and body is $B_1,
\ldots, B_n$, the \emph{ancestorised} form of $C$, $\Omega(C)$, is a
Prolog clause defined as follows. The head of $\Omega(C)$ is
$\textit{Expd}(H, \texttt{AL})$, where $\texttt{AL}$ is a newly
introduced variable. The body of $\Omega(C)$ is $E_0, E_1, \ldots,
E_n$. Here, $E_0$ is the goal $\texttt{NewAL =
  [}H\texttt{|}\texttt{AL}\texttt{]}$, where \texttt{NewAL} is a new
variable, and $E_i = \textit{Expd}(B_i, \texttt{NewAL})$,
for $0 < i \leq n$.
\end{definition}

\noindent As an example, the ancestorised form of the Iocaste clause shown in
lines 1--2 in Figure~\ref{figure:dl_predicates_iocaste} is the
following: 

\begin{Verbatim}[numbers=left,numbersep=2pt,frame=single]
Ans(A, AL) :- NewAL = [Ans(A)|AL], hasChild(A, B, NewAL), 
              hasChild(B, C, NewAL), Patricide(B, NewAL), 
              not_Patricide(C, NewAL).
\end{Verbatim}

\noindent Here \texttt{AL} denotes the old, while \texttt{NewAL} denotes
the updated ancestor list. 

\begin{definition}[The compiled form of a DL predicate]\label{def:comp}
Let $P$ be a DL predicate with the functor $N/A$ and clauses
$C_1, \ldots, C_n, n \geq 0$. Let $H$ denote a most general goal with
name $N$ and arity $A$, i.e.\ a term each argument of which is a
distinct variable. The \emph{compiled version} of $P$, denoted by
$\Delta(P)$, is the sequence of clauses $F_1, \ldots, F_{n+3}$,
defined as follows, where $\textit{not\_}H$ is the negation
  of goal $H$, see Definition~\ref{def:negatedterm}:
\begin{description}
\item[$F_1$:] $\textit{Expd}$\texttt{($H$, AL) :- member(G, AL),
  G==$H$, !, fail.} (cf.\ line 7, Figure~\ref{figure:interpreter})
\item[$F_2$:] $\textit{Expd}$\texttt{($H$, AL) :-
  memberchk(}$\textit{not\_}$\texttt{$H$, AL).} (cf.\ line 8,
  Figure~\ref{figure:interpreter})
\item[$F_3$:] $\textit{Expd}$\texttt{($H$, AL) :- abox:$H$.}
\item[$F_{3+i}$:] $\Omega(C_{i}), 0 < i \leq n$.
\end{description}
\end{definition}

\noindent This definition says that the compiled version of a
predicate contains the ancestorised
version of the clauses in the predicate, preceded with three new clauses. These new
clauses are responsible for loop elimination, ancestor resolution and
for accessing the content of the ABox (stored in the Prolog module
\texttt{abox}, cf.\ the prefix ``\texttt{abox:}'').

Note that clause $F_{3}$ provides the link between a compiled predicate and
its ABox part, where the predicate representing the ABox has one argument
less than the compiled predicate. However, certain optimisations of
Section~\ref{generation} remove the additional argument of the 
compiled predicate.  By placing the ABox predicates in the \texttt{abox}
module we make sure that the ABox part is separated from the rest of the
compiled predicate. However, for the sake of readability, we omit the \texttt{abox:}
prefixes from the example programs presented in the paper.

\begin{definition}[The compiled form of a DL program]
Let $P$ be a  DL program, and let $\{N_1/A_1, \ldots, N_k/A_k\}$
  be the signature of $P_\mathcal{T}$. The
  compiled form of 
  $P$ is the set $\{C_1, \ldots, C_k\}
    \union \{\texttt{abox:}C \mid C \in \textit{PDL}(P_\mathcal{A})\}$, where $C_i =
    \Delta(Z_i)$ and $Z_i = \{C \in
      P_{\mathcal{T}}| $ $N_i/A_i$ is the functor of the head of $C$ $\}$.
\end{definition}

\noindent Thus the compiled form of a DL program $P$ is obtained by
compiling the predicates belonging to each functor appearing in the TBox
part of $P$, and adding the ABox DL predicates, stored in the \texttt{abox}
Prolog module. 

Some of the clauses in the compiled form of a predicate can be
omitted under certain conditions. For example,  we do not have to generate clauses of
type $F_2$ for roles (cf.\ item (d) in Figure~\ref{figure:pttpdif}). Furthermore, if
$N/A$ does not appear in the ABox signature, then we can omit the clause
of type $F_3$ for the predicate $N/A$. Also, there are predicates which
have no TBox clauses and thus consist of
nothing but an $F_3$ clause. In case of such \emph{atomic predicates} we
can even get rid of the $F_3$ clause, if we remove the additional argument
(holding the ancestor list) from each invocation and precede it with the
\texttt{abox:} module qualification. 
These optimisations will be covered in detail in 
Section~\ref{classification}.

Note that there can be predicates in a DL program which appear in
clause bodies but not in clause heads.  As an example, consider the
predicate \texttt{not\_Ans}, called in lines 5 and 8 of
Figure~\ref{figure:dl_predicates_iocaste}. This predicate has no
clauses, yet it can succeed using ancestor resolution.

Thus it is important that the operation $\Delta$ can be applied to 
empty predicates. In this case the compiled version consists solely of clauses $F_2$
and $F_3$, because clause $F_1$, serving for loop elimination, can be omitted
as an empty predicate cannot appear on the ancestor list.  If,
based on the ABox signature, we can omit $F_3$ as well, we get a
special case: a compiled predicate which can succeed \emph{only}
through ancestor resolution, i.e.\ using clause $F_2$. Predicates of
this type are called \emph{orphan predicates}, while their invocations are
called \emph{orphan goals} (for the exact definition of orphan predicates see
Section~\ref{subsection:princ_optim}).

\subsection{Compilation examples}
\label{shiqcompiled}

As discussed above, the compilation of DL predicates relies on adding
appropriate pieces of Prolog code to the DL clauses to handle ancestor
resolution and loop elimination. We demonstrate this technique by
presenting the complete Prolog translation of our two introductory
examples. 

The DL predicates of the Iocaste example were presented in
Figure~\ref{figure:dl_predicates_iocaste}. The compiled form of this DL
program is shown in Figure~\ref{figure:naive_translation_iocaste}.

\begin{figure}[htbp]
\begin{Verbatim}[numbers=left,numbersep=2pt,frame=single]
Ans(A, B) :-  member(C, B), C == Ans(A), !, fail.
Ans(A, B) :-  memberchk(not_Ans(A), B).
Ans(A, B) :-  C = [Ans(A)|B], hasChild(D, E), hasChild(A, D),  
              Patricide(D, C), not_Patricide(E, C).

Patricide(A, B) :- member(C, B), C == Patricide(A), !, fail.
Patricide(A, B) :- memberchk(not_Patricide(A), B).
Patricide(A, _) :- Patricide(A).
Patricide(A, B) :- C = [Patricide(A)|B], hasChild(E, D), 
                   hasChild(D, A), Patricide(D, C), not_Ans(E, C).

not_Patricide(A, B) :- member(C, B), C == not_Patricide(A), !, fail.
not_Patricide(A, B) :- memberchk(Patricide(A), B).
not_Patricide(A, _) :- not_Patricide(A).
not_Patricide(A, B) :- C = [not_Patricide(A)|B], hasChild(E, A), 
                       hasChild(A, D), not_Patricide(D, C), not_Ans(E, C).

not_Ans(A, B) :- memberchk(Ans(A), B).

Patricide(o).   not_Patricide(t). 
hasChild(i, o). hasChild(i, p). hasChild(o, p). hasChild(p, t).
\end{Verbatim}
\caption{The complete Prolog translation of the Iocaste problem.}
\label{figure:naive_translation_iocaste}
\end{figure}

Most predicates in this program  have an additional argument, used to pass
the ancestor list from call 
to call. For example, in line 10,  the goal
\texttt{Patricide(D, C)} is invoked, where \texttt{C} contains the new
ancestor list constructed in line 9.

In general, the content of the ABox can be either described as Prolog
facts, as shown in lines 20--21, or can be stored externally in some
database. In the latter case one has to  provide ``stubs'' to
access the content of the ABox. Namely, one should provide three
predicates, for \texttt{Patricide/1}, 
\texttt{not\_Patricide/1} and  \texttt{hasChild/2}. These
predicates should instantiate their head variables by querying
the underlying database in an appropriate manner. In the following, for the
sake of simplicity, we describe the 
content of the ABox as Prolog facts in the generated
programs. 

As the Iocaste  example does not contain role axioms, the role
predicate \texttt{hasChild} is an  atomic predicate. Therefore the
two-argument version is invoked directly, without ancestorisation, see
e.g.\ \texttt{hasChild(D, E)} in line 3.

In Figure~\ref{figure:naive_translation_iocaste}, the first clauses of most
predicates are responsible for loop elimination: the clauses in lines 1, 6,
and 12 check whether the ancestor list contains the goal in question, and
cause the predicate to fail, if this is the case.

Clauses in lines 2, 7, 13, and 18 are used to check whether the ancestor list
contains the negation of the goal in question. If so, ancestor resolution
takes place, which possibly substitutes the head variable \texttt{A}. As explained
earlier, we leave a choice point here, so that the remaining clauses of the
given predicate can be executed if, for example, the branch using the
ancestor resolution fails.

Line 18 in Figure~\ref{figure:naive_translation_iocaste}
shows how  an orphan predicate is translated, producing a single clause.

Having compiled the program of Figure~\ref{figure:naive_translation_iocaste},
we can retrieve the instances of the concept \texttt{Ans} in the following
way: 

\begin{small}
\begin{verbatim}
| ?- setof(X, 'Ans'(X, []), Sols).
Sols = [i] ?
\end{verbatim}
\end{small}

Let us now compare the handmade translation of the Iocaste problem in
Figure~\ref{figure:iocaste_handmade} with the machine translation shown
in Figure~\ref{figure:naive_translation_iocaste}. The goal \texttt{Ans(X)}
in the former corresponds to \texttt{Ans(X, 
[])} in the latter. Furthermore, \texttt{dPatricide(Z, X)} corresponds to
\texttt{Patricide(Z,[\ldots Ans(X)\ldots])}. The second argument of the
\texttt{dPatricide/2} goal, variable \texttt{X}, stores the top individual
of the Iocaste pattern (i.e.\ Iocaste herself), so that each member of the
chain in Figure~\ref{figure:patterns} can be checked to be a child of
\texttt{X}. The same effect is achieved in the machine translation by
placing \texttt{Ans(X)} on the ancestor list, and retrieving it later using
ancestor resolution.

A further difference is that the predicate
\texttt{not\_Patricide/2} does not appear in the handmade variant. This is
because \texttt{not\_Patricide/2} describes the same pattern as
\texttt{Patricide/2} (see Figure~\ref{figure:patterns}), but builds it in
the reverse order. 

Also note that the predicates in the machine translation have more clauses
than in the handmade version. Some of these are superfluous, and will
actually be removed by optimisations presented in
Section~\ref{generation}. This includes the clause responsible for loop
elimination in \texttt{Ans/2}, and the one responsible for ancestor
resolution in \texttt{Patricide/2}. However, the clause ensuring loop
elimination in \texttt{Patricide/2} has to stay, as termination can not be
assured without it, in the presence of potentially cyclic \texttt{hasChild}
relations.

To conclude the presentation of  the generic
compilation scheme we  show the
translation of the Happy knowledge base in
Figure~\ref{figure:naive_translation_happy}. Here a new ancestor list is built in lines 8
and 13. As a trivial simplification, we do not build a new
ancestor list if it is not passed to any of the goals in the body. This
happens when the clause invokes atomic predicates only, as in 
lines 3--4 of the \texttt{Happy} predicate. 

\begin{figure}[htbp]
\begin{Verbatim}[numbers=left,numbersep=2pt,frame=single]
Happy(A, B)      :- member(C, B), C == Happy(A), !, fail.
Happy(A, B)      :- memberchk(not_Happy(A), B).
Happy(A, _)      :- hasChild(A, B), hasChild(B, C), hasChild(B, D),
                    Clever(C), Pretty(D).

not_Clever(A, B) :- member(C, B), C == not_Clever(A), !, fail.
not_Clever(A, B) :- memberchk(Clever(A), B).
not_Clever(A, B) :- F = [not_Clever(A)|B], hasChild(C, A), hasChild(D, C), 
                    hasChild(C, E), Pretty(E), not_Happy(D, F).

not_Pretty(A, B) :- member(C, B), C == not_Pretty(A), !, fail.
not_Pretty(A, B) :- memberchk(Pretty(A), B).
not_Pretty(A, B) :- F = [not_Pretty(A)|B], hasChild(C, A), hasChild(D, C), 
                    hasChild(C, E), Clever(E), not_Happy(D, F).

not_Happy(A, B)  :-  memberchk(Happy(A), B).

Clever(lisa). Pretty(lisa). hasChild(kate, bob). hasChild(bob, lisa).
\end{Verbatim}
\caption{The complete Prolog translation of the Happy problem.}
\label{figure:naive_translation_happy}
\end{figure}

Notice that the Prolog code in
Figure~\ref{figure:naive_translation_happy} is much bigger than the
hand-made translation in
Figure~\ref{figure:happystraightforwardtranslation}. However, the
optimisations of Section~\ref{generation} will simplify this code so
that it becomes the same as that in
Figure~\ref{figure:happystraightforwardtranslation}.

\subsection{Summary}

In this section we have showed how to transform a $\mathcal{SHIQ}$
description logic knowledge base into a Prolog program performing 
instance retrieval tasks for the given knowledge base.

In the first stage of the transformation we convert the $\mathcal{SHIQ}$ axioms to an equivalent
set of so-called DL clauses,  using the techniques of \cite{motik06PhD}
and \cite{zombori}. These clauses are then compiled
to Prolog code, using specialised variants of PTTP techniques, such as ancestor
resolution and loop elimination. We gave a formal description of the
transformation process and proved that it is sound and complete, and that
it always terminates.

The transformation has an important property: it does not modify the
ABox part of the $\mathcal{SHIQ}$ knowledge base in question. This
allows for the ABox to be stored externally. Equally important is the
fact that the transformation of the TBox part relies only on the
signature of the ABox, but not on the content.

\section{Optimising DL compilation}
\label{generation}

The translation principles presented in the previous section
are complete and result in programs which can already be executed in a
standard Prolog environment, but they are not efficient enough. In this
section we describe a series of optimisations which result in a much
more efficient Prolog translation. We note that most of these
optimisations could also be built into the interpreter itself, but here
we deal with the compiled form only.

In Section~\ref{shiqreasoning} we introduced the general
interpretation and compilation schemes for so called DL clauses,
which are more general than the clauses obtained from  $\mathcal{SHIQ}$
knowledge bases. However, in the present section, we do assume that the
DL program to be optimised is obtained from a $\mathcal{SHIQ}$ knowledge base $\textit{KB}$,
i.e.\ it is of the form $\textit{PDL}(\textit{DL}(\textit{KB}))$.
In other words, we assume that the DL clauses, from which the given DL
program originates, are of the form shown in
Figure~\ref{figure:dlclauses}. 

Regarding the issue of equality
predicates, this implies that the body of a Horn clause can contain no
equality goals, only inequalities. This is because the DL clauses in
Figure~\ref{figure:dlclauses} include equality literals but no
inequalities, and in the process of building contrapositives the
former become inequality goals. The binary-first body ordering ensures that
these inequality goals are invoked only when ground, and thus --
taking into account the UNA principle -- they can be implemented using
the \verb+\+\texttt{==\,/2} built-in Prolog predicate.

\subsection{Principles of optimisation}\label{subsection:princ_optim}

The process of optimisation is summarised in
Figure~\ref{figure:process}. As the very first step we do filtering:
we remove those clauses that need not to be included in the final
program as they are never used in the execution (see
Section~\ref{preprocessing}).

Next, we classify the remaining predicates (see
Section~\ref{classification}). This information is used in subsequent
optimisations to make the code generated from a specific class of
predicates more efficient. The first two optimisations are global, in
the sense that e.g.\ the removal a clause during
filtering  requires the examination of  other parts of the knowledge base.

\begin{figure}[htbp]
  \centering
  \psfrag{rol}{roles}
  \psfrag{pi1}{filtering}
  \psfrag{mi1}{classifying}
  \psfrag{o}{ordering}
  \psfrag{o1}{indexing}
  \psfrag{o2}{ground g.}
  \psfrag{o21}{optimi-}
  \psfrag{o22}{sation}
  \psfrag{o3}{decom-}
  \psfrag{o31}{position}
  \psfrag{o4}{project-}
  \psfrag{o41}{ion}
  \psfrag{gr0}{}
  \psfrag{gr1}{}
  \psfrag{pr}{filtering}
  \psfrag{cl}{classification}

  \includegraphics[scale=0.58]{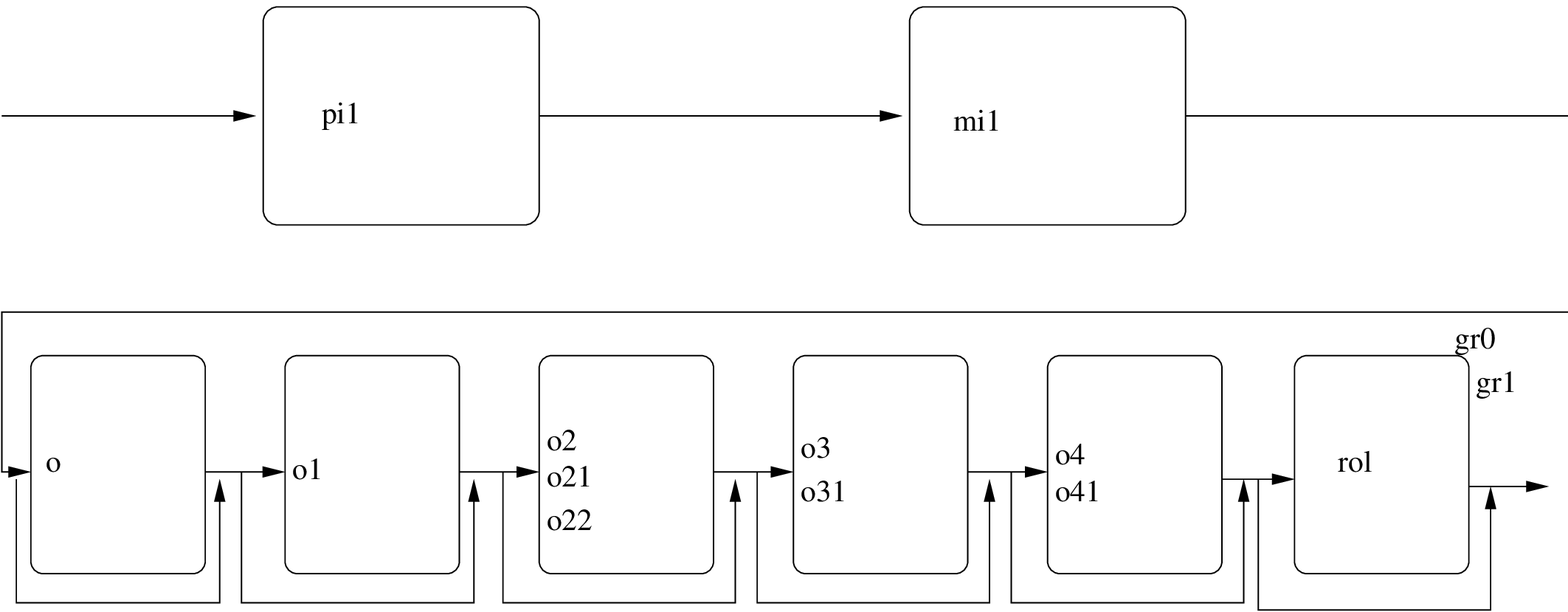}
  \caption{The process of optimisation.}
  \label{figure:process}
\end{figure}

Classification is followed by a sequence of further 
optimisations. Most of these are \emph{local} in the sense
that they concern only a  part of the program, e.g.\  a single
predicate. The optimisations are independent from each other: any
combination of these can be used when generating the final Prolog
program (cf.\ the arrows in Figure~\ref{figure:process}). These
optimisations are summarised below. Note that some further, lower
level optimisations are described in Section~\ref{further}.

\begin{enumerate}
\item[(o1)] ordering of goals in clause bodies (Section~\ref{ordering})
\item[(o2)] support for multiple argument indexing (Section~\ref{indexing})
\item[(o3)] efficient ground goal execution (Section~\ref{ground})
\item[(o4)] decomposition of clause bodies (Section~\ref{decomposition})
\item[(o5)] projection for eliminating multiple answers (Section~\ref{projection})
\item[(o6)] efficient translation of roles and their inverses (Section~\ref{roles})
\end{enumerate}

All above optimisations except for (o2) and (o6) concern
unary predicates. Therefore in the sections corresponding to (o1)
and (o3)--(o5) we implicitly assume that all clauses discussed belong
to unary predicates.

Before going into details we introduce some definitions regarding DL
predicates, to be used in the upcoming sections. Note that a predicate is
referred to by its functor or, if the arity is known from the context, by
its name.

\begin{definition}[Predicate reachability]
A predicate $P_1$ \emph{directly calls} predicate $P_2$ if $P_2$ is
invoked in any of the clauses of $P_1$. It is possible to \emph{reach}
$P_2$ from $P_1$ if (1) either $P_1$ directly calls $P_2$ or (2) there
exists a predicate $T$ from which it is possible to reach $P_2$ and
$P_1$ directly calls $T$.
\end{definition}

\noindent Thus, the relation \emph{reach} is the transitive closure of the relation
\emph{directly calls} between predicates. As an example, let us
consider the knowledge base in
Figure~\ref{figure:dl_predicates_iocaste}. Here predicate
\texttt{not\_Ans/1} is reachable from \texttt{Ans/1}, although it is
not directly called. The definition of reachability can
naturally be reformulated for clauses:

\begin{definition}[Reachability of clauses]
A predicate $P_2$ is reachable from a clause $C_1$, if $C_1$ invokes a
predicate $T$, such that $P_2$ is reachable 
from, or identical to, predicate $T$. A clause $C_2$, belonging to a predicate $P_2$, is
reachable from a clause $C$ (predicate $P$), if the predicate $P_2$ is
reachable from the clause $C$ (predicate $P$).
\end{definition}

\begin{definition}[Properties of DL predicates]
A predicate $P$ is recursive if it is reachable from itself.  We speak
about negative recursion if $P$ is reachable from $\textit{not\_}P$, or
vice versa \cite{przymusinski94wellfounded}. We refine this notion further
by saying that $P$ is an ANR (ancestor negative recursion) predicate, if
$P$ can occur as an ancestor of $\textit{not\_}P$ (i.e.\ the latter is
reachable from the former). Furthermore, $P$ is said to be a DNR
(descendant negative recursion) predicate, if $P$ can become a descendant
of $\textit{not\_}P$ (i.e.\ the former is reachable from the latter).
\end{definition}

\noindent Obviously $P$ is ANR if, and only if, $\textit{not\_}P$ is DNR
(as $\textit{not\_not\_}P$ is $P$). 

Using the above definitions, each DL predicate is classified into one of the
following groups.

\begin{enumerate}
\item A predicate $P$ is \emph{atomic} if all its clauses are ground and
  have empty bodies. Atomic predicates correspond to sets of ABox
  assertions.  Examples for atomic predicates are \texttt{Clever/1},
  \texttt{Pretty/1} and \texttt{hasChild/2} in the Happy and Iocaste DL
  programs.

\item $P$ is a \emph{query predicate} if it is not atomic and it
 satisfies the following three conditions:

  \begin{enumerate}
    \renewcommand{\theenumii}{\roman{enumii}}
    \item $P$ is not recursive;
    \item $P$ is not reachable from $\textit{not\_}P$ (i.e.\ $P$ is not DNR);
    \item all predicates  invoked within the clauses of $P$
      are either atomic or query predicates.
  \end{enumerate}

   Query predicates can be thought of as database queries.  They can be
   defined in terms of atomic predicates using
   conjunction and disjunction only. Thus the  execution of query predicates
   does not require  any special features, such as  keeping track of
   ancestors.

 An example of a query predicate is \texttt{Happy/1} in
 Figure~\ref{figure:dl_predicates_happy}.

\item A predicate is an \emph{orphan} predicate if it has an invocation
  in a clause body (which is called \emph{orphan goal} or
  \emph{orphan call}), but it does not appear in the head of any of
  the clauses. Orphan goals can succeed  only by ancestor resolution.

  Examples include predicates \texttt{not\_Ans/1} and
  \texttt{not\_Happy/1} in Figures~\ref{figure:dl_predicates_iocaste}
  and \ref{figure:dl_predicates_happy}. 

\item Finally, a predicate $P$ is a \emph{general} predicate if it is
  neither atomic, nor query, nor orphan. A general predicate $P$ can
  be further classified into subgroups based on whether $P$ is
  recursive, is of type ANR or DNR. The general predicates in the
  Iocaste knowledge base (Figure~\ref{figure:dl_predicates_iocaste}) are the following: \texttt{Ans/1} (not
  recursive, not DNR, ANR), \texttt{Patri\-cide/1} (recursive, not
  DNR, not ANR) and \texttt{not\_Patricide/1} (recursive, not DNR, not
  ANR).
\end{enumerate}

\subsection{Filtering}
\label{preprocessing}

Filtering removes those clauses of the DL predicates which
are not required in producing solutions.

\begin{definition}[Eliminable clauses]
A clause $C$ is called \emph{eliminable} in a DL
program $\textit{DP}$, if the body
of $C$ always fails in the execution of an arbitrary goal in $\textit{DP}$.
\end{definition}

\noindent Obviously, eliminable clauses can be removed from a DL program
without changing its behaviour. The set of Prolog clauses obtained this way
will still be called a DL program, but sometimes we will use the term
\emph{full DL program} to refer to the DL program before any clauses have
been removed. We now proceed to discuss special types of eliminable clauses.

\begin{definition}[False-orphan clauses]
\label{def:falseorphan}
Let $C$ be a clause of a predicate $P$ in the DL program $\textit{DP}$.  $C$ is said
to have the \emph{false-orphan} property w.r.t.\ $\textit{DP}$, if the body of $C$ invokes an
orphan predicate $O$, $P \neq \textit{not\_}O$, and it is not possible
to reach $P$ (and thus $C$) from $\textit{not\_}O$ in $\textit{DP}$. In this case $O$ is called 
a false-orphan goal in $C$.
\end{definition}

\begin{proposition}
\label{prop:falseorphan}
A clause $C$ having the false-orphan property in $\textit{DP}$ is eliminable in $\textit{DP}$.
\end{proposition}

\begin{proof}\small
Let $C$ be a clause of a predicate $P$, containing a false-orphan goal $O$.
By definition, it is not possible to reach $P$ from $\textit{not\_}O$, and
$P \neq \textit{not\_}O$. These two conditions imply that the ancestor list
supplied to $O$ contains no elements with the functor of
$\textit{not\_}O$. 

As an invocation of $O$ can only succeed by ancestor resolution, the
invocation of $O$ fails, and so clause $C$ can never succeed.
\end{proof}

Consider a clause $C$, being the only clause of predicate $P$, which
is removed because  it has the false-orphan  property defined above. At this point $P$
becomes an orphan predicate and some of the clauses invoking $P$ may thus
become eliminable, causing new orphan predicates to appear, and potentially
giving rise to further clauses with the false-orphan  property. 

Let us now define two further kinds of
clauses, which later will be shown to be eliminable.
\begin{definition}[Two-orphan clauses]
Let $C$ be a clause in the DL program $\textit{DP}$. $C$ is said to have the
\emph{two-orphan} property w.r.t.\ $\textit{DP}$, if the body of $C$ invokes
predicates $O_1$ and $O_2$, which are orphans in $\textit{DP}$, and which have
different functors.
\end{definition}

\begin{definition}[Contra-two-orphan clauses]
Let $C$ be a clause of a DL program $\textit{DP}$. Let $O_1$, $O_2$, and $O_3$ be
orphan predicates in $\textit{DP}$, where  $O_1 \neq O_2$ . The clause $C$ is said to have the
\emph{contra-two-orphan} property w.r.t.\ $\textit{DP}$, if the head of $C$ is of the form
$\textit{not\_}O_1(X)$, and the body of $C$ contains the goals
$O_2(Y)$ and $\textit{not\_}O_3(Z)$, where $X$, $Y$ and $Z$ are not
necessarily distinct variables or constants.
\end{definition}

\noindent A clause having the contra-two-orphan property is a 
contrapositive of a specific two-orphan clause, hence the naming of the property. 

For the next two propositions by a \emph{clause of interest} we mean a
clause having the two-orphan or the contra-two-orphan property.  We will prove
that a clause of interest cannot participate in a successful execution and
hence can be eliminated. Furthermore, we show iteratively that clauses that
become clauses of interest due to the elimination of other clauses of
interest are eliminable, too. Therefore the next proposition speaks about a
DL program in which we have already eliminated some clauses of interest
(initially zero clauses).

\begin{proposition}
\label{prop:orphans:main}
Let $\textit{DP}$ be a DL program obtained from a full DL program $DP_0$ by first
removing zero or more clauses of interest, and next repeatedly eliminating the clauses with the
false-orphan property, as long as possible.  If $O$ is an orphan
predicate in $\textit{DP}$, and a clause $C$ in $\textit{DP}$  invokes the predicate
$\textit{not\_}O$, then $C$ has either the two-orphan or the
contra-two-orphan property.
\end{proposition}
\begin{proof}\small
We can assume that the   clause $C\in DP$ is of the form `$H(Y)\texttt{ :- }
\textit{not\_}O(X), G_1,\ldots, G_n$', $n\geq 0$. Consider  
the following clause $C'$: `$O(X)\texttt{ :- }\textit{not\_}H(Y), G_1,\ldots
G_n$', which is a contrapositive of the same DL clause as $C$
is. Therefore $C'$ had to be present in the full DL program $DP_0$. However,
$C'$ is not present in $\textit{DP}$, because it belongs to the predicate $O$,
which is an orphan in $\textit{DP}$. Thus clause $C'$ was removed at some point.
Let $DP'$ be the last DL program in which $C'$ is present, i.e.\ $DP' \supseteq DP \cup
\{C'\}$ and $C'$ has one of the three
orphan-related properties introduced above, which justify its removal from $DP'$.

Let us first discuss if there can be any false-orphan goals in $C'$
w.r.t.\ the program $DP'$.  Because $O$ is an orphan predicate in $\textit{DP}$,
there is a clause $D$ in $\textit{DP}$ which calls $O$, and because there are no
false-orphans in $\textit{DP}$, this clause is reachable from
$\textit{not\_}O$. Thus in $DP'$, which contains the clause $C'$
belonging to the predicate $O$, all goals in the body of $C'$ are
reachable from $\textit{not\_}O$ through clause $D$. Consequently, all
these goals are also reachable from the clause $C$ of predicate $H$, as $C$
contains the goal $\textit{not\_}O(X)$. This means that the first goal in
the body of $C'$, $\textit{not\_}H(Y)$, cannot be a false-orphan in $DP'$,
because it is reachable from its negation, $H$. Consider now a goal $G_i$,
$0<i \leq n$, and assume that it is an orphan goal in $DP'$. Because $DP' \supseteq
DP$, $G_i$ is an orphan goal in $\textit{DP}$, too. As $G_i$ is present in the body
of $C$, and there are no false-orphans in $\textit{DP}$, $C$ is reachable from
$\textit{not\_}G_i$ in $\textit{DP}$. But then, in $DP'$, $G_i$ in $C'$ is also
reachable from $\textit{not\_}G_i$, as $C'$ is reachable from $C$, which,
in turn, is reachable from $\textit{not\_}G_i$.  Thus $G_i$ in $C'$ is not
a false orphan in the DL program $DP'$.  We have thus shown that the clause
$C'$ does not have the false-orphan property w.r.t.\ $DP'$.

Next, assume that $C'$ has the contra-two-orphan property in $DP'$. This implies
that the head of $C'$ is the negation of an orphan, i.e.\
$\textit{not\_}O$ is an orphan in $DP'$. Again because $DP' \supseteq
DP$, $\textit{not\_}O$ is an orphan predicate in $\textit{DP}$, and thus has no
clauses. This is a contradiction with the fact stated above that a goal
with the functor $O$ occurs in a clause $D$ \emph{reachable} from
$\textit{not\_}O$. 

This means that $C'$ has the two-orphan property in $DP'$. We now consider
two cases. If there are two orphan goals with different functors in the set
$\{G_i|0<i \leq n\}$, then clause $C$ has the two-orphan property, as
well. Otherwise, $\textit{not\_}H(Y)$ has to be an orphan goal, and
there has to be another orphan goal, with a different functor, amongst the
$G_i$'s, say $G_k$.  Let $O_1 = \textit{not\_}H$ and let $O_2$ and $Z$
denote the name and argument of the goal $G_k$ (i.e.\ $G_k =
O_2(Z)$). Using this notation the head of $C$ is of the form
$\textit{not\_}O_1(Y)$, while its body contains the goals
$O_2(Z)$ and $\textit{not\_}O(X)$, where $O_1 \neq O_2$ holds. Thus $C$
satisfies the contra-two-orphan property.
\end{proof}

\begin{proposition}
\label{prop:interest:eliminable}
Let $\textit{DP}$ be a DL program
obtained from a full DL program by removing some clauses of interest and/or
some clauses having the false-orphan property.  Any clause having the
two-orphan property or the contra-two-orphan property in $\textit{DP}$ is eliminable.
\end{proposition}

\begin{proof}\small
Let $DP'$ be the DL program
obtained from $\textit{DP}$ by repeatedly eliminating the clauses with the
false-orphan property as long as possible.

Let us indirectly assume that there is a successful execution path in $\textit{DP}$,
which uses a clause of interest, and of these consider the one used
earliest, say $C$. As a clause having the false-orphan property cannot be
part of a successful execution, the path in question is a valid path in
$DP'$.

Let us first assume that $C$ is a two-orphan clause. For this clause to
succeed, the two orphans with different functors require two different
ancestor goals, which are their negations, i.e.\ negations of orphans. One
of the ancestors can come from the query goal, but the other has to be
present in an earlier clause. Because of
Proposition~\ref{prop:orphans:main}, this is a clause of interest, which
contradicts the fact that $C$ is the earliest such clause in the
execution. 

Next, assume that $C$ has the contra-two-orphan property. In this
case $C$ has to be the very first clause used, because if there were a
preceding clause $C_0$, then $C_0$ would have to contain a negated orphan
goal (as the head of $C$ is a negated orphan), and so, again due to
Proposition~\ref{prop:orphans:main}, $C_0$ would be a clause of interest.
As $C$ is the first clause called, the ancestor list supplied to the goals
in its body contains only the head of $C$.
However, a contra-two-orphan clause contains an orphan goal whose functor is
different from that of the negated clause head, and so this orphan goal
fails, contradicting our initial assumption.
\end{proof}
\noindent We note that a clause containing several orphan goals, all with the
same functor cannot be eliminated, as all these goals can succeed by
resolving against a single ancestor in the ancestor list.

Proposition \ref{prop:interest:eliminable} makes it possible to iteratively
remove all three kinds of eliminable clauses introduced.
This process
terminates when there are no clauses with the above properties:

\begin{definition}[Filtered DL programs]
A DL program is \emph{filtered} if there are no clauses in the program
which have the false-orphan property, or the two-orphan property, or the
contra-two-orphan property.
\end{definition}

\begin{proposition}
\label{prop:orph:not:called}
Let $\textit{DP}$ be a filtered DL program. If $O$ is an orphan predicate in $\textit{DP}$,
then no clause in $\textit{DP}$ can contain a goal which invokes the predicate
$\textit{not\_}O$.
\end{proposition}

\begin{proof}\small
This is a simple consequence of Proposition~\ref{prop:orphans:main}.
\end{proof}

\noindent This means that in a filtered program an orphan goal can 
succeed only if the initial query predicate is its negation (cf.\ the
\texttt{not\_Ans/1} orphan goal in
Figure~\ref{figure:naive_translation_iocaste}). Another consequence of
Proposition~\ref{prop:orph:not:called} is that ancestor
resolution for orphan predicates is deterministic. We have proved
earlier that ancestor resolution in general is deterministic
(cf.\ Proposition~\ref{prop:nounification}), but for this we had to
assume that the binary-first rule was applied. Note that this
assumption is not needed now.

\paragraph{Implementation}

Our first optimisation is to transform the DL program into an
equivalent filtered form. To obtain a filtered program we use an
iterative process. Here we start from the initial DL program and we
eliminate as many clauses as we can. However, if we successfully
eliminated the last remaining clause of at least one predicate, then
we start the whole process again. We do as many iterations as needed
to reach a fixpoint, i.e.\ to have a set of clauses from which we
cannot eliminate any more clauses.

\paragraph{Example}

As an example for filtering, let us consider the DL program of the
Happy problem presented in Figure~\ref{figure:dl_predicates_happy}. In
the first iteration, we can eliminate clauses in lines 4--5 and 7--8
as they invoke the orphan goal \texttt{not\_Happy(B)}, and there is no
way to reach these clauses from predicate \texttt{Happy/1}.

As these were the last clauses of their corresponding predicates,
\texttt{not\_Clever/1} and \texttt{not\_Pretty/1} have actually become
orphans. Therefore, we apply one more iteration. Now we cannot
eliminate anything else: we have reached a fixpoint, containing a
single TBox clause in lines 1--2 (and the ABox facts in line 10).

\subsection{Classification}
\label{classification}

Within the filtered DL program we distinguish between different groups
of predicates based on their properties. This classification is useful when
generating the Prolog programs as it provides 
guidelines for what to generate and it also serves as a basis for
further optimisations. As discussed in
Section~\ref{subsection:princ_optim}, we distinguish between atomic,
orphan, query and general predicates.

A predicate $P$ is classified as atomic or orphan simply by checking
whether the specified condition holds for $P$.

However, to determine the set of query predicates, we use an iterative
process similar to the one used in filtering
(cf.\ Section~\ref{preprocessing}). The idea is that we iterate as
long as we find at least one new query predicate. We note that
we actually use a single iterative process which encapsulates
filtering as well as query predicate classification.

All the remaining predicates are classified as general predicates.

\paragraph{Use of classification information} 

Having classified the predicates of a DL program we can apply specific
compilation schemes for certain classes. We now examine
each of the predicate classes.

\begin{itemize}
\item \emph{atomic predicates}: Atomic predicates directly correspond
  to tables in a data\-ba\-se and thus their translation does not require an
  extra argument for the ancestor list.

\item \emph{query predicates}: The conditions in the definition
  guarantee that in the case of a query predicate $P$, we (i) do not
  need to check for loops, (ii) do not need to apply ancestor
  resolution, and (iii) do not need to pass the ancestor list to any
  of the goals in the body of $P$.

  Consequently, the code for query predicates does not require the
  additional argument for the ancestor list, similarly to atomic
  predicates.

\item \emph{orphan predicates}: The translation of an orphan predicate is a
  predicate consisting of a single clause of type $F_2$
  (cf.\ Definition~\ref{def:comp}). When an orphan predicate is invoked, a
  small optimisation can be applied regarding the ancestor list
  argument. The body of an orphan predicate contains nothing but an
  ancestor check: if the ancestor list contains the negation of the orphan
  predicate, then it succeeds. Now, unless the orphan predicate is invoked
  from within its negation, the ancestor list passed to it need not include
  the parent goal, i.e.\ the predicate from which it is invoked. This means
  that the ancestor list argument can be the same as the one in the parent
  goal.  This is the case, for example, in lines 10 and 16 in
  Figure~\ref{figure:naive_translation_iocaste}. Thus in these two lines
  the second argument of the orphan goal \texttt{not\_Ans}, namely variable
  \texttt{C}, can be replaced by the variable \texttt{B}.

\item \emph{general predicates}: We need to generate loop tests only
  for recursive, and ancestor tests only for DNR predicates. Updating
  the ancestor list is only required for ANR predicates.
\end{itemize}

\paragraph{Examples} 

We now discuss some examples of how the Prolog code can be simplified due
to classification. The DL predicate \texttt{Happy/1} in
Figure~\ref{figure:dl_predicates_happy} is classified as a query
predicate. Having removed DL predicates \texttt{not\_Clever/1} and
\texttt{not\_Pretty/1} in the filtering step, we can further simplify the
Happy program by removing the ancestor list arguments. This results in
the code shown in Figure~\ref{figure:happy_preprocessed}. Note that we
have actually obtained the hand-made translation for the Happy
problem (see Figure~\ref{figure:happystraightforwardtranslation}).

\begin{figure}[htbp]
\begin{Verbatim}[numbers=left,numbersep=2pt,frame=single]
Happy(A) :- hasChild(A, B), hasChild(B, C), hasChild(B, D),
            Clever(C), Pretty(D).

Clever(lisa). Pretty(lisa). hasChild(kate, bob). hasChild(bob, lisa).
\end{Verbatim}
\caption{The Happy program after filtering and classification.}
\label{figure:happy_preprocessed}
\end{figure}

We also note that in the case of the Iocaste program in
Figure~\ref{figure:naive_translation_iocaste}, classification directly
results in omitting lines 1, 2, 7 and 13. Lines 1 and 2 can be removed
because predicate \texttt{Ans/1} is classified as a non-recursive
non-DNR general predicate. Ancestor tests in lines 7 and 13 can be
omitted as \texttt{Patricide/1} and \texttt{not\_Patricide/1} are
non-DNR predicates.

\subsection{Body ordering}
\label{ordering}

An important optimisation is to order the goals in the body of the
generated clauses so as to minimise the execution time. This is a
generic idea used in some form or other by many systems. For example,
in the case of databases, query optimisation is an essential task, as
without it one would not be able to answer complex queries
\cite{databasequery}. Query optimisation is similarly important when
querying non-relational information sources, such as XML
\cite{xquery}.

Query optimisation often relies on statistical information, such as
the size of database tables or the number of distinct values in a
given column. In the present work we do not take into account
such information and so we restrict our attention to optimisations
which consider only the TBox part of the DL programs.

Prolog systems also use body reordering. For example, the Mercury
compiler reorders the conjunctions in clauses for more efficient
execution \cite{mercury}. Body reordering, instantiation analysis and
related techniques are used by many parallel systems as well. For
example, in the Andorra system \cite{andorra} the deterministic goals
in a clause are moved to the front.

In our case we have very special clauses to work with, as described in
Section~\ref{dlclauses}. This allows us to use a simple, specialised
ordering technique.

Below we first propose a possible ranking between the different kinds
of goals in a body. This ranking uses heuristics applicable for DL
programs. Next, we introduce the simple algorithm we use for body
ordering. Note that this algorithm is actually only the first step, as
it forms the basis of a more complex body restructuring technique
described in Section~\ref{decomposition}.

\subsubsection{Ranking of goals}

Let us start with considering some principles for ranking. 

\begin{itemize}
\item Atomic and query predicates can be answered by using ABox facts
  only, i.e.\ they correspond to (maybe complex) database queries. 

\item General predicates, such as \texttt{Patricide/2} for example,
  may require complex, possibly recursive, execution on the Prolog
  side.
\end{itemize}

\noindent These considerations lead to some heuristics which are summarised below.

\begin{itemize}
\item We invoke atomic and query predicates before
general predicates.

\item We prefer to invoke a ground role predicate at a given point, instead
  of a role predicate with one or two uninstantiated variables. The former 
  simply checks whether a relation holds between two individuals. The
  latter enumerates all possible pairs of individuals for which the given
  relation holds, leaving a possibly huge choice point behind.

\item Given two role predicates with potentially uninstantiated
  variables we prefer to invoke first the one having the
  \emph{head variable}, i.e.\ the variable in the head of a clause, as its
  argument. The main justification for this is that the head variable may
  actually be instantiated, which is not the case for any other variable.
\end{itemize}

A further issue to discuss is the place of the orphan goals within a
body. Recall that orphan goals can only succeed by ancestor resolution, and
only if their negation is the query goal. For example, the orphan goal
\texttt{not\_Ans(E, C)} in line 10 in
Figure~\ref{figure:naive_translation_iocaste} can only succeed if invoked
within an \texttt{Ans(X)} query goal. However, when an orphan goal
succeeds, its first argument (variable \texttt{E} above) may stay
uninstantiated.

These properties of orphan goals suggest to put them in the first
available place where they are ground. However, it also seems to be a
good idea to move an orphan goal to the very front of the body. This is
because orphan goals tend to fail very often: if they are placed at
the front, in the case of failure, we do not need to execute the rest
of the clause.

Note, however, that placing orphan goals at the front invalidates the
proof of Proposition~\ref{prop:nounification} on page
\pageref{prop:nounification}, as now the case (c2) can also
happen. Fortunately, Proposition~\ref{prop:orph:not:called} ensures that
ancestor resolution remains deterministic for DL programs, even if the
invocations of the orphan goals are moved to the front of a body.

Based on the above discussion, we have designed an appropriate ranking
order, called \emph{base ranking}, which is summarised in
Figure~\ref{figure:ranking}. Here we define 10 categories of goals and
give orphan goals the highest priority. Higher priority means earlier
placement in the body. If there are more goals within the same
category, the selection between them is unspecified, i.e.\ any of them
can be chosen. For example, if we have two non-ground atomic concepts,
either of these can come first.

Note that the base ranking ensures the binary-first rule introduced in
Definition~\ref{def:binfirst}, except for orphan goals. Furthermore,
part (b) of Proposition~\ref{prop:ground} ensures that the variables
occurring in inequalities will get instantiated, so we do not have to
deal with non-ground inequalities.

\begin{figure}[htbp]
\begin{magicbox}
\begin{enumerate}
\item orphan goal
\item ground inequality
\item ground role
\item ground atomic or query concept
\item role with 1 unbound variable
\item role with 2 unbound variables, but at least one of them is a head variable
\item role with 2 unbound variables
\item non-ground atomic or query concept
\item ground general concept
\item non-ground general concept
\end{enumerate}
\end{magicbox}
\caption{The base ranking of the different types of
  goals within a body.}
\label{figure:ranking}
\end{figure}

\subsubsection{The ordering algorithm}

In Figure~\ref{figure:ordering} we present a simple algorithm which orders
the body of a clause of a DL program. This algorithm has three inputs:
the body to be ordered ($B$), a pre-defined ranking of the different kinds
of goals ($R$) and an initial variable list ($V$) containing those
variables that are known to be instantiated at the beginning.

\begin{figure}[htbp]
\begin{magicbox}
      \begin{enumerate}
        \item input parameters: $B$, $R$, $V$, $i := 1$
        \item if $B = \emptyset$, exit with $G_1, \ldots, G_{i-1}$
        \item $G_\mathrm{i} := $ highest priority goal in $B$ according to ranking $R$ w.r.t.\ variables $V$
        \item $B := B \setminus \{G_\mathrm{i}\}$
        \item if $G_\mathrm{i}$ is non-orphan $V := V \cup
          \textrm{variables of goal } G_\mathrm{i}$
        \item $i := i+1$
        \item goto step 2
      \end{enumerate}
\end{magicbox}
\caption{The ordering algorithm used to optimise the execution of a body.}
\label{figure:ordering}
\end{figure}

The idea is to repeatedly select the highest priority goal from the remaining goals, and
place it in the ordered goal sequence forming the final body (see step
3). To be able to assess the groundness of arguments we keep track of
the set $V$ of variables instantiated so far. $V$ is initialised from
the input parameter (step 1) and is updated to include the variables
of the goal selected (step 5). Having selected a goal, we continue by
iteratively ordering the rest of the body (step 7).

As an example, reordering the body of the main clause of
\texttt{Patricide/2}  (cf.\ lines 9--10 in
Figure~\ref{figure:naive_translation_iocaste}) yields the clause presented in
Figure~\ref{figure:reordered_patricide}. Here the orphan call
\texttt{not\_Ans/2} is moved to the front. The second goal is a role predicate
containing a head variable. The third one is also a role predicate with at least
one variable instantiated: the instantiation state of variable \texttt{E} is
not known at compile-time. Finally, the last goal is a ground general
concept call.  To make the comparison of the original and the reordered
clauses easier, in  Figure~\ref{figure:reordered_patricide}  we keep the
variable names of the original clause.

\begin{figure}[htbp]
\begin{Verbatim}[numbers=left,numbersep=2pt,frame=single]
Patricide(A, B) :- C = [Patricide(A)|B], not_Ans(E, B), hasChild(D, A), 
                   hasChild(E, D), Patricide(D, C).
\end{Verbatim}
\caption{The reordered version of the main clause of \texttt{Patricide/2}.}
\label{figure:reordered_patricide}
\end{figure}

For another example let us consider clause \texttt{Happy/1} in
Figure~\ref{figure:happy_preprocessed}. Using body reordering on this
clause we get the clause presented in
Figure~\ref{figure:reordered_happy_clause}. Note that the goal
\texttt{Clever(C)} is now moved forward into the place where it first becomes ground.

\begin{figure}[htbp]
\begin{Verbatim}[numbers=left,numbersep=2pt,frame=single]
Happy(A) :- hasChild(A, B), 
            hasChild(B, C), Clever(C), 
            hasChild(B, D), Pretty(D).
\end{Verbatim}
\caption{The reordered version of the \texttt{Happy/1} clause.}
\label{figure:reordered_happy_clause}
\end{figure}

\subsection{Multiple argument indexing}
\label{indexing}

In this section we discuss a transformation  of role predicates which makes
their Prolog execution more efficient.

Notice that goal \texttt{has\_child(E, D)} in line 2 in
Figure~\ref{figure:reordered_patricide} is always called with the
second argument instantiated. If we use a database system to store the
content of the ABox this call is executed efficiently. This is
because databases can do indexing on every column of a table.
In most Prolog systems, however, indexing is done only on the first head
argument. This may raise performance issues if we use
Prolog for storing large amounts of ABox facts.

To achieve multiple argument indexing in the generated programs we do
the following. For each role $P$ we generate a new role
\texttt{idx\_$P$}. This new set of Prolog facts (called \emph{index
  predicate}) captures the inverse relation between the arguments of
$P$, i.e.\ \texttt{idx\_$P(X, Y)$} holds if, and only if, \texttt{$P(Y,
  X)$} holds. In the case of the Iocaste problem this effectively
means that we add the following index predicate to the generated
program:

\begin{small}
\begin{Verbatim}[numbers=left,numbersep=2pt,frame=single]
idx_hasChild(o, i). 
idx_hasChild(p, i). 
idx_hasChild(p, o). 
idx_hasChild(t, p).
\end{Verbatim}
\end{small}

Consider an invocation of a role $P$ where the second argument is
instantiated, but the first is (possibly) not.  We replace each such
invocation by a call of \texttt{idx\_$P$} with the two arguments
switched. For example, the ordered clause for \texttt{Patricide/2} in
Figure~\ref{figure:reordered_patricide} takes the following form (note that
the variable \texttt{E} cannot be assumed to be instantiated by the orphan call
\texttt{not\_Ans(E, B)}):

\begin{small}
\begin{Verbatim}[numbers=left,numbersep=2pt,frame=single]
Patricide(A, B) :- C = [Patricide(A)|B], not_Ans(E, B), hasChild(D, A), 
                   idx_hasChild(D, E), Patricide(D, C).
\end{Verbatim}
\end{small}

Note that we do not actually have to generate index predicates for
every role in the ABox, because using compile time analysis we can
identify those role predicates $P_1, \ldots, P_i$ that need indexing at all
(i.e.\ those which are called at least once in such a way that their second
argument is ground, but the first is possibly not).

Also note that most Prolog implementations create a choice point when
both arguments of a role predicate $P$ are instantiated, although it is obvious
that such invocations can only succeed once (as an ABox cannot
contain a given $P(i,j)$ axiom twice). For example, consider the goals
\texttt{hasChild(i, o)} or \texttt{idx\_hasChild(p, i)}.

To avoid these choice points we apply the
commonly known technique of using auxiliary predicates. Namely, given
a role predicate $R/2$ (including the index predicates introduced above)
with facts $F$ we do the following. For every maximal set $D \subseteq
F$ of facts, which share their first argument we introduce a single
\emph{grouping clause} $R(A, Y)\texttt{ :- }T(Y)$. Here, $Y$ is a variable and $A$
is the constant shared by all 
of the facts in the first argument position in $D$. $T$ is the name of
a newly introduced predicate containing facts $T(Z_1), \ldots, T(Z_k)$
which correspond to the constants in the second arguments of the facts
in $D$, i.e.\ $\{Z_1, \ldots, Z_k\} = \{ B | R(A, B) \in D \}$.

As an example, we show the optimised version of the four clauses of the
predicate \texttt{idx\_hasChild/2} introduced above. Here, line 2 contains
a grouping clause invoking the auxiliary predicate
\texttt{idx\_hasChild\_p/1}. This makes it possible for Prolog not to
create any choice points when invoking the goal \texttt{idx\_hasChild(p,
i)} or \texttt{idx\_hasChild(p, o)}.

\begin{small}
\begin{Verbatim}[numbers=left,numbersep=2pt,frame=single]
idx_hasChild(o, i). 
idx_hasChild(p, Y) :- idx_hasChild_p(Y).
idx_hasChild(t, p).

idx_hasChild_p(i).
idx_hasChild_p(o).
\end{Verbatim}
\end{small}

\subsection{Ground goal optimisation}
\label{ground}

An important optimisation step is to make sure that the truth value of
ground goals, i.e.\ goals with all arguments instantiated, is
calculated only once. Note that by default this behaviour is not
provided by the Prolog execution, but is supported, for example, by Mercury
\cite{mercury}.

To achieve this, we duplicate a general or query predicate $P$,
i.e.\ we create two versions of $P$ depending on whether we assume
that the head variable is instantiated or not. These variants of $P$
are called non-deterministic (\emph{nondet}) and deterministic
(\emph{det}) variants, respectively.

We also create a \emph{choice} predicate for the general case which
checks if the head variable is ground at runtime. This predicate then
calls either the \emph{nondet} or the \emph{det} variant of predicate $P$.

The differences between the two variants of a predicate $P$ are the
following:

\begin{enumerate}
\item We place a Prolog cut (denoted by \texttt{!}) at the end of each clause
  in the \emph{det} variant. This results in pruning the rest of the
  search space after a successful execution of the \emph{det} variant.
\item We order the body of the clauses in the \emph{det} variant based
  on the assumption that the head variable $H$ is instantiated
  (i.e.\ the ordering algorithm in Figure~\ref{figure:ranking} is
  executed with the initial variable list $V = \{H\}$).
\end{enumerate}

Finally, we transform every goal in the program calling a general or
query predicate $P$ into another goal which calls \texttt{choice\_$P$}
instead. This technique is illustrated in
Figure~\ref{figure:ground_goal_handling}.

\begin{figure}[htbp]
\begin{Verbatim}[numbers=left,numbersep=2pt,frame=single]
choice_Patricide(A, B) :-
        (   nonvar(A) -> det_Patricide(A, B)
        ;   nondet_Patricide(A, B)
        ).

nondet_Patricide(A, B) :- member(C, B), C == Patricide(A), !, fail.
nondet_Patricide(A, _) :- Patricide(A).
nondet_Patricide(A, B) :- C=[Patricide(A)|B], not_Ans(D,B), hasChild(E,A), 
                          idx_hasChild(E, D), det_Patricide(E, C).

det_Patricide(A, B) :- member(C, B), C == Patricide(A), !, fail.
det_Patricide(A, _) :- Patricide(A), !.
det_Patricide(A, B) :- C=[Patricide(A)|B],not_Ans(D,B), idx_hasChild(A,E), 
                       idx_hasChild(E, D), det_Patricide(E, C), !.
\end{Verbatim}
\caption{The two variants of predicate \texttt{Patricide/2}.}
\label{figure:ground_goal_handling}
\end{figure}

In lines 10 and 16, instead of \texttt{choice\_Patricide/2},
we directly invoke predicate \texttt{det\_Patricide/2}. This is a
further optimisation step. Namely, we can directly call the \emph{det}
variant of a predicate $P$ if we know already at compile-time that the
first argument of the specific invocation of $P$ is ground. In our
case we can be sure that variable \texttt{E} is instantiated at the
time of calling \texttt{det\_Patricide/2}, because the predicate call
\texttt{idx\_hasChild(E, D)} instantiates it. 

Note the difference between the body goals in the two variants in
Figure~\ref{figure:ground_goal_handling} (lines 9--10 and 15--16). In
the \emph{det} variant we assume that variable \texttt{A} is
instantiated at call time, therefore we use the \texttt{idx} variant
of the goal \texttt{hasChild(E, A)}.

Proposition~\ref{prop:anc} ensures that all unary goals within a \emph{det}
variant of a predicate are ground at the time of their invocation. Thus all
these goals will directly invoke the \emph{det} variant of their predicate.

\subsection{Decomposition}
\label{decomposition}

The goal of decomposition is to split a body into independent
components. This is achieved by uncovering the dependencies between
the goals of the body. This process, on one hand, introduces a higher
level body ordering, where the independent goal groups are ordered
first, and then the individual groups are split and ordered
recursively. More importantly, the discovery of independent components
makes it possible to use a generalisation of the ground goal
optimisation, by applying this technique to a whole independent goal
group. For DL programs generated from a $\mathcal{SHIQ}$ KB this
practically means recovering certain useful structural properties of
the initial TBox axioms. Before we go into details we show an example
to demonstrate a problem which can be solved using decomposition.

\subsubsection{An introductory example}

Let us recall the single clause of the predicate \texttt{Happy/1} shown in
Figure~\ref{figure:reordered_happy_clause}, stating that someone is
happy if she has a child having both a clever and a pretty child.
Although the body of this clause is ordered
according to our base ranking, in certain cases the execution of it is
far from optimal. For example, consider the ABox specified below:

\begin{Verbatim}[commandchars=\\\{\}]
hasChild(kate, bob).
hasChild(bob, lisa\(\sb{i}\)).     \textrm{for \(i=1\ldots n\)}
Clever(lisa\(\sb{i}\)).            \textrm{for \(i=1\ldots n\)}
\end{Verbatim}

\noindent Thus we know that \texttt{bob} is the child of \texttt{kate}
and he has $n$ clever children, but nobody is known to be
pretty. This ABox does not entail that
\texttt{kate} is happy, i.e.\ the  goal  \texttt{Happy(kate)} fails.
However, obtaining this negative answer involves lots of useless computation. Namely, we
enumerate all children of \texttt{bob} and check whether they are
clever. We do this in spite of the fact that \texttt{bob} has no
pretty children at all, even though having a pretty grandchild is a
necessary condition for \texttt{kate} being happy. What happens is
that we explore the choice point created in line 2 in
Figure~\ref{figure:reordered_happy_clause}, although goals in line 3
are bound to fail.

Note that this behaviour would not change if we applied ground goal
optimisation here, i.e.\ if we used the \emph{det} variant of the
clause (cf.\ Section~\ref{ground}) in
Figure~\ref{figure:reordered_happy_clause}. The order of the goals in
the body would be the same. The cut at the end of the clause would not
matter either, as the goal \texttt{Happy(kate)} fails.

What we need here is the realisation that the \texttt{hasChild(B, C),
  Clever(C)} group of subgoals, used for checking that \texttt{bob}
has a clever child, is independent from the remaining subgoals of the
body. Thus, once we have proved that \texttt{bob} has a clever child,
there is no point in proving this property in other ways.

\subsubsection{The solution}

In the above example, the real reason behind the inefficient execution
is that during the Prolog translation we do not utilise the structural
properties of the TBox axiom in Figure~\ref{figure:happy_kb}. This
axiom actually describes that somebody is happy if she has a child
satisfying a certain condition, namely having a clever child as
well as a pretty child. This condition can be split into two
independent parts: \texttt{hasChild(B, C), Clever(C)} and
\texttt{hasChild(B, D), Pretty(D)}. The two parts share only a single
variable \texttt{B}. If \texttt{B} is ground, we can stop
enumerating her children once a clever one is found, as a new
value for variable \texttt{C} can not affect the remaining goals.

The solution is to use this knowledge by \emph{decomposing} the body
of clause \texttt{Happy/1}, as shown in
Figure~\ref{figure:happy_clause_decomposed}.

\begin{figure}[htbp]
\begin{Verbatim}[numbers=left,numbersep=2pt,frame=single]
Happy(A) :-
     hasChild(A, B),
     (   hasChild(B, C),
         Clever(C) -> true
     ),
     (   hasChild(B, D),
         Pretty(D) -> true
     ).
\end{Verbatim}
\caption{The decomposed version of the \texttt{Happy/1} clause.}
\label{figure:happy_clause_decomposed}
\end{figure}

The clause for \texttt{Happy/1} starts with a single goal representing
the condition that somebody should have at least one child in order to
be happy (line 2). The required properties of this child are captured by the
two consecutive \emph{components} in the clause (lines 3--5 and
6--8). The idea here is that we only look for the first solution of
these components, i.e.\ we place an implicit Prolog cut at the end of the component (by
using the conditional expression operator \texttt{->}). This ensures
that once a component succeeds it prunes the rest of its search
space. This is, in fact, the ground goal optimisation, applied to a
whole component, rather than to a single goal.

Note that the goal \texttt{hasChild(A, B)} in line 2 generates a
choice point, which we cannot eliminate here as we cannot be sure that
\texttt{B} has the required properties. On the other hand, if the
ground goal optimisation (cf.\ Section~\ref{ground}) is also applied,
then the cut (\texttt{!}) at the end of the \emph{det} variant clause prunes
this choice point.

\subsubsection{The process of decomposition}

Decomposition relies on identifying independent components in clause
bodies, i.e.\ subgoal sequences which do not share uninstantiated
variables. Such techniques have been extensively studied, mostly in
the context of parallel execution of logic programs, for example in
\cite{manuel2}.

Because of the special properties of DL predicates we can apply here a very
simple algorithm. The decomposition process is actually a
modification of the ordering algorithm introduced in 
Figure~\ref{figure:ordering}: the steps 2a--2c  shown in
Figure~\ref{figure:decomposition_process} are added after step 2 of the
ordering algorithm.

\begin{figure}[htbp]
\begin{magicbox}
      \begin{enumerate}
        \item[2a.] split $B$ into a partition $\{B_1, \ldots, B_n\}$
          w.r.t.\ $V$
        \item[2b.] if $n = 1$ continue at 3
        \item[2c.] apply  the ordering algorithm recursively (starting from step 1)
          for $(B_j, R, V)$, producing a (possibly composite) goal $C_j$, for each $j = 1, \ldots, n$;
          and return $G_1, \ldots, G_{i-1}$,  $c(C_1), \ldots, c(C_n)$.

      \end{enumerate}
\end{magicbox}
\caption{The process of decomposition (extension of
  the algorithm in Figure~\ref{figure:ordering}.}
\label{figure:decomposition_process}
\end{figure}

Decomposition starts with step 2a, which partitions the set of body goals
into one or more subsets in such a way that goals in different partitions
share only variables in $V$ (the set of variables considered to be
instantiated) and the maximal number of partitions is obtained. If the
decomposition results in a single partition (see step 2b), then we continue
with the normal goal ordering algorithm.

If multiple partitions have been obtained then each of these is
ordered and recursively decomposed (step 2c). In this case the output
of the modified ordering algorithm contains the goals collected so
far, followed by the components. The latter are distinguished from
ordinary goals by being encapsulated in a $c(\ldots)$ structure. This
marks the independent units where pruning can be applied. 

Note that the components themselves also undergo an ordering
phase, but this is not detailed here.

We illustrate the idea of recursive decomposition on the \emph{nondet}
variant of clause \texttt{Ans/1} from the Iocaste problem. The result is
shown in Figure~\ref{figure:decomposition_iocaste}. The first evaluation of
the step 2a yields a single component. Therefore step 3 of the ordering
algorithm (Figure~\ref{figure:ordering}) is performed, the highest priority
goal is selected and placed at the beginning of the body (see line 3 of
Figure~\ref{figure:decomposition_iocaste}). Next, the process of
decomposition is repeated for the remaining goals, where the evaluation of
step 2a yields two components, shown in lines 4--7 and 8. As the second
component contains a single goal, there is no need for explicit pruning (as
the call of a \texttt{det\_\ldots} predicate leaves no choice points
behind).

\begin{figure}[htbp]
\begin{Verbatim}[numbers=left,numbersep=2pt,frame=single]
nondet_Ans(A, B) :-
        C=[Ans(A)|B],
        hasChild(A, D),
        (   hasChild(D, E),
            det_not_Patricide(E, C) ->
            true
        ),
        det_Patricide(D, C).
\end{Verbatim}
\caption{The \texttt{Ans/1} clause after decomposition.}
\label{figure:decomposition_iocaste}
\end{figure}

Also note that variables used for ancestor resolution in the generated
program are not considered during the decomposition process as this is
performed on the DL program directly. This is the reason why goals in
line 5 and line 8 can be placed into separate components, although
both of them contain variable \texttt{C}.

\subsection{Projection and supersets}
\label{projection}

As discussed in the previous section, decomposition helps in reducing the
number of unnecessary choice points in a clause body by using conditional
structures. However, the choice point in the first component 
of the \emph{nondet} variant of a clause has to remain, and this can cause
serious performance problems.

As an example, let us consider the behaviour of the clause
\texttt{nondet\_Ans/2}, shown in Figure~\ref{figure:decomposition_iocaste},
when run on a large Iocaste pattern. Here, the first component is the goal
\texttt{hasChild(A, D)} in line 3, which enumerates all objects in the
parent-child relationship. Let us assume that the first few facts in the
\texttt{hasChild/2} predicate are \texttt{hasChild(i, e}$_i$\texttt{)}, for
$i = 1, \ldots, k$, cf.\ the rightmost pattern in
Figure~\ref{figure:patterns}. Thus the goal \texttt{hasChild(A, D)} first
succeeds with the substitution \texttt{A = i, D = e$_1$}. As explained in
Section~\ref{motivex}, the remaining two components of
\texttt{nondet\_Ans/2} (lines 4--8 of
Figure~\ref{figure:decomposition_iocaste}) will complete successfully,
without leaving a choice point, and thus the solution \texttt{A = i} is
obtained. We now backtrack to the choice point in line 3, to look for other
individuals satisfying \texttt{nondet\_Ans/2}. However, the next few
substitutions returned by the \texttt{hasChild} goal in line 3 will be
\texttt{A = i, D = e$_i$}, $i = 2, \ldots, k$. In all these substitutions \texttt{A} obtains the
value \texttt{i}, which is already known to be a solution. Thus the exploration
of this part of the search space is absolutely useless. Having obtained a
solution \texttt{A = i}, one should ignore all further ABox facts of the form
\texttt{hasChild(i, \_)}. However, one cannot cut away the choice point in
line 3, because there could be other \texttt{hasChild(x, \_)} facts, which
lead to further solutions. Contrastingly, no such problem appears in the
\emph{det} version of the same predicate, as a cut is placed at the very
end of the clause (cf. the ground goal optimisation, Section~\ref{ground}).

We eliminate the need for the \emph{nondet} variant of a
predicates by the optimisation presented in this section. This works by
first calculating a so called superset of the predicate, which is a set of
individuals containing all the solutions of the predicate.  Next, the
elements of the superset are enumerated, and the \emph{det} variant of the
predicate is called for each individual in the superset. 

We now proceed with the definition of the notion of superset. Next, we
show how it can be used to eliminate the non-deterministic predicates
from the generated programs.

\subsubsection{The notion of superset}
Let $I(P)$ denote the set of solutions of a unary predicate (unary clause)
$P$ w.r.t.\ a Prolog program. By a solution of a clause $C$ we mean a
solution of the predicate which $C$ belongs to, obtained through the
successful execution of clause $C$.

\begin{definition}[The superset of a predicate]
A set of instances $S$ for which $I(P) \subseteq S$ holds is
called a \emph{superset} of predicate (clause) $P$.
\end{definition}

According to the definition, the superset of a predicate is a set of
instances which contains all the solutions of the predicate (and
possibly some other individuals as well). For example, the set of
individuals $\{\verb+i+,\verb+o+,\verb+p+\}$ forms a superset of predicate
\texttt{Ans/1}, as it contains the individual \texttt{i}. 

Now, given a predicate $P$ and one of its supersets $S$, we can
eliminate the \emph{nondet} variant of $P$ as follows: we create a new
predicate which invokes the \emph{det} variant of $P$ for each
individual $i \in S$. Technically, this logic can be built into the
\emph{choice} predicate, as exemplified below:

\begin{Verbatim}[numbers=left,numbersep=2pt,frame=single]
choice_Ans(A, B) :-
        (   nonvar(A) -> det_Ans(A, B)
        ;   member_of_superset_Ans(A), 
            det_Ans(A, B)
        ).
...
\end{Verbatim}

\noindent Here we call the \emph{det} variant directly if \texttt{A} is
instantiated (line 2). However, we also call the \emph{det} variant if
\texttt{A} is uninstantiated (line 4), following a goal (in line 3) which
enumerates the elements of the superset in the variable \texttt{A}.

This technique has an important property: it ensures that
each solution is returned only once. For example, invoking
\texttt{choice\_Ans(A, [])}  enumerates instance
\texttt{i} only once w.r.t.\ the usual Iocaste ABox.  The
above scheme can be used for supersets which \emph{do not fit into
  memory}: the Prolog goal \texttt{member\_of\_superset\_Ans(A)} can be
implemented as a database
invocation which enumerates the individuals in the superset.

We noted at the end of Section~\ref{ground} that all goals within a
\emph{det} variant themselves invoke the \emph{det} variant of their
predicate. Thus, if the projection optimisation is applied to all
predicates of a program, then the \emph{choice} predicates can only be
called from the conjunctive query. Such predicates are called \emph{entry
predicates}, and are known to have an empty ancestor list argument.

We now describe an algorithm which assigns a set of instances to a
predicate $P$, then we show that this set is actually a superset of
$P$.

\subsubsection{Calculating supersets}

Our goal is to find a method for building supersets for predicates such
that the supersets (1) do not contain too many non-solutions and (2) are
easy to calculate.

\begin{definition}[Projection of predicates]
The \emph{projection} of a role predicate $P$ with respect to its
$n$th ($n=1,2$) argument is $\textit{Pr}_n(P) = \{v_n | (v_1,v_2) \in
I(P)\}$. The projection of a concept predicate $C$ with respect to its
only argument is $\textit{Pr}_1(C) = I(C)$. If $G$ is a goal,
$\textit{Pr}_n(G)$ means the projection of the predicate
invoked by the goal $G$ w.r.t.\ its $n$th argument.
\end{definition}

For example, $\textit{Pr}_1(\texttt{hasChild(A,B)})$ w.r.t.\ the usual
Iocaste knowledge base is the set  $\{\verb+i+,\verb+o+,\verb+p+\}$, excluding
\texttt{t}, as \texttt{t} has no children. Note that this projection
can be calculated by the Prolog call
\verb+setof(A, B^hasChild(A, B), R)+.

We now introduce the notion of \emph{projected label} for clauses. This is
either a superset of the clause, or the functor of a predicate whose
solution set contains all solutions of the given clause. Within this definition we use
a refinement of the notion of DNR predicate: a call of a predicate $Q$ in
the body of predicate $P$ is said to be a DNR invocation, if $P$ is
reachable from $\textit{not}\_Q$.

\begin{definition}[Projected label]
\label{def:proj_label}
Let $C$ be a unary clause in a DL program $\textit{DP}$. Let $W$ be the set of all
atomic and query goals in $C$ which contain the head variable. We define
the \emph{projected label} of $C$, denoted by $\textit{Pl}(C)$, as follows.

If $C$ is a fact of the form $C(a)$ then $\textit{Pl}(C)$ is the set
$\{a\}$. 
Otherwise, if $W \neq \emptyset$, then $\textit{Pl}(C)$ is
calculated as the intersection of the projections of the goals in $W$
w.r.t.\ the head variable, i.e.\ $\textit{Pl}(C) = \mathop{\cap}_{G_i
  \in W} \textit{Pr}_{p_i}(G_i)$, where $p_i$ is the position of the
head variable in the goal $G_i$.

If $W = \emptyset$ and $C$ contains a goal which is not a DNR invocation, 
then $\textit{Pl}(C)$ is the functor of an arbitrary such
goal (e.g.\ the one which  comes first w.r.t.\ the standard
Prolog term ordering). If all goals in $C$ are DNR invocations, then
$\textit{Pl}(C)$ is the set of all individual names in $\textit{DP}$.
\end{definition}

\noindent The notion of projected label has an interesting invariant:
$I'(\textit{Pl}(C)) \sqsupseteq I(C)$, where the function $I'$ takes either
a functor of a predicate $P$ and maps it to $I(P)$, or takes an arbitrary set
$S$ and maps it to itself. The invariant states that the ``solution set''
of the projected label of a clause $C$ contains all the solutions of
$C$. This invariant can be easily checked by going through the four cases
of the definition.

As an example, let us consider the Iocaste program presented
in Figure~\ref{figure:dl_predicates_iocaste}. The projected label of
the clause in lines 1--2 is the set $\textit{Pr}_1(\texttt{hasChild(A,
  B)})$, while the projected label of the clause
\texttt{not\_Patricide/1} in lines 7--8 is the intersection
$\textit{Pr}_1(\texttt{hasChild(A, B)}) \cap
\textit{Pr}_2(\texttt{hasChild(C, A)})$. As a second example, let us
consider a case where the projected label is a functor:
if $C$ is \texttt{p(X) :- q(X), r(X)}, then $\textit{Pl}(C) =
\texttt{q/1}$, assuming that \texttt{q/1} and \texttt{r/1} are
general, non-DNR predicates.

Using the definition of the projected label, we introduce the notion
of miniset graphs, which we will use to define the notion of the
miniset of a predicate.

\begin{definition}[The miniset graph of a DL program]
Let $S$ be a DL program containing the predicates $\{P_1, \ldots, P_n\}$,
where a predicate $P_i$ consists of clauses   $\{C_{i1}, \ldots,
C_{ik_i}\}$. The \emph{miniset graph} of $S$ is a labelled directed graph 
$(V, E, \mathcal{L})$, where $\mathcal{L}$ is a function assigning
labels to vertices. To each predicate $P_i$ and each clause $C_{ij}$ there
corresponds a node in the graph: $p_i$ and $c_{ij}$, respectively.
Thus $V = \mathop{\cup}_{P_i \in S}\{p_i, c_{i1}, \ldots, c_{ik_i}\}$. 

Each node $p_i$ is labelled with the functor of $P_i$, i.e.\
$\mathcal{L}(p_i) = \mathrm{the\ functor\ of\ }P_i$. A node corresponding to
a clause $C_{ij}$ is labelled with the projected label of the clause, i.e.\
$\mathcal{L}(c_{ij}) = \textit{Pl}(C_{ij})$.

There are directed edges from each predicate node $P_i$ to the nodes representing its clauses, i.e.\
$(p_i, c_{ij}) \in E, 1 \leq i \leq n, 1 \leq j \leq k_i$. Furthermore, for
each clause $C_{ij}$ whose projected label is a predicate functor $F$,
there is an edge from the corresponding node $c_{ij}$ to the node of the
predicate with the functor $F$.

\end{definition}

\noindent As an example,
the miniset graph of the Iocaste program of
Figure~\ref{figure:dl_predicates_iocaste} is shown in
Figure~\ref{figure:miniset_iocaste}.

\begin{figure}[htbp]
\centering
\psfrag{Ans}{\texttt{Ans/1}}
\psfrag{mAns1}{\texttt{\ogt\cgt}}
\psfrag{mAns2}{\texttt{\ogt i, o, p\cgt}}
\psfrag{Pat}{\texttt{Patricide/1}}
\psfrag{nPat}{\texttt{not\_Patricide/1}}
\psfrag{mPat1}{\texttt{\ogt o\cgt}}
\psfrag{mnPat1}{\texttt{\ogt t\cgt}}
\psfrag{mPat2}{\texttt{\ogt o, p, t\cgt}}
\psfrag{mnPat2}{\texttt{\ogt o, p\cgt}}
\includegraphics[scale=0.6]{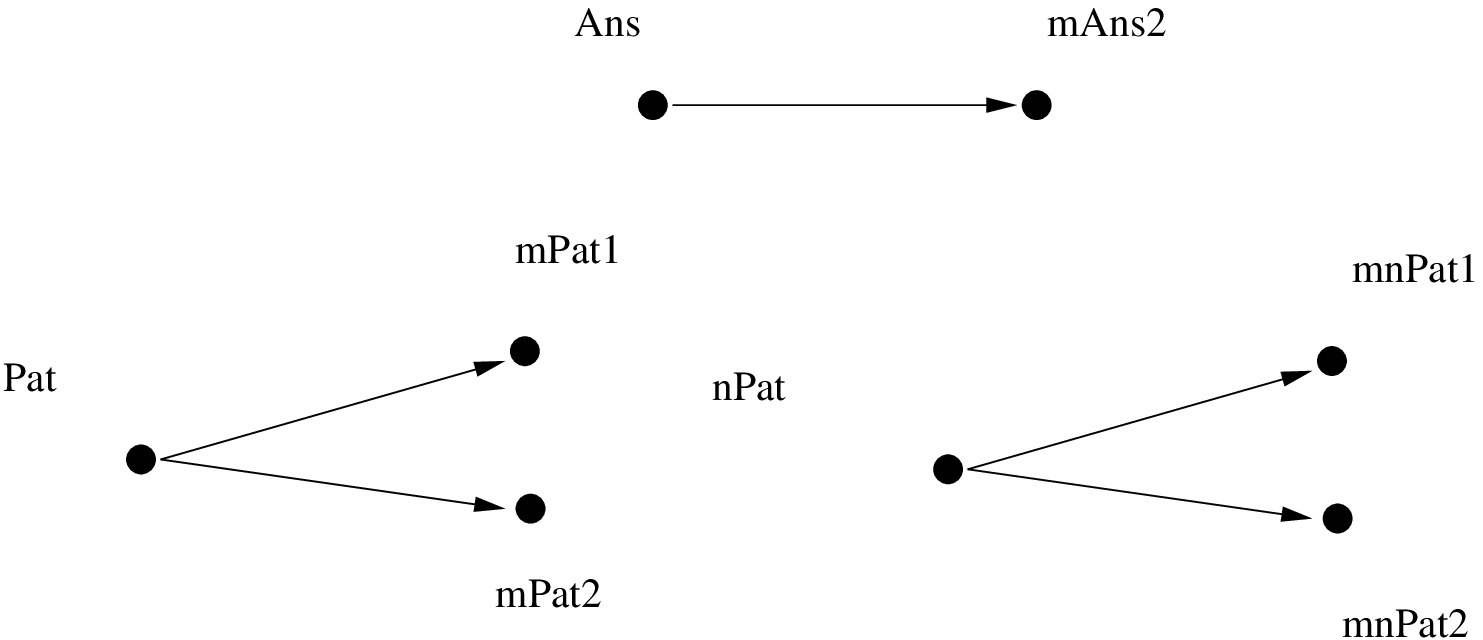}
\caption{The miniset graph of the Iocaste DL program in Figure~\ref{figure:dl_predicates_iocaste}.}
\label{figure:miniset_iocaste}
\end{figure}

\noindent Now we are ready to formulate the definition of the miniset
of a predicate.

\begin{definition}[The miniset of a predicate]
\label{def:miniset_predicate}
Let $G$ be a miniset graph of a DL program. The \emph{miniset} $M(P)$ of a
predicate $P$ in this program is calculated as the union of the labels of
the nodes which (1) are reachable from the node corresponding to $P$ in the
graph $G$, and (2) are labelled with a set.
\end{definition}

\noindent For example, the miniset of predicate \texttt{Patricide/1}
in the Iocaste knowledge base is $\{\texttt{o}\} \cup \{\texttt{o}, \texttt{p},
\texttt{t}\}= \{\texttt{o}, \texttt{p}, \texttt{t}\}$. Note that for an ABox stored in a database
the calculation of minisets can be done using database queries, as
the projected labels in the miniset graph refer only to
atomic and query predicates.

\begin{proposition}
For an arbitrary predicate $P$ in a DL program $\textit{DP}$, $M(P)$ is a superset of $P$.
\end{proposition}

\begin{proof}\small
If predicate $Q$ has clauses $C_1, \ldots, C_n$, then $I(Q)
=\mathop{\cup}_{C_i \in Q} I(C_i)$, assuming $Q$ cannot succeed using
ancestor resolution. When solving $P$, ancestor resolution cannot be used
at the very first entry to $P$, because the ancestor list is then
empty. Furthermore, if there is an edge from a clause $C$ to a predicate
$Q$ in the miniset graph of $\textit{DP}$, then the invocation of $Q$ is known to be
non-DNR, as specified in the definition of the projected label. This means
that ancestor resolution is not applicable when $Q$ is invoked from $C$.
Also note that the invariant $I(C)
\subseteq I(Q)$, similar to that mentioned after the definition
of projected label, holds for the edge $C \rightarrow Q$. 

Each execution of the goal $P(X)$ has a corresponding finite path in the miniset
graph of $\textit{DP}$. The endpoint of this path has a set as a label, which
contains the value assigned to $X$. Thus the answer to the query $P(X)$ is
contained in a set label reachable from $P$, and thus in $M(P)$, too.
\end{proof}

\subsubsection{Implementation}

In our implementation we calculate the miniset of a predicate $P$ in the
following way. First, for each clause reachable from $P$ in the miniset
graph, we collect the conjunction of the goals participating in the
construction of the projected label for the given clause.  Next, we build
an auxiliary predicate whose body is the disjunction of these
conjunctions. Finally, we calculate the set of solutions of this auxiliary
predicate using the standard predicate \texttt{setof}, and enumerate the
members of the superset using the list membership predicate
\texttt{member}:

Below we show an example of superset calculation for a fictitious predicate,
assuming it gives rise to three three conjunctions shown in lines 4--6,
where \texttt{X} is the head variable:
\begin{Verbatim}[numbers=left,numbersep=2pt,frame=single]
member_of_superset_goal(A) :- setof(X, goal(X), S), member(A, S).

goal(X) :- 
        (   hasChild(Y, X), hasChild(X, Z), hasFriend(X, W)
        ;   hasChild(X, Y) 
        ;   Rich(X)
        ).
\end{Verbatim}

\noindent Note that we simplify the \texttt{goal/1} above: we omit the
first branch of this disjunction (line 4) as it is subsumed by the more
general goal in the second branch (line 5). 

There are cases when the projection optimisation is not applied. For very
simple predicates, e.g.\ those invoking atomic goals only, calculating the
projection is simply duplication of work, and so this optimisation is not
used.  Another case is when all goals in the body of a clause are DNR
invocations (i.e.\ can succeed via ancestor resolution).  The superset of
such a clause (and of its predicate, too) is defined to contain all the
individuals in the ABox. Obviously, in such cases the superset optimisation
is not applied either. The definition of superset could be refined to decrease the
number of such cases. As an alternative, source-to-source transformation
techniques can be used to eliminate the need for ancestor resolution in the
very early phase of execution, as discussed in \cite{iclp2008}.

To conclude this section, in Figure~\ref{figure:final_translation_iocaste}
we show the most interesting parts of the compiled Iocaste problem.  To
save space we have omitted the definition of the predicate
\texttt{det\_Patricide/2} (lines 11--14 of
Figure~\ref{figure:ground_goal_handling}), as well as the \emph{choice} and
\emph{det} variants of the
predicate \texttt{not\_Patricide/2}, which are very similar to
corresponding variants of \texttt{Patricide/2}.  All optimisations
discussed so far have been applied here, including the superset
optimisation. Notice how simple is the entry predicate for
\texttt{Patricide} (lines 15--16): it only invokes the atomic ABox
predicate. This is because the last clause of \texttt{Patricide/2}
(cf.\ lines 9--10 in Figure~\ref{figure:naive_translation_iocaste})
contains an orphan goal which cannot succeed when \texttt{Patricide} is
used as an entry predicate (as it has an empty ancestor list argument). For
the same reason the clauses for loop elimination and ancestor resolution
(cf.\ lines 6--7 in Figure~\ref{figure:naive_translation_iocaste}) can be
removed.  This leaves us with a single clause with a single atomic goal,
for which there is no point in generating the superset.  Because of this,
we do not even generate the conditional structure usually present in
\emph{choice} predicates.

\begin{figure}[htbp]
\begin{Verbatim}[numbers=left,numbersep=2pt,frame=single]
choice_Ans(A, B) :-
        (   nonvar(A) -> det_Ans(A, B)
        ;   setof(A, C^hasChild(A,C), D), member(A, D), det_Ans(A, B)
        ).

det_Ans(A, B) :- C = [Ans(A)|B], 
                 hasChild(A, D),
                 (   hasChild(D, E),
                     det_not_Patricide(E, C) ->
                     true
                 ),
                 det_Patricide(D, C), !.

choice_Patricide(A, _) :-
        Patricide(A).

...
\end{Verbatim}
\caption{The final Prolog translation of the Iocaste problem.}
\label{figure:final_translation_iocaste}
\end{figure}

\subsection{Transforming role axioms}
\label{roles}

We present here a compilation scheme for $\mathcal{SHIQ}$ role axioms which
is more efficient than the one introduced in Section~\ref{compiling}.  We
consider role subsumption axioms only, as an equivalence axiom $R \equiv S$
can be replaced by two subsumptions, and the transitivity axioms are
removed by the first stage of the transformation (see
Section~\ref{shiqprinciples}). 

The general scheme of Section~\ref{compiling} applies loop elimination for
role axioms. This is required because, for example, the subsumption axioms
$R\sqsupseteq S$ and $S \sqsupseteq R$ are transformed to the
following two DL clauses, whose Prolog execution obviously
leads to an infinite loop:

\begin{Verbatim}[numbers=left,numbersep=2pt,frame=single]
R(A, B) :- S(A, B).
S(A, B) :- R(A, B).
\end{Verbatim}

\noindent
In general, looping of role-predicates is related to role equivalence
(roles $R$ and $S$ above are obviously equivalent). The main idea is to
avoid the need for loop elimination by designating one of the equivalent
roles as the representative of the others. All invocations of these
predicates are replaced by appropriate calls of the representative
predicate. Furthermore, of the two subsumption axioms stating role
equivalence, we keep only the one where the non-representative role is
defined in terms of the representative one.  In the above example, if $R$
is the representative, we replace all occurrences of $S$ by $R$ throughout
the TBox, and retain only the second of the above clauses, the one
corresponding to the axiom $S \sqsupseteq R$.

Note that the above scheme  does not work when a role subsumption axiom
states that a role $R$ is a \emph{symmetric}: 
$R^{-}\sqsubseteq R$. The  Prolog translation of this axiom, $R(X,Y)
\texttt{ :- } R(Y,X)$, is an obvious loop. We break this loop by introducing
an auxiliary predicate name \texttt{base\_}$R$, replacing all occurrences
of $R$ in clause heads by \texttt{base\_}$R$, and defining the predicate
$R$ in terms of \texttt{base\_}$R$ by the two clauses $R(X,Y) \texttt{ :- }
\texttt{base\_}R(X,Y)$ and $R(X,Y) \texttt{ :- } \texttt{base\_}R(Y,X)$.

We start the formal discussion with some auxiliary definitions.

\begin{definition}[Reduced graph]
Let $G$ be an arbitrary directed graph. The \emph{reduced graph} of
$G$, denoted by $G_r$, is defined as follows. The vertices of $G_r$
are the strongly connected components (SCC) of $G$. There is an edge in
$G_r$ from $A$ to $B$ if, and only if, there is an edge in $G$ from one
of the vertices in the SCC corresponding to $A$ to one of the vertices
in the SCC corresponding to $B$.
\end{definition}

\begin{definition}[Canonical inverse of role]
Let $R$ be an atomic role or its inverse. The canonical inverse of $R$,
denoted by $\textit{Inv}(R)$, is defined as follows.

\[ \textit{Inv}(R) = \left\{ \begin{array}{ll}
  S       & \mbox{if $R = S^{-}$}\\
  R^{-}     & \mbox{otherwise}
\end{array}
\right.
\]
\end{definition}

\begin{definition}[The role dependency graph]
For a given knowledge base $\textit{KB}$ the role
dependency graph $G = (V, E)$ is defined as follows.  The set of vertices $V$
of $G$ is  the set of atomic roles occurring in $\textit{KB}$ and of their inverses. There is a directed edge from $P_i$
to $P_j$ \emph{and} from $\textit{Inv}(P_i)$ to $\textit{Inv}(P_j)$,
if, and only if, $P_i \sqsubseteq P_j \in KB$.
\end{definition}

\noindent Let $G$ be a role dependency graph w.r.t.\ a knowledge base $\textit{KB}$
and let us consider its reduced graph $G_r$.  Notice that each node of
$G_r$ is a component of the original graph whose elements are equivalent
roles. Also notice that if roles $R_1, \ldots, R_n$ all belong to a single
component $E$, then roles $\textit{Inv}(R_1), \ldots, \textit{Inv}(R_n)$
belong to a single component as well, which we call the inverse of the
first component, and denote by $E^{-}$. A role is symmetric if, and only if,
its component is the inverse of itself.

Consider the set $E \cup E^{-}$, where $E$ is a component of a role
dependency graph. Predicate invocations of two
roles in this set return the same pairs of individuals (possibly
in a different order). Therefore we designate a single atomic role name, say the one
which comes first in the lexicographic order, as the \emph{representative}
of all roles in this set. Thus for any role $R \in E \cup E^{-}$, let 
$\textit{Repr}(R)$ denote the first of the atomic role names in this set, according
to the lexicographic order. 

We now discuss how to transform role predicate invocations and role
predicate heads, so that they use  representative roles only. The
transformation schemes for invocations and heads are the same, except for
symmetric roles $R$, where the auxiliary predicate $\texttt{base\_}R$ is
used in the heads to break the loops.

Let $\textit{DP}$ be a DL program generated from a knowledge base $\textit{KB}$, and let $G$
be the role dependency graph of $\textit{KB}$.  Let $RR =\textit{Repr}(R)$ denote
the representative of a role $R$.  Let us consider the compiled version of
the program $\textit{DP}$, as defined in Section~\ref{compiling}. We first remove
the ancestor list arguments and the loop elimination clauses (denoted by
$F_2$) from all role predicates. We then perform the following
transformations on all role predicate invocations and heads, except for
those prefixed with the module name \texttt{abox}  (occurring in the
bodies of clauses of type $F_3$).
\begin{enumerate}
  \renewcommand{\theenumi}{\alph{enumi}}
\item
If $R$ and $RR$ belong to the same
component of $G$, then the role predicate invocation
$R(X,Y)$ is replaced by $RR(X,Y)$; otherwise it is replaced  by $RR(Y,X)$.
\item If $R$ is not a symmetric role, then the role predicate head is
  transformed as described in point a.\ above.
\item If $R$ is a symmetric role, then the role predicate head 
  $R(X,Y)$ is replaced by $\texttt{base\_}RR(X,Y)$.
\end{enumerate}

Here we view the compiled program as a set of clauses, rather than a set of
predicates. This is important for two reasons. First, when replacing role
names with their representatives, several instances of the same clause may
be produced, of which only one should be kept. Second, changing clause
heads means that clauses are moved from one predicate to another.

Finally the compiled DL program is extended with the following predicates:
\begin{itemize}
\sloppy
\item  For each symmetric atomic role name  $R$, which is a representative
  of a set of roles, we add the following  definition: $R(X,Y) \texttt{ :- }
\texttt{base\_}R(X,Y)$ and $R(X,Y) \texttt{ :- } \texttt{base\_}R(Y,X)$. This
way $R$ becomes the symmetric closure of $\texttt{base\_}R$, which is
populated using the ABox and/or role subsumption axioms.
\item For each atomic role name $R$, which is not a representative of a
  role set, we build the (tautological) clause $R(X,Y) \texttt{ :- } R(X,Y)$,
  and transform its body according to point a.\ above. Such clauses will
  only be used when the role $R$ occurs in a composite query. (An
  alternative is not to include these clauses, and to apply the
  transformation of point a.\ above to the composite query.)
\end{itemize}

The above transformation  can be easily combined with the role
indexing technique introduced in Section~\ref{indexing}. This is
incorporated in the DLog system, but the details are not discussed here.

The transformation scheme  has several
advantages. First and foremost, it ensures that the evaluation of a role
predicate cannot loop, and so there is no need for the ancestor list
argument and the loop-elimination clause in the role
predicates. Furthermore, it avoids those duplicate solutions, which are due to
interchangeability of equivalent roles. However, a role predicate can still
produce duplicate solutions (e.g.\ when the role subsumes two other roles
sharing a solution), and the transformation scheme could be refined further
to improve the efficiency of execution.

\subsection{Summary}

We presented several optimisations which result in a
much more efficient Prolog translation, in comparison with the generic
compilation scheme described in Section~\ref{shiqreasoning}. These
optimisations preserve the most important property of the generic
compilation scheme, e.g.\ the separation of the TBox from the ABox. In
the following we give a brief summary of the optimisations
presented.

In \emph{filtering}, we remove those clauses that need
not to be included in the final program as they are not used in the
execution. We proved that the certain clauses (those having the \emph{false-orphan},
the \emph{two-orphan}, or the \emph{contra-two-orphan} property) can be removed.

\emph{Classification} puts each predicate into one of the four categories:
\emph{atomic}, \emph{query}, \emph{orphan} and \emph{generic}. For each
class, we presented an optimised translation scheme.

The \emph{ordering} optimisation arranges the body goals so as to
minimise the execution time. We defined a heuristic and specified an
ordering algorithm which uses this heuristic.

The \emph{indexing} optimisation is introduced to get around the
problem that in most Prolog systems indexing is done only on the first
head argument and this may raise performance issues if we use Prolog
for storing large amounts of ABox facts.

The \emph{ground goal optimisation} makes sure that if a ground goal
succeeds, then all choice points within it are pruned. To achieve
this, we create two versions of each unary predicate, which handle the
cases of the head variable being instantiated or uninstantiated.

The goal of \emph{decomposition} is to split a body into independent
components: this recursive process introduces a more refined notion of
body ordering and a generalisation of the ground goal optimisation. We
described the decomposition process and specified its relation to the
body ordering algorithms.  

The idea of the \emph{superset} optimisation is to determine, for each predicate $P$, a
set of instances $S$ for which $I(P) \subseteq S$ holds, where $I(P)$
is the set of solutions of $P$. If the size of $S$ is not
significantly larger than that of $I(P)$, then we can use $S$ to
efficiently reduce the initial instance retrieval problem to a finite
number of instance checks. We defined the notions of miniset graphs and
minisets and showed that the so-called miniset of a predicate fulfils
the above criteria.

Finally, we defined an efficient translation scheme for the
$\mathcal{SHIQ}$ role axioms ($R \sqsubseteq S$).

\section{The DLog system}
\label{dlog}

In this section we first introduce the software architecture of the
DLog system. Next, we discuss the implementation specific
optimisations we have developed. Finally, we present the various
parameters one can use to tune the behaviour of the DLog system.

\subsection{Architecture}
\label{architecture}

DLog is a resolution based Description Logic ABox reasoner for the
$\mathcal{SHIQ}$ DL language, which implements the techniques described in
the paper. DLog has been developed in Prolog, involving a total of
approximately 180KB of Prolog source code. Our main implementation is in
SICStus Prolog, a port to SWI Prolog has been completed recently.

The general architecture of the DLog implementation is shown in
Figure~\ref{figure:architecture}. The system can be used as a server
or as a stand-alone application. In either case, the input of the
reasoning process is provided by a DIG file. The DIG format \cite{dig} is a
standardised XML-based interface for Description Logic Reasoners.

\begin{figure}[htbp]
  \centering
  \psfrag{di}{DIG input}
  \psfrag{pp}{SICStus Process}
  \psfrag{r}{Runtime system}
  \psfrag{a}{ABox}
  \psfrag{t}{TBox}
  \psfrag{a1}{no modification/}
  \psfrag{a11}{indexing}
  \psfrag{a2}{compiling}
  \psfrag{q}{Ask part}
  \psfrag{us}{User queries}
  \psfrag{pr1}{Generated}
  \psfrag{pr2}{Prolog}
  \psfrag{pr3}{program}
  \psfrag{ma}{Module \texttt{abox}}
  \psfrag{mt}{Module \texttt{tbox}}
  \psfrag{out}{DIG output}
  \psfrag{Dat}{Database}

  \includegraphics[scale=0.6]{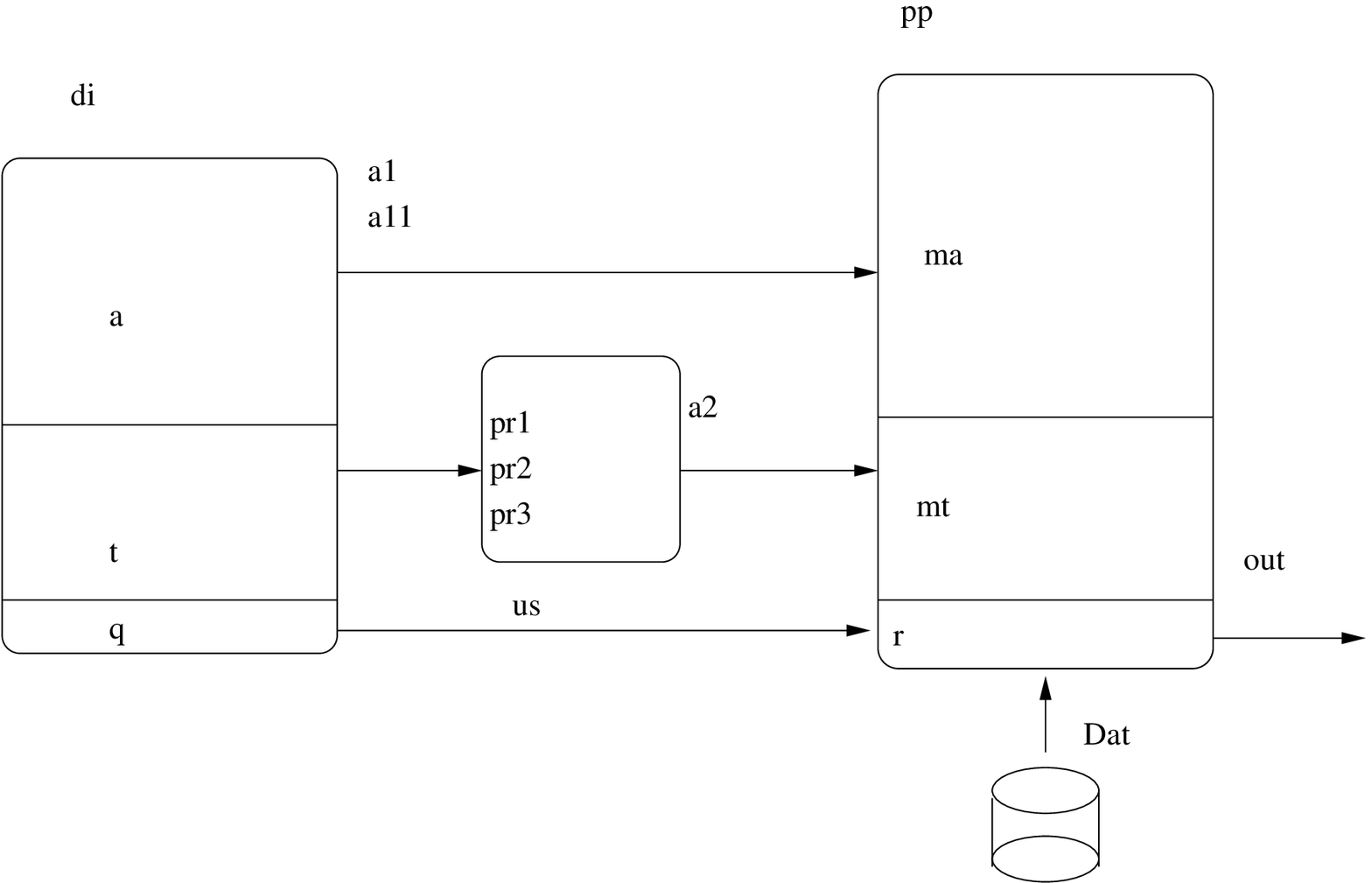}
\caption{The architecture of the DLog system.}
\label{figure:architecture}
\end{figure}

The input file has three parts: the (potentially) large ABox, the
smaller TBox and the Ask part describing the queries. The content of
the ABox is asserted into module \texttt{abox}, either with no
modifications, or (if the indexing optimisation is applied) together with the index
predicates. Note that the ABox can also be supplied as a database. This is
essential for really large data sets.

The content of the TBox is transformed into a Prolog program following
the techniques described in previous sections. This program is
then compiled into module \texttt{tbox}.

The content of the Ask part in the DIG input contains the user
queries. In the simplest case, the user poses an instance retrieval
query which directly corresponds to a concept in the TBox. Such cases
are answered by directly invoking the appropriate choice predicate. In
the more complex case we have a \emph{conjunctive query} as introduced
in Section~\ref{specpttp}.

We handle conjunctive queries by reducing  the problem of
query answering to a normal DL reasoning task \cite{horrocks00how}. We
simply apply body reordering (Section~\ref{ordering}) and
decomposition (Section~\ref{decomposition}) on a conjunctive query and
use the normal Prolog execution for the resulting goal. We are aware
that much more sophisticated techniques are available
\cite{motik06PhD}, but at the moment our simple approach seems to be
efficient enough.

\subsection{Low-level optimisations}
\label{further}

During the implementation we have applied several optimisations which
can be considered implementation specific or low level. Below we
give a brief summary of these optimisations.

\paragraph{Loop and ancestor separation} 

It is worth separating the data structures used for loop elimination
and ancestor resolution. This way we can update them separately, which
results in more efficient execution.

\paragraph{Hashing} 

Rather than using lists, we introduced a more efficient data structure
to store the goals used for loop elimination and ancestor resolution. For
this purpose we developed a special hashing library written in Prolog and C,
relying on the foreign language interface of the Prolog system used. 

As an example for the benefits of hashing consider the following DL
knowledge base.

\begin{Verbatim}[numbers=left,numbersep=2pt,frame=single,commandchars=\\\{\}]
\(\exists\texttt{hasFriend}\ldotp\texttt{Alcoholic} \sqsubseteq \neg\texttt{Alcoholic}\)
\(\exists\texttt{hasParent}\ldotp\neg\texttt{Alcoholic} \sqsubseteq \neg\texttt{Alcoholic}\)

hasParent(i1, i2). hasParent(i1, i3). hasFriend(i2, i3). 
\end{Verbatim}

This TBox states that if someone has a friend who is alcoholic then
she is not alcoholic (she sees a bad example). Furthermore, if someone
has a non-alcoholic parent then she is not alcoholic either (she sees
a good example). The ABox contains two \texttt{hasParent} and one
\texttt{hasFriend} role assertions, but nothing about someone being
alcoholic or non-alcoholic. Interestingly, it is possible to conclude
that \texttt{i1} is non-alcoholic as one of her parents is bound to be
non-alcoholic (as at least one of two people who are friends has
to be non-alcoholic).

For certain ABoxes, the Prolog translation of this knowledge base has
a runtime which is quadratic in the number of \texttt{hasParent}
relations, if the ancestors are stored in a list. An example of such an ABox
is the following:

\begin{alltt}
hasParent(i\(\sb{k}\), i\(\sb{k+1}\)), \(k = 1,\ldots,n\)
hasFriend(i\(\sb{n+1}\), i\(\sb{n+2}\))
hasParent(i\(\sb{n+2+t}\),i\(\sb{n+1+t}\)), \(t = 1,\ldots,n\)
\end{alltt}

Here, for each individual \texttt{i$_{n+1+t}$}, $t > 0$, the Prolog
code checks if the ancestor list contains the term
\texttt{not\_Alcoholic(i$_{n+1+t}$)}. This has a linear cost w.r.t.\ the
size of the ancestor list, assuming that the check for a given
ancestor is performed by a linear scan of the ancestor list. The
quadratic time complexity can be reduced to (nearly) linear when a
hash table is used for storing the ancestors (with a nearly
constant time ancestor check).

\paragraph{Placing the update operations} 

In the translation scheme presented in Section~\ref{generation} the
extension of the ancestor list takes place at the very beginning of each
clause (see e.g.\ line 7 of
Figure~\ref{figure:final_translation_iocaste}). However, updating a hash
structure is more expensive than adding a new element to a list. Therefore
we perform the hash update operation as late as possible, i.e.\ before the
first goal which uses the updated hash value. This, for example,
corresponds to moving the ancestor update operator in line 7 of
Figure~\ref{figure:final_translation_iocaste} to before line 10.

\paragraph{Clause-level categorisation} 

The predicate categorisation (see Section~\ref{classification}) can be
refined so that the characteristics of individual clauses of the predicate
are taken into account. For example, even if a predicate is recursive, some
of its clauses will never lead to recursive calls of this predicate. For
these clauses, there is no point in updating the loop data
structure. Similarly, if $\textit{not\_P}$ cannot be reached from a clause
of $P$, then there is no need for updating the ancestor data structure in
the given clause.

\subsection{Execution parameters}
\label{options}

Most of the optimisations discussed in Section~\ref{generation} can be
enabled/or disabled in DLog, resulting in different generated Prolog
programs. The possible parameter settings are summarised below (the
parameter values allowed are shown in parentheses, the first value is the
default):

\begin{itemize}
\item decompose (yes/no): whether to decompose the bodies  (Section~\ref{decomposition})
\item indexing (yes/no):  whether to generate index predicates for roles (Section~\ref{indexing})
\item projection (yes/no):  whether to calculate supersets (Section~\ref{projection})
\item filtering (yes/no): whether to do filtering (Section~\ref{preprocessing})
\item ground\_optim (yes/no): whether to use ground goal optimisation (Section~\ref{ground})
\item orphan (first/general): whether orphan calls are brought to the
  beginning of the clause or handled in the same way as general concept calls
\item hashing (yes/no): whether to apply hash tables instead of lists for storing ancestors
\end{itemize}

\section{Performance Evaluation}
\label{evaluation}

This section presents a comparison of the performance of DLog with
that of existing DL reasoning systems. The aim here is to obtain an
insight into the practical applicability of the methods described in
Sections \ref{shiqreasoning} and \ref{generation}.

During the tests we have found several anomalies which resulted
in significant performance drops in the case of certain DL
reasoners. We believe that most of these will be fixed by the
respective authors in the near future. Here, however, we took each
system ``as it is'', which means that we examined how their most
up-to-date version performs on various inputs.

Our tests suggest that resolution-based techniques are very promising
in practical applications, where the TBox is relatively small, but the
ABox is huge.

\subsection{Test environment}

We have compared our system with three state-of-the-art description
logic reasoners: RacerPro 1.9.0, Pellet 1.5.0 and the latest version
of KAON2 (August 2007). RacerPro \cite{HMSW04} is a commercial system,
Pellet \cite{pellet} is open-source, while KAON2 \cite{motik06PhD} is
free of charge for universities for non-commercial academic usage. We
did not consider other available reasoning systems mainly because they
are either outdated or they do not support ABox reasoning at all
(e.g.\ this is the case for the widely used FaCT system).

We contacted the authors of each reasoning system in order to obtain the
preferred sequence of API calls for running our tests. From one of
them we did not receive any response so we used the API according to
the documentation. The benchmarks were executed by a test framework we
have specifically developed for testing DLog and other systems. For
each query, we started a clean instance of the reasoner and loaded the
test knowledge base. Next, we measured the time required to
execute the given query. Each query was executed 5 times. The best and
the worst results were excluded and the remaining three were used for
calculating an average. In case the execution was very fast (less
than 10 milliseconds) we have repeated the test 1000 times and
calculated the average. We made sure that all systems return the same
answers for the given queries.

The tests were performed on a Fujitsu-Siemens S7020 laptop with a
Pentium-M 1.7GHz processor, 1.25GB memory, Ubuntu Linux 7.04 with
Linux kernel 2.6.20-16 and SICStus Prolog 3.12.8. The version of the
Java virtual machine, used for KAON2 and Pellet, is 1.5.0.

\subsection{Test Ontologies}

For the benchmark we have used three families of ontologies. The first
one corresponds to the Iocaste problem introduced in
Figure~\ref{figure:iocaste_kb}. For performing this test we have
created a program which generates random Iocaste knowledge bases using
certain initial parameters (number of nodes, branching factor, etc.).

First we used this program to generate ``clean'' Iocaste ontologies,
i.e.\ DL knowledge bases with the ABox containing nothing else but
Iocaste patterns of a given size
(cf.\ Figure~\ref{figure:patterns}). These knowledge bases are named
\texttt{cN}, where \texttt{N} is the size of the pattern. For example,
\texttt{c100} denotes the DL knowledge base with a single TBox axiom
and an ABox containing 102 individuals according to
Figure~\ref{figure:patterns} with $n = 100$.

We have also generated ``noisy'' Iocaste knowledge bases. By ``noise''
we mean irrelevant individuals, role and concept assertions which we
added to the ABoxes (for example pairs in \texttt{hasChild} relation
which are not relevant to the Iocaste problem).  We did this in
order to be able to measure how sensitive the inference engines are to
this kind of ABox modification. By using irrelevant nodes we actually
simulate real life search situations, where the task is to find some
specific instances within huge amounts of data. The noisy Iocaste
knowledge bases are named \texttt{n1}, \texttt{n2}, \texttt{n3}, and
\texttt{n4}. Table~\ref{table:testfilesprop}, on
page~\pageref{table:testfilesprop}, shows the
properties of the clean and noisy Iocaste ontologies, together with their
DLog compilation times, under various parameter settings.

The top four rows of the table contain information on the
knowledge bases. The first row
gives the size of the corresponding DIG files (in megabytes), the second
and third shows the number of TBox and ABox axioms, while the fourth row shows the
time (in seconds) it took for the DLog system to parse the DIG files and convert
them to Prolog terms (load time). We can see that the largest clean ontology
contains a bit more than 20000 ABox axioms, while the largest noisy
ontology has more than 30000 axioms. Each ontology contains only a
single TBox axiom (cf.\ Figure~\ref{figure:iocaste_kb}). 

Subsequent sections in Table~\ref{table:testfilesprop} correspond to
the various parameter settings we have tried the DLog system with. For
each setting, we give the translation time (the time it took to
generate the Prolog program as described in
Sections~\ref{shiqreasoning} and \ref{generation}) and the time it
took the SICStus Prolog system to actually compile the generated
program. The total time is the sum of three values: the load time, the
translation time and the compile time.

In our tests, out of the possible $2^8$ option variations
(cf.\ Section~\ref{options}), we have only used the following ones:

\begin{itemize}
\item \texttt{base}: everything is left as default
\item \texttt{[g(n)]}: do not use ground goal optimisation 
\item \texttt{[p(n)]}: do not use projection  
\item \texttt{[f(n)]}: do not use filtering 
\item \texttt{[i(n)]}: do not use indexing 
\item \texttt{[o(n)]}: handle orphan goals as general concept goals 
\item \texttt{[d(n)]}: do not use decomposition
\item \texttt{[pd(n)]}: do not use projection and decomposition
\item \texttt{[od(n)]}: do not use decomposition and handle orphan goals as general concept goals
\end{itemize}

Table~\ref{table:testfilesprop} shows that most settings have very
similar compile-time properties. However, the setting \texttt{[i(n)]}
disables the generation of index predicates, which
results in a more compact code. This means shorter translation and
compilation times.

Note that the Iocaste ontologies, both the clean and the noisy ones,
use the $\mathcal{ALC}$ DL language.

The second ontology we used for testing is VICODI
\cite{nm03fuzzy}, an ontology about European history,
created manually. It is the result of an EU-IST programme
project. Technically VICODI is an $\mathcal{ALH}$ ontology with a
fairly simple TBox and a huge ABox. We have obtained VICODI from the
VICODI homepage in the form of a Protege project. We have converted
this project into OWL and DIG formats using Protege 3.3.1 and used
these as inputs for the various reasoners. The sizes of these
converted files are 9.5 and 23 megabytes, respectively.

The VICODI TBox consists of 182 concept and 9 role subsumption
axioms. The ABox contains 84550 role axioms and 29614 concept axioms.

Finally, we have also tested our system on LUBM, the Leigh University
Benchmark \cite{lubm}. LUBM was developed 
specifically as a benchmark for the performance analysis of description logic
reasoners. The ontology describes the organisational structure of
universities and it uses the $\mathcal{ALCHI}$ language. The ABox part
can be automatically generated by specifying a size parameter (the
number of universities).

We have used four variants of the  ontology,  denoted by \texttt{lubm1}, \texttt{lubm2},
\texttt{lubm3} and \texttt{lubm4}. All contain 36 concept inclusion, 6
concept equivalence, 5 role inclusion and 4 role equivalence
axioms. They also contain a transitive role and 21 domain and 18 range
restrictions. The number of ABox axioms in the various LUBM ontologies
and their sizes in megabytes are shown in
Table~\ref{table:lubm_prop}.

{
\renewcommand{\arraystretch}{1.2}
\begin{table}[htb]
\centering
\caption{The properties of the LUBM test ontologies}
\begin{tabular}{llrrrr}
\hline\hline
& testfile &\texttt{lubm1}&\texttt{lubm2}&\texttt{lubm3}&\texttt{lubm4}\\
\hline
& OWL filesize (MBytes) & 6.90 & 15.84 & 23.24 & 32.97\\
& DIG filesize (MBytes) & 16.57 & 37.99 & 58.80 & 81.74\\
& concept assertions & 18128 & 40508 & 58897 & 83200\\
& role assertions & 49336 & 113463 & 166682 & 236514\\ 
\hline\hline
\end{tabular}
\label{table:lubm_prop}
\end{table}
}

\subsection{Results}

We now present the performance results for the Iocaste, VICODI and LUBM
ontologies. For each case we give a detailed explanation of the
results.

\subsubsection{The Iocaste ontologies}

The performance results of the DLog system on the Iocaste ontologies are
presented in Table~\ref{table:result}, on page~\pageref{table:result}.
Here, we show four values for each parameter setting. Three values (loop,
ancres, orphancres) give statistical information, describing the number of
loop eliminations, ancestor resolutions and orphan ancestor resolutions
(ancestor resolutions in orphan goals). Finally, we show the most important
value, the runtime in seconds.

With the best settings (\texttt{base}) DLog solved each task within a
fraction of a second, including the biggest clean and the biggest
noisy cases as well. Actually, using projection
(cf.\ Section~\ref{projection}) seems to be a key factor, as without
it the performance drops dramatically. We can also notice that the
lack of multiple argument indexing (cf.\ Section~\ref{indexing}) has
very negative effect on the execution time.  With the last
parameter setting DLog was unable to solve all tasks (this is denoted by
-). In this setting we do not use decomposition and we treat orphans
as normal predicates. The reason why this setting has the worst
performance is that it causes the orphan goal \texttt{not\_Ans(D, B)}
to be placed at the very end of the corresponding
\texttt{det\_Patricide/2} clause
(cf.\ Figure~\ref{figure:ground_goal_handling}).

We have compared the performance of DLog, using the \texttt{base}
parameter setting, with that of the other three reasoning
systems. These aggregate results are shown in
Table~\ref{table:aggregate}. In this table, as in the rest of the
paper, whenever we compare various systems/options, the best total
time is given in \textbf{bold}.

{
\renewcommand{\arraystretch}{1.25}
\begin{table}[htb]
\caption{Aggregate results for the Iocaste ontologies (times in seconds)}
\begin{center}\begin{tabular}{llrrrrrrrrr}
\hline
& Testfile &c10&c20&c100&c1000&c10$^4$&n1&n2&n3&n4\\
\hline\hline
\multirow{3}{*}{\begin{sideways}DLog\end{sideways}} & load &0.07&0.08&0.15&0.33&1.47&0.14&0.24&0.38&1.99\\
& runtime &0.00&0.00&0.00&0.01&0.11&0.00&0.00&0.00&0.02\\
& \textbf{total} &\textbf{0.07}&\textbf{0.08}&\textbf{0.15}&\textbf{0.34}&\textbf{1.58}&\textbf{0.14}&\textbf{0.24}&\textbf{0.38}&\textbf{2.01}\\
\hline
\multirow{3}{*}{\begin{sideways}KAON2\end{sideways}} & load &0.45&-&-&-&-&0.46&0.60&0.97&2.36\\
& runtime &0.72&-&-&-&-&0.67&4.72&63.60&425.17\\
& total &1.17&-&-&-&-&1.13&5.32&64.57&427.53\\
\hline
\multirow{3}{*}{\begin{sideways}RacerPro\end{sideways}} & load&0.01&0.01&0.03&0.51&4.68&0.03&0.10&0.68&6.04\\
& runtime &0.07&0.09&0.15&1.68&79.91&0.10&0.47&1.76&23.25\\
& total &0.08&0.10&0.18&2.19&84.59&0.13&0.57&2.44&29.29\\
\hline
\multirow{3}{*}{\begin{sideways}Pellet\end{sideways}} & load &1.27&1.35&1.44&2.19&-&1.32&1.53&2.36&5.92\\
& runtime &0.19&0.32&1.31&456.40&-&0.33&0.80&2.48&23.95\\
& total &1.46&1.68&2.76&458.58&-&1.65&2.33&4.84&29.87\\
\hline
\end{tabular}
\end{center}
\label{table:aggregate}
\end{table}
}

Here, for each Iocaste ontology and each reasoning system we give the
following values: the load time, the runtime and their sum, the total
time. The load time in the case of the DLog system includes parsing,
translating and SICStus compilation (cf.\
Table~\ref{table:testfilesprop}). For the other systems, load time is the
time it takes to reach the point when a query can be posed (we do not have
detailed information what the systems are actually doing here other than
parsing the input). Note that the size of the input given to DLog is bigger
than that for the other systems, as the DIG format is more verbose than the
OWL one.

KAON2 showed a very poor performance on the clean Iocaste ontologies:
\texttt{c10} was the only test case it was able to solve within the
time limit. To understand better what is going on we have tested KAON2
with clean Iocaste patterns of length $n=11, \ldots 15$.  The results
of this experiment are summarised in
Table~\ref{table:kaon2_iocaste}. Here we can see that KAON2 scales
very badly when increasing the size of the pattern. Note that the
increase between the consecutive test cases is minimal: the ontology
\texttt{c}$_{i+1}$ has one more instance
and two more role assertions, than the ontology \texttt{c}$_i$.

{
\renewcommand{\arraystretch}{1.2}
\begin{table}[htb]
\centering
\caption{Performance of KAON2 on the Iocaste ontologies (times in seconds)}
\begin{tabular}{llrrrrrr}
\hline\hline
& test file &c10&c11&c12&c13&c14&c15\\
\hline
& runtime &0.72 &0.68& 3.51 & 16.18 & 17.03 & 309.91\\
\hline\hline
\end{tabular}
\label{table:kaon2_iocaste}
\end{table}
}

Another interesting thing is that KAON2 actually ran \emph{faster} on
ontology \texttt{c11} than on \texttt{c10}. It also seems to scale
reasonably well (at least comparing to the other cases) from
\texttt{c13} to \texttt{c14}.

In the case of the noisy ontologies KAON2 also behaved
strangely. Although it was able to solve all the tests within 10
minutes, we definitely expected KAON2 to solve these cases much
faster. This is because KAON2 uses resolution, similarly to DLog,
which theoretically means that it should be resistant to noise to a
large extent.

We have actually learnt \cite{boris} that in KAON2 many things depend
on the order of rule applications, something which is a very difficult
task to set properly. Choosing a bad order may result in a big
performance drop. This can be a reason for the anomalies we have seen
in the case of the Iocaste ontologies.

RacerPro was able to solve each test case within the time limit. It
showed a very consistent behaviour both in the case of the clean and
the noisy variants. From the test results it seems that RacerPro
scales linearly although with a much worse constant than DLog. As a
tableau based reasoner, RacerPro showed a surprisingly good
performance in the case of the largest noise variant \texttt{n4},
as well (23.25 seconds).

Pellet was nearly as fast as RacerPro in the case of the noisy
variants. On the clean Iocaste ontologies, however, it was clearly
outclassed by RacerPro as Pellet was not able to solve \texttt{c10000}
within 10 minutes and in all the other cases it was fairly slow as
well. We have also found that in several cases Pellet threw certain Java
exceptions on the very same input it successfully processed earlier or
later. We guess that this can be  due to the use of Java hash codes.

As a conclusion, we state that DLog is several orders of magnitude faster
on the Iocaste benchmark than the other ABox reasoning system examined,
considering both the reasoning time (runtime) and the total execution time.

\paragraph{Using databases}

Instead of creating large internal Prolog databases for storing the
ABox, we can actually put the content of the ABox into a real database
and use DLog to generate a program from only the TBox. We have used
this technique for the largest noisy variant \texttt{n4} with the option setting
\texttt{[i(n)]}. Here, according to Table~\ref{table:testfilesprop}
and Table~\ref{table:result}, we use 1.41 seconds for the compilation and
0.02 seconds for runtime. By using a database for storing the content
of the ABox, we expect drastic decrease in the total compilation time,
with a slight increase in the execution time.

The actual (MySQL) database contains 15 tables, of which 10
correspond to concepts (i.e.\ they have only one column), while the
rest corresponds to roles (i.e.\ they have two columns). Note that
because of the top-down execution, the Prolog program generated from
the TBox actually accesses only tables \texttt{Patricide},
\texttt{not\_Patricide} and \texttt{hasChild}.  We have 5058 pairs in
\texttt{hasChild} relation, 855 instances are known to be patricide
and 314 are known to be non-patricide.

The performance results are summarised in
Table~\ref{table:database}. The database variant of \texttt{n4}
enumerated all the instances of concept \texttt{Ans} in 0.36
seconds. This, compared to the original 0.02 seconds is much
slower. However, the time we spent at compile-time was altogether 0.07
seconds, resulting in a total execution time of 0.43 seconds.

{
\renewcommand{\arraystretch}{1.2}
\begin{table}[htb]
\centering
\caption{The in-memory and database variants of \texttt{n4} (times in seconds)}
\begin{tabular}{llccccc}
\hline\hline
& DLog & load time & translation time & compilation time & runtime & total \\
\hline
& in-memory &0.88& 0.52 & 0.01 & 0.02 & 1.43\\
& database &0.05& 0.01 & 0.01 & 0.36 & 0.43\\
\hline\hline
\end{tabular}
\label{table:database}
\end{table}
}

From the figures of Table~\ref{table:database}, 
one may think that the main benefit of using a database for storing
the ABox lies in reducing the compilation time. However, we believe
that by using further optimisations, such as transforming query
predicates to database queries, the version using a database can also
produce better execution times than the in-memory variant.

We have thus shown that it is feasible to use a database for storing the
content of an ABox, and, in the case of the Iocaste ontologies, the
database approach provides better overall performance than the variant
which stores the ABox as Prolog facts.

\paragraph{Hashing}

We have also measured how much is the execution time affected by the
data structures used for storing ancestor goals. For this, we have
picked the best parameter setting, \texttt{base}, and run the tests by
replacing the hash tables with simple lists as assumed throughout
Section~\ref{generation}. The results are summarised in Table
\ref{table:datastructures} together with the hash-based results from
Table~\ref{table:result}.

{
\renewcommand{\arraystretch}{1.2}
\begin{table}[htb]
\centering
\caption{The effect of hashing on the Iocaste ontologies (times in seconds)}
\begin{tabular}{llrrrrrrrrr}
\hline\hline
& testfile &c10&c20&c100&c1000&c10000&n1&n2&n3&n4\\
\hline
& hash &0.00& 0.00 & 0.00 & \textbf{0.01} & \textbf{0.11}  & 0.00 & 0.00 & 0.00 & \textbf{0.02}\\
& list &0.00& 0.00 & 0.00 & 0.11 & 10.52 & 0.00 & 0.00 & 0.00 & 0.03\\
\hline\hline
\end{tabular}
\label{table:datastructures}
\end{table}
}

We can see that in the case of the large Iocaste patterns
(\texttt{c1000} and \texttt{c10000}) the hashing implementation
outperforms the solution using lists significantly.

\subsubsection{VICODI}

To test the performance of the DL reasoners on the VICODI ontology, we
used the following two queries, borrowed from \cite{motik06PhD}:

\begin{Verbatim}[numbersep=2pt,frame=single,commandchars=\\\{\}]
VQ\(\sb{1}\)(X)     \(\equiv\) Individual(X)
VQ\(\sb{2}\)(X,Y,Z) \(\equiv\) Military-Person(X), hasRole(Y, X), related(X, Z)
\end{Verbatim}

The results are summarised in Table~\ref{table:aggregate_vicodi}. The
DLog system used 8.61 seconds to load the VICODI ontology. From this,
4.91 seconds were actually spent on parsing the input and transforming
the DL knowledge base into DL predicates. DLog used 3.38 seconds
to generate the Prolog code. The rest (0.36 seconds) was used by
SICStus Prolog to compile the generated Prolog program. Having loaded the
knowledge base,
the execution was nearly instantaneous: 0.05 seconds for
\texttt{VQ$_1$} and 0.09 seconds for \texttt{VQ$_2$}.

{
\renewcommand{\arraystretch}{1.2}
\begin{table}[htb]
\centering
\caption{Aggregate results for the VICODI ontology (times in seconds)}
\begin{center}\begin{tabular}{llrrrr}
\hline\hline
&& DLog & KAON2 & RacerPro & Pellet\\
\hline
& load time & 8.61 & \textbf{5.88} & 34.96 & -\\
\hline
\multirow{2}{*}{\begin{sideways}\texttt{VQ$_1$}\end{sideways}} & runtime  & \textbf{0.05} & 0.36 & 76.48 & -\\
& total & 8.66 & \textbf{6.24} & 111.44 & -\\
\hline
\multirow{2}{*}{\begin{sideways}\texttt{VQ$_2$}\end{sideways}} & runtime & \textbf{0.09} & 0.35 & 76.61 & -\\
& total & 8.70 & \textbf{6.23} & 111.57 & -\\
\hline\hline
\end{tabular}
\end{center}
\label{table:aggregate_vicodi}
\end{table}
}

RacerPro spent nearly 35 seconds for loading the ontology. The
execution of \texttt{VQ$_1$} was fairly slow: it took 76.48 seconds to
enumerate all the instances of class \texttt{Individual}. We also
measured the execution time by first checking the consistency of the
ABox, then preparing the query answering engine before posing the
query itself. The consistency check took 65.86 seconds, the query
engine preparation 1.29 seconds and the query itself 8.25
seconds. This results in a total time of 75.40, which (as expected) is
comparable to the total time of simply loading and querying.

In the case of \texttt{VQ$_2$}, RacerPro produced nearly the same
results. We believe this is because RacerPro spends most of its time
in checking ABox consistency, which requires the same amount of time
in both queries.

Pellet was unable to answer any of the queries within the 10 minutes
time limit. We believe that Pellet properly read the input as we could
formulate VICODI queries which Pellet was able to answer, but this was
not the case for queries \texttt{VQ$_1$} and \texttt{VQ$_2$}. We have
also tried the Windows version of Pellet, but we have experienced the
same behaviour. Actually, in \cite{motik06PhD} Pellet 1.3 beta was
tested against the VICODI ontology with acceptable results. Thus it
seems that recent changes in the Pellet reasoner are responsible for the
performance drop we have found.

KAON2 could not read the VICODI OWL input we generated with Protege: we got
an exception. To be able to run the tests, we used a version of the
ontology specifically made for KAON2 (available on the VICODI
website). This version of the ontology is physically twice as large as the
normal OWL dialect (i.e.\ it is 18MB). On this, KAON2 was very
convincing. It took 5.88 seconds to load the ontology and 0.36 seconds to
answer query \texttt{VQ$_1$}. Answering query \texttt{VQ$_2$} was even a
bit faster, it required 0.35 seconds. We note that neither RacerPro, nor
Pellet supports this format of the VICODI ontology, so the comparison is
not fully fair.

To conclude we can say that KAON2 had the best overall performance
when dealing with the VICODI ontology. DLog answered the queries even
faster than KAON2, but for the compile-time tasks we needed a few
seconds more.  We note,
however, that the DIG input is larger by 5MB than the KAON2 version of
the VICODI ontology which naturally results in more load time work
for us.

\subsubsection{LUBM}

We have tested the LUBM ontologies with the following two queries:

\begin{Verbatim}[numbersep=2pt,frame=single,commandchars=\\\{\}]
LQ\(\sb{1}\)(X)     \(\equiv\) Person(X), hasAlumnus(h{}ttp://www.University0.edu, X)
LQ\(\sb{2}\)(X,Y)   \(\equiv\) Chair(X), Department(Y), worksFor(X, Y), 
\hspace*{0.25cm}            subOrganizationOf(Y, h{}ttp://www.University0.edu)
\end{Verbatim}

These queries were selected from the 14 test queries available on the
LUBM homepage. Answering \texttt{LQ$_1$} requires proper handling of
role subsumptions and inverses. \texttt{LQ$_2$} is interesting as it is
a complex conjunctive query. The performance results are summarised in
Table~\ref{table:lubm}. For DLog we used the \texttt{base} parameter
setting, i.e.\ we apply all optimisations. 

Loading \texttt{lubm1} took DLog 6.96 seconds. From this it took 5.29
seconds to read the DIG file and create the DL predicates. We
needed 1.47 seconds to generate the Prolog code. Finally, it took
0.18 seconds for SICStus Prolog to compile the generated
code. Answering \texttt{LQ$_1$} required only 0.26 seconds, while
\texttt{LQ$_2$} was answered instantaneously. 

Loading the larger \texttt{lubm} ontologies required much
more time, and the time needed for answering \texttt{LQ$_1$}  increased roughly
in proportion with the load time. However, the second query,
\texttt{LQ$_2$}, was executed instantaneously on all of the LUBM
ontologies.

{
\renewcommand{\arraystretch}{1.3}
\begin{table}[htb]
\centering
\caption{Aggregate results for the LUBM ontologies (times in seconds)}
\begin{center}\begin{tabular}{llrrrrrrrr}
\hline
& Query & \multicolumn{4}{c}{\texttt{LQ$_1$}} & \multicolumn{4}{c}{\texttt{LQ$_2$}}\\
\hline
& Testfile &lubm1&lubm2&lubm3&lubm4&lubm1&lubm2&lubm3&lubm4\\
\hline\hline
\multirow{3}{*}{\begin{sideways}DLog\end{sideways}} & load& 6.96 & 11.83 & 15.79 & 21.34 & 6.96 & 11.83 & 15.79 & 21.34\\
& runtime & \textbf{0.26} & \textbf{0.63} & \textbf{0.92} & \textbf{1.32} & \textbf{0.00} & \textbf{0.00} & \textbf{0.00}&\textbf{0.00} \\
& total & \textbf{7.22} & \textbf{12.46} & \textbf{16.71} & \textbf{22.66} & \textbf{6.96} & \textbf{11.83} & \textbf{15.79} & \textbf{21.34}\\
\hline
\multirow{3}{*}{\begin{sideways}KAON2\end{sideways}} & load& 6.56 & 13.56 & 20.66 & 28.73 & 6.56 & 13.56 & 20.66 & 28.73\\
& runtime & 0.70 & 0.99 & 1.33 & 1.69 & 0.66 & 0.93 & 1.27 & 1.62 \\
& total & 7.26 & 14.55 & 21.99 & 30.42 & 7.12 & 14.49 & 21.93 & 30.35\\
\hline
\multirow{3}{*}{\begin{sideways}RacerPro\end{sideways}} & load& 24.84 & 91.57 & X & X & 24.84 & 91.57 & X & X\\
& setup & 29.41 & 112.29 & X & X & 29.41 & 112.29 & X & X\\
& runtime & 2.69 & 5.89 & X & X & 4.07 & 7.49 & X & X\\
& total & 56.94 & 209.75 & X & X & 58.32 & 211.35 & X & X\\
\hline
\multirow{3}{*}{\begin{sideways}Pellet\end{sideways}} & load& 16.76 & - & - & - & 16.76 & - & - & -\\
& setup & 4.84 & - & - & - & 4.84 & - & - & -\\
& runtime & 27.09 & - & - & - & 27.19 & - & - & -\\
& total & 48.69 & - & - & - & 48.79 & - & - & -\\
\hline
\end{tabular}
\end{center}
\label{table:lubm}
\end{table}
}

Note that a significant part of the compile-time work for DLog is the
generation of the index predicates (cf.\ Section~\ref{indexing}). This
effectively doubles the number of the role assertions. The use of this
optimisation becomes unnecessary if we use a Prolog system with
multiple argument indexing or we store the ABox externally in a
database -- which is the preferred use of the DLog system. Also note
that the DIG input given to DLog is significantly larger
(cf.\ Table~\ref{table:lubm_prop}) than the OWL input the other
reasoning systems use.

KAON2 behaved very nicely on the LUBM ontologies: it was able to
answer both queries \texttt{LQ$_1$} and \texttt{LQ$_2$} on all
ontologies very quickly. We note that the official version of
KAON2 was  actually unable to solve the LUBM tests due to certain technical
problems. After contacting the author, the bugs causing this failure
were quickly fixed.

RacerPro managed to solve both queries on the ontologies \texttt{lubm1}
and \texttt{lubm2} with total times between 56.94 and 211.35
seconds. Here we can see the usual pattern: there is no real
difference between the execution times of \texttt{LQ$_1$} and
\texttt{LQ$_2$}. Unfortunately, on the bigger ontologies, RacerPro had
memory problems. 

Pellet solved the queries only on the smallest LUBM ontology. This
required 48.69 and 48.79 seconds. On the larger ontologies Pellet did
not signal memory problems, but simply ran out of the 10 minutes time
limit. 

Note that in the case of RacerPro and Pellet we also show the setup
time which is the time of the ABox consistency tests these systems
always perform at startup. We can see that RacerPro really spends most of
its time in this phase. On the other hand, Pellet spends fairly little
on consistency checking.

To sum up the results of the LUBM tests we can say that DLog and KAON2 were the
only systems able to solve both queries on all LUBM
ontologies. Of these two systems DLog emerges as the winner by a
small margin (although in terms of runtime DLog is much faster). It is
again worth noticing that, as in other cases, the execution times of DLog and
KAON2 are very good compared to those of the tableau-based reasoners.

\section{Future work}
\label{future}

In this section we give a brief overview of future work on the DLog system, for
improving its performance as well as extending its capabilities.

\paragraph{Partial evaluation}
Recall property (p2) in Definition~\ref{def:dlike}, which states that
each DL clause either contains a binary literal or it is ground, or it
contains no constants and exactly one variable. Note that the body of
the latter type of clauses is actually a conjunction of concept
goals. It is because of such clauses that the ancestor list can be
non-ground.

One can apply partial evaluation techniques, such as in
\cite{DBLP:conf/ecai/Venken84}, to unfold clauses containing no binary
literals. Such unfolding should be continued until each clause
contains either a binary literal or a unary literal corresponding to
an ABox predicate. Both such types of literals ensure that all their
arguments are ground upon exit. This means that we no longer need to
cater for executing unary predicates with uninstantiated arguments
(except for the outermost query predicate). Also, the ancestor list
becomes ground, which simplifies hashing. The absence of logic
variables in the data structures opens up the possibility of compiling
into Mercury code, rather than Prolog, which is expected to execute
much faster than standard Prolog. Some initial results on work in this
direction are reported in \cite{iclp2008}.

\paragraph{Tabling in the presence of ancestors}
It is often the case that the same goal is invoked several times during
query execution. Tabling \cite{xsb} can be used to prevent unnecessary
execution of such goals. Note, however, that unary goals in DLog have an
additional ancestor list argument. In most cases this additional argument
differs from call to call, making traditional tabling techniques
useless. Therefore it looks worthwhile to develop special  tabling methods
for DLog execution, which keep track of those ancestors that are actually
required for the successful completion of a given goal invocation.  This is
expected to improve the execution of queries on knowledge bases heavily
relying on ancestor resolution, such as the \texttt{Alcoholic} example of
Section~\ref{further}.

\paragraph{Relaxing the Unique Name Assumption}
Allowing different individual names to denote the same
individual is very important, as web-based reasoning requires exactly
this.  However, dismissing UNA has serious implications on the
transformation process. 

First, the definition of the DL program (Section~\ref{specpttp}) has to be
modified: we can no longer omit the contrapositives with an equality or an
inequality in the head. Such clauses will become parts of the two Prolog
predicates for inferring the equality and inequality of individuals. The
inequality predicate has to be further extended with some generic code, as
explained in Section~\ref{general}, which has to read the whole ABox. This,
however, goes against the main idea of the work presented here: focusing on
a small part of the ABox during query execution.

A possible compromise is to support a user-defined equality relation. This
would mean that the user can specify an equality relation for individual
names. The transitive-symmetric-reflexive closure of this relation is then
used as the equality, while its complement becomes the inequality
relation. In this case we can retain the transformation process, changing
only the code generated for the invocations of equality and inequality
relations. However, a user-defined equality can be inconsistent with the
rest of the knowledge base: e.g.\ while the user specifies that $i_1 =
i_2$, the ABox can contain assertions $C(i_1)$ and $\neg C(i_2)$. Therefore
this approach needs further investigation.

\paragraph{Other improvements}
As explained in Section~\ref{architecture}, presently we apply a simple
query ordering technique for execution of conjunctive queries. This can be
improved using the techniques of \cite{motik06PhD}. Furthermore, we
presently do not use statistical information in query ordering. Techniques
relying on statistical data are well researched in the context of
databases.  The use of such techniques in DLog should be investigated as
these can result in significant increase of execution performance.

The transformation scheme for role predicates, discussed in
Section~\ref{roles}, can be made more efficient by e.g.\ removing redundant
role axioms.

We also plan the extension of the external interfaces of DLog to
support new input formats, in addition to the DIG standard. We
presently have an experimental interface to support database
queries. Further work is needed to implement general interfaces to
database systems, including optimisations such as passing appropriate
conjunctive queries to database management systems, instead of single
queries.

\section{Summary and conclusions}
\label{conclusion}

In this paper we have presented the description logic reasoning system
DLog. Unlike the traditional tableau-based approach, DLog determines
the instances of a given $\mathcal{SHIQ}$ concept by transforming the
knowledge base into a Prolog program. This technique allows us to use
top-down query execution and to store the content of the ABox
externally in a database, something which is essential when large
amounts of data are involved.

We have compared DLog with the best available ABox reasoning
systems.  The test results show that DLog is significantly faster than traditional tableau-based
reasoning systems in all our benchmarks. In most of the cases DLog also outperforms KAON2,
which uses a similar resolution based approach as DLog.

We note that trends and behaviours of the various algorithms on
certain inputs can be more interesting than the actual runtimes (as
the latter can be very much affected by specific implementation
details). Considering also this, we argue that DLog and KAON2 are much
better suited for large data sets than tableau-based reasoners.

As an overall conclusion, we believe that our results are very
promising and clearly show that  description logic is an interesting
application field for Prolog and logic programming.

\section*{Acknowledgements}

We would like to thank all people who have contributed to the
development of the DLog system. Zsolt Nagy worked on Prolog based DL
reasoning with empty TBoxes. P{\'e}ter Boros{\'a}n created the first
version of the module responsible for transforming DIG input to DL
predicates. We are very grateful to Zsolt Zombori who extended this
module to support the $\mathcal{SHIQ}$ language and still works on the
DLog project. Zsolt also contributed to
Section~\ref{shiqprinciples} of this paper.

We are thankful to the developers of the existing reasoning systems
for their time and help during the testing phase. We are especially
grateful to Boris Motik for his valuable comments.

We would also like to thank Andr{\'a}s Gy{\"o}rgy B{\'e}k{\'e}s,
Tam{\'a}s Benk{\H o}, Zsolt Nagy and Zsolt Zombori for their
insightful comments on draft versions of this paper.

Finally, we would like to thank the anonymous reviewers whose valuable
suggestions were of great help in  improving the quality of the paper.

\bibliography{dlog}

\appendix
{
\renewcommand{\arraystretch}{1.08}
\begin{table}[htb]
\caption{Properties of the test files (times in seconds)}
\begin{center}\begin{tabular}{llrrrrrrrrr}
\hline
& Testfile &c10&c20&c100&c1000&c10000&n1&n2&n3&n4\\
\hline\hline
& size(MB) &0.00&0.00&0.02&0.19&1.88&0.01&0.06&0.35&2.82\\
& TBox &1&1&1&1&1&1&1&1&1\\
& ABox &22&42&202&2002&20002&100&646&3897&30797\\
\hline\hline
& load(sec) &0.04&0.06&0.13&0.23&0.78&0.12&0.21&0.23&0.88\\
\hline\hline
\multirow{3}{*}{\begin{sideways}\texttt{base}\end{sideways}} & translate &0.00&0.01&0.01&0.08&0.67&0.01&0.02&0.13&1.10\\
& compile &0.03&0.01&0.01&0.02&0.02&0.01&0.01&0.02&0.01\\
& \textbf{total} &\textbf{0.07}&\textbf{0.08}&\textbf{0.15}&\textbf{0.33}&\textbf{1.47}&\textbf{0.14}&\textbf{0.24}&\textbf{0.38}&\textbf{1.99}\\
\hline
\multirow{3}{*}{\begin{sideways}\texttt{g(n)}\end{sideways}} & translate &0.01&0.00&0.01&0.08&0.69&0.00&0.04&0.18&1.04\\
& compile &0.00&0.02&0.02&0.01&0.01&0.02&0.02&0.02&0.02\\
& \textbf{total} &\textbf{0.05}&\textbf{0.08}&\textbf{0.16}&\textbf{0.32}&\textbf{1.48}&\textbf{0.14}&\textbf{0.27}&\textbf{0.43}&\textbf{1.94}\\
\hline
\multirow{3}{*}{\begin{sideways}\texttt{p(n)}\end{sideways}} & translate &0.01&0.01&0.02&0.08&0.71&0.01&0.02&0.15&1.29\\
& compile &0.01&0.01&0.01&0.02&0.01&0.01&0.02&0.01&0.02\\
& \textbf{total} &\textbf{0.06}&\textbf{0.08}&\textbf{0.16}&\textbf{0.33}&\textbf{1.50}&\textbf{0.14}&\textbf{0.25}&\textbf{0.39}&\textbf{2.19}\\
\hline
\multirow{3}{*}{\begin{sideways}\texttt{f(n)}\end{sideways}} & translate &0.01&0.01&0.01&0.07&0.72&0.01&0.02&0.16&1.05\\
& compile &0.01&0.01&0.01&0.01&0.02&0.01&0.02&0.01&0.03\\
& \textbf{total} &\textbf{0.06}&\textbf{0.08}&\textbf{0.15}&\textbf{0.31}&\textbf{1.52}&\textbf{0.14}&\textbf{0.25}&\textbf{0.40}&\textbf{1.96}\\
\hline
\multirow{3}{*}{\begin{sideways}\texttt{i(n)}\end{sideways}} & translate &0.00&0.01&0.01&0.02&0.30&0.01&0.04&0.06&0.52\\
& compile &0.01&0.01&0.01&0.02&0.02&0.01&0.01&0.00&0.01\\
& \textbf{total} &\textbf{0.05}&\textbf{0.08}&\textbf{0.15}&\textbf{0.27}&\textbf{1.10}&\textbf{0.14}&\textbf{0.26}&\textbf{0.29}&\textbf{1.41}\\
\hline
\multirow{3}{*}{\begin{sideways}\texttt{o(n)}\end{sideways}} & translate &0.00&0.00&0.01&0.08&0.70&0.01&0.04&0.11&1.10\\
& compile &0.01&0.01&0.01&0.00&0.03&0.01&0.01&0.02&0.01\\
& \textbf{total} &\textbf{0.05}&\textbf{0.07}&\textbf{0.15}&\textbf{0.31}&\textbf{1.51}&\textbf{0.14}&\textbf{0.26}&\textbf{0.36}&\textbf{1.99}\\
\hline
\multirow{3}{*}{\begin{sideways}\texttt{d(n)}\end{sideways}} & translate &0.00&0.00&0.01&0.08&0.71&0.01&0.02&0.15&1.05\\
& compile &0.00&0.01&0.01&0.01&0.01&0.01&0.02&0.01&0.02\\
& \textbf{total} &\textbf{0.04}&\textbf{0.07}&\textbf{0.15}&\textbf{0.32}&\textbf{1.50}&\textbf{0.14}&\textbf{0.25}&\textbf{0.39}&\textbf{1.95}\\
\hline
\multirow{3}{*}{\begin{sideways}\texttt{pd(n)}\end{sideways}} & translate &0.01&0.02&0.02&0.07&0.71&0.01&0.02&0.12&1.06\\
& compile &0.01&0.01&0.02&0.02&0.02&0.01&0.01&0.02&0.01\\
& \textbf{total} &\textbf{0.06}&\textbf{0.09}&\textbf{0.17}&\textbf{0.32}&\textbf{1.51}&\textbf{0.14}&\textbf{0.24}&\textbf{0.37}&\textbf{1.95}\\
\hline
\multirow{3}{*}{\begin{sideways}\texttt{od(n)}\end{sideways}} & translate &0.01&0.01&0.01&0.08&0.73&0.02&0.03&0.12&1.04\\
& compile &0.01&0.01&0.01&0.02&0.02&0.02&0.01&0.01&0.04\\
& \textbf{total} &\textbf{0.06}&\textbf{0.08}&\textbf{0.15}&\textbf{0.33}&\textbf{1.53}&\textbf{0.16}&\textbf{0.25}&\textbf{0.36}&\textbf{1.96}\\
\hline
\end{tabular}
\end{center}
\label{table:testfilesprop}
\end{table}
}

{
\renewcommand{\arraystretch}{0.95}
\begin{table}[htb]
\caption{Dlog results for the Iocaste ontologies (times in seconds)}
\begin{center}\begin{tabular}{llrrrrrrrrr}
\hline
& Testfile &c10&c20&c100&c1000&c10000&n1&n2&n3&n4\\
\hline\hline
\multirow{4}{*}{\begin{sideways}\texttt{base}\end{sideways}} & loop &0&0&0&0&0&0&2&6&0\\
& ancres &0&0&0&0&0&0&0&0&0\\
& orphanc &18&38&198&1998&19998&81&130&448&3197\\
& \textbf{runtime} &\textbf{0.00}&\textbf{0.00}&\textbf{0.00}&\textbf{0.01}&\textbf{0.11}&\textbf{0.00}&\textbf{0.00}&\textbf{0.00}&\textbf{0.02}\\
\hline
\multirow{4}{*}{\begin{sideways}\texttt{g(n)}\end{sideways}} & loop &0&0&0&0&0&0&2&6&0\\
& ancres &0&0&0&0&0&0&0&0&0\\
& orphanc &18&38&198&1998&19998&81&130&448&3197\\
& \textbf{runtime} &\textbf{0.00}&\textbf{0.00}&\textbf{0.00}&\textbf{0.01}&\textbf{0.12}&\textbf{0.00}&\textbf{0.00}&\textbf{0.00}&\textbf{0.02}\\
\hline
\multirow{4}{*}{\begin{sideways}\texttt{p(n)}\end{sideways}} & loop &0&0&0&0&0&0&2&6&0\\
& ancres &0&0&0&0&0&0&0&0&0\\
& orphanc &99&399&9999&10$^6$&10$^8$&81&130&448&3197\\
& \textbf{runtime} &\textbf{0.00}&\textbf{0.00}&\textbf{0.04}&\textbf{4.14}&\textbf{500.42}&\textbf{0.00}&\textbf{0.00}&\textbf{0.00}&\textbf{0.01}\\
\hline
\multirow{4}{*}{\begin{sideways}\texttt{f(n)}\end{sideways}} & loop &0&0&0&0&0&0&2&6&0\\
& ancres &0&0&0&0&0&0&0&0&0\\
& orphanc &18&38&198&1998&19998&81&130&448&3197\\
& \textbf{runtime} &\textbf{0.00}&\textbf{0.00}&\textbf{0.00}&\textbf{0.01}&\textbf{0.11}&\textbf{0.00}&\textbf{0.00}&\textbf{0.00}&\textbf{0.02}\\
\hline
\multirow{4}{*}{\begin{sideways}\texttt{i(n)}\end{sideways}} & loop &0&0&0&0&0&0&2&6&0\\
& ancres &0&0&0&0&0&0&0&0&0\\
& orphanc &18&38&198&1998&19998&81&130&448&3197\\
& \textbf{runtime} &\textbf{0.00}&\textbf{0.00}&\textbf{0.00}&\textbf{0.10}&\textbf{9.58}&\textbf{0.00}&\textbf{0.00}&\textbf{0.00}&\textbf{0.02}\\
\hline
\multirow{4}{*}{\begin{sideways}\texttt{o(n)}\end{sideways}} & loop &0&0&0&0&0&0&2&6&0\\
& ancres &0&0&0&0&0&0&0&0&0\\
& orphanc &9&19&99&999&9999&45&47&53&46\\
& \textbf{runtime} &\textbf{0.00}&\textbf{0.00}&\textbf{0.00}&\textbf{0.01}&\textbf{0.12}&\textbf{0.00}&\textbf{0.00}&\textbf{0.00}&\textbf{0.02}\\
\hline
\multirow{4}{*}{\begin{sideways}\texttt{d(n)}\end{sideways}} & loop &0&0&0&0&0&0&2&9&0\\
& ancres &0&0&0&0&0&0&0&0&0\\
& orphanc &18&38&198&1998&19998&81&130&445&3197\\
& \textbf{runtime} &\textbf{0.00}&\textbf{0.00}&\textbf{0.00}&\textbf{0.01}&\textbf{0.13}&\textbf{0.00}&\textbf{0.00}&\textbf{0.00}&\textbf{0.02}\\
\hline
\multirow{4}{*}{\begin{sideways}\texttt{pd(n)}\end{sideways}} & loop &0&0&0&0&0&0&2&9&0\\
& ancres &0&0&0&0&0&0&0&0&0\\
& orphanc &99&399&9999&10$^6$&10$^8$&81&130&445&3197\\
& \textbf{runtime} &\textbf{0.00}&\textbf{0.00}&\textbf{0.04}&\textbf{4.15}&\textbf{502.98}&\textbf{0.00}&\textbf{0.00}&\textbf{0.00}&\textbf{0.02}\\
\hline
\multirow{4}{*}{\begin{sideways}\texttt{od(n)}\end{sideways}} & loop &0&0&-&-&-&0&2&43&0\\
& ancres &0&0&-&-&-&0&0&0&0\\
& orphanc &256&2302&-&-&-&0&0&0&0\\
& \textbf{runtime} &\textbf{0.00}&\textbf{3.56}&-&-&-&\textbf{0.01}&\textbf{0.03}&\textbf{0.18}&\textbf{4.11}\\
\hline
\end{tabular}
\end{center}
\label{table:result}
\end{table}
}
\clearpage

\end{document}